\documentclass[journal,12pt,onecolumn,draftclsnofoot]{IEEEtran}

\usepackage{amsmath}
\usepackage{graphicx,psfrag} 
\usepackage{amsmath,amsthm,amsfonts,amssymb,bm} 
\usepackage[ruled,vlined]{algorithm2e}
\usepackage{algpseudocode}
\usepackage{xcolor}
\usepackage{bbm}             
\usepackage{tabularx}        
\usepackage{bbm}
\usepackage{setspace}
\usepackage[latin1]{inputenc}
\allowdisplaybreaks[3]

\newcommand{\bs}{\mathbf}

\newtheorem{proposition}{Proposition}

\newtheorem{definition}{Definition}
\newtheorem{assumption}{Assumption}
\newtheorem{remark}{Remark}

\begin{document}

\title{
Multi-Layer Bilinear Generalized Approximate Message Passing
}

\author{
Qiuyun Zou, Haochuan Zhang, and Hongwen Yang

\thanks{
Q. Zou and H. Yang are with Beijing University of Posts and Telecommunications,
Beijing 100876, China (email: qiuyun.zou@bupt.edu.cn; yanghong@bupt.edu.cn).  \ (\textit{Corresponding author: H. Yang.})
}
\thanks
{
H. Zhang is with School of Automation, and with Research Institute of Integrated Circuit Innovation, both in Guangdong University of Technology, Guangzhou 510006, China (email: haochuan.zhang@qq.com).
}

\thanks{

}
}

\maketitle

\begin{abstract}
In this paper, we extend the bilinear generalized approximate message passing (BiG-AMP) approach, originally proposed for high-dimensional generalized bilinear regression, to the multi-layer case for the handling of cascaded problem such as matrix-factorization problem arising in relay communication among others.
Assuming statistically independent matrix entries with known priors, the new algorithm called ML-BiGAMP could approximate the general sum-product loopy belief propagation (LBP) in the high-dimensional limit enjoying a substantial reduction in computational complexity.
We demonstrate that, in large system limit, the asymptotic MSE performance of ML-BiGAMP could be fully characterized via a set of simple one-dimensional equations termed state evolution (SE). We establish that the asymptotic MSE predicted by ML-BiGAMP' SE matches perfectly the exact MMSE predicted by the replica method, which is well-known to be Bayes-optimal but infeasible in practice. This consistency indicates that the ML-BiGAMP may still retain the same Bayes-optimal performance as the MMSE estimator in high-dimensional applications, although ML-BiGAMP's computational burden is far lower. As an illustrative example of the general ML-BiGAMP, we provide a detector design that could estimate the channel fading and the data symbols jointly with high precision for the two-hop amplify-and-forward relay communication systems.

\end{abstract}


\section{Introduction}
In the context of matrix completion \cite{Kabashima2016Phase}, robust principal component analysis \cite{candes2011robust},  dictionary learning \cite{kreutz2003dictionary,tosic2011dictionary}, and representation learning \cite{bengio2013representation},  the matrix factorization problem could be formalized as the following generalized bilinear regression problem: the signal recovery of $\bs{H}$ and $\bs{X}$ from $\bs{Y}=\boldsymbol{\phi}(\bs{Z},\bs{W})$ with $\bs{Z}=\bs{HX}$ and $\mathcal{P}(\bs{Y}|\bs{Z})=\int \delta(\bs{Y}-\boldsymbol{\phi}(\bs{Z},\bs{W}))\mathcal{P}(\bs{W})\text{d}\bs{W}$, where $\bs{Y}$ is observed from  $\bs{Z}$ and noise $\bs{W}$ through a deterministic and element-wise mapping $\boldsymbol{\phi}(\cdot)$, and $\bs{H}$ and $\bs{X}$ are matrices to be factorized. To solve this inference problem, Parker \textit{et al} proposed bilinear generalized approximate message passing (BiG-AMP) \cite{parker2014bilinear} algorithm, which achieved the Bayes-optimal error in large system setting with affordable computational complexity. Inspired by this seminal work, we consider in this paper an even more ambitious problem, i.e., multi-layer generalized bilinear regression.
\begin{figure}[!t]
\centering
\includegraphics[width=0.6\textwidth]{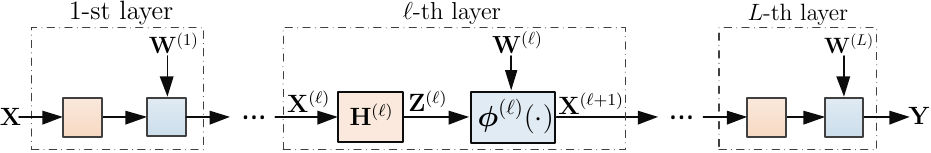}
\caption{The multi-layer generalized bilinear inference problem is to estimate the input signal $\bs{X}^{(\ell)}$ and measurement matrix $\bs{H}^{(\ell)}$ of each layer from the observation $\bs{Y}$.
}
\label{Fig:System}
\end{figure}
The multi-layer generalized bilinear model\footnote{Note that in \cite{pandit2020inference}, each layer of the model (\ref{alg:system}) was divided into two layer: odd-indexed layer (linear mixing space) and even-indexed layer (element-wise mapping).} can be described as
\begin{align}
\bs{X}^{(\ell+1)}=\boldsymbol{\phi}^{(\ell)}\left(\bs{H}^{(\ell)}\bs{X}^{(\ell)},\bs{W}^{(\ell)}\right),\quad \ell=1,\cdots,L,
\label{alg:system}
\end{align}
where $\bs{X}=\bs{X}^{(1)}$ is the input of the network, $\{\bs{X}^{(\ell)}\}_{\ell=2}^L$ are hidden layer signals, and $\bs{Y}=\bs{X}^{(L+1)}$ is the observation. In addition, $\bs{Z}^{(\ell)}\in \mathbb{R}^{N_{\ell+1}\times K}$ is obtained from $\bs{X}^{(\ell)}\in \mathbb{R}^{N_{\ell}\times K}$ going through a linear mixing defined by $\bs{Z}^{(\ell)}=\bs{H}^{(\ell)}\bs{X}^{(\ell)}$, while $\bs{X}^{(\ell+1)}$ is further generated from $\bs{Z}^{(\ell)}$ and random variable $\bs{W}^{(\ell)}$, whose probability distribution is $\mathcal{P}(\bs{W}^{(\ell)})$, using a deterministic and element-wise function $\boldsymbol{\phi}^{(\ell)}(\cdot)$.

The multi-layer generalized bilinear inference problem (\ref{alg:system}) arises in many contexts, such as, deep generative prior \cite{pandit2020inference,pandit2020inference2,yeh2017semantic,jalali2019solving}, massive multiple-input multiple-output (MIMO) relay system \cite{yang2020bayes,wen2010asymptotic}, and machine learning \cite{emami2020generalization,gabrie2018entropy}, where the correlations between sets of variables in different subsystems involve multiple layers of interdependencies. To address this issue, \cite{manoel2017multi,zou2020estimation} extended approximate message passing (AMP) \cite{bayati2011dynamics,donoho2009message} to provide inference algorithms for multi-layer region.
The  AMP algorithm, an approximation to sum-product loopy belief propagation (LBP), was firstly proposed for sparse signal reconstruction in standard linear inverse inference. The AMP's mean square error (MSE) performance could be predicted by a scalar formula called state evolution (SE) under the assumption of \textit{i.i.d.} sub-Gaussian random matrix regimes.  Further, it was shown that the AMP's SE matched perfectly the fixed point of the minimum mean square error (MMSE) estimator derived by replica method \cite{guo2005randomly}. In addition, the AMP algorithm is closely related to the celebrated iterative soft thresholding (IST) algorithm \cite{daubechies2004iterative}, in which the only difference is the \textit{Onsager term}.
Another algorithm for multi-layer inference refers to multi-layer vector AMP (ML-VAMP) \cite{pandit2020inference}, which extended the VAMP algorithm to cover the multi-layer case. Recently, it has been proven that VAMP and AMP have identical fixed points in their state evolutions \cite{zhang2020identical}. The VAMP algorithm holds under a much broader class of large random matrices (right-orthogonally invariant) than AMP algorithm but has higher computational complexity for their overlapping regions due to the singular value decomposition (SVD) operation, which is very close to expectation propagation (EP) \cite{minka2001family}, expectation consistent (EC) \cite{opper2005expectation,he2017generalized}, and orthogonal approximate message passing (OAMP) \cite{ma2017orthogonal}.  For the case of $K>1$, \cite{pandit2020inference} extended the ML-VAMP algorithm to the matrix case, called ``ML-Mat-VAMP''. Similar to AMP-like algorithms, the asymptotic MSE performance of ML-Mat-VAMP could be predicted in a certain random large system limits. However, the ML-Mat-VAMP algorithm is costly in computation due to the SVD operation.

To handle the multi-layer generalized bilinear inference problem, in the present work, we extend the celebrated bilinear generalized AMP (BiG-AMP) algorithm \cite{parker2014bilinear} to multi-layer case and propose the  multi-layer bilinear generalized approximate message passing (ML-BiGAMP). The ML-BiGAMP algorithm solves the vector-valued estimation problem into a sequence of scalar problems and linear transforms, and is thus low-complexity, which is an approximation of the sum-product LBP by performing Gaussian approximation and Taylor expansion. Similar to other AMP-like algorithms, by performing large system analysis, we give SE analysis of the ML-BiGAMP algorithm, which exactly predicts the asymptotic MSE performance of ML-BiGAMP when the latter should be run for a sufficiently large number of iterations. In addition, we apply replica method\footnote{ Although replica method is known as a non-rigorous tool, this method is widely believed to be exact in the context of theoretical statistical physics \cite{Kabashima2016Phase}. Recently,  several literatures have proven that the replica prediction is correct in the case of \textit{i.i.d.} Gaussian matrices (e.g., \cite{reeves2016replica}).} derived from statistic physics \cite{mezard1987spin} to analyze the achievable MSE performance of the exact MMSE estimator for multi-layer generalized bilinear inference problem. Indeed, a first cross-check of the correctness of our results is the fact that the asymptotic MSE predicted by ML-BiGAMP'SE agrees precisely with the exact MMSE as predicted by replica method in certain random large system limit. The main contributions of this work are summarized as follows:
\begin{itemize}
\item
We propose a computationally efficient iterative algorithm, \textit{multi-layer bilinear generalized approximate message passing} or \textit{ML-BiGAMP}, for estimating $\{\bs{X}^{(\ell)}\}_{\ell=1}^L$ and $\{\bs{H}^{(\ell)}\}_{\ell=1}^L$ from the network output $\bs{Y}$ of the form in (\ref{alg:system}).
\item
Under the \textit{i.i.d.} Gaussian measurement matrices, we show that the asymptotic MSE performance of the ML-BiGAMP algorithm could be fully characterized by a set of one-dimensional iterative equations termed state evolution.
\item
We establish that the asymptotic MSE predicted by ML-BiGAMP'SE matches perfectly the exact MMSE predicted by the replica method, which is well known to be Bayes-optimal but infeasible in practice. The fixed point equations of the exact MMSE estimator further reveal the decouple principle, that is, in large system limit, the input output relationship of the model (\ref{alg:system}) is decoupled into a bank of scalar additive white Gaussian noise (AWGN) channels w.r.t. the input signal $\bs{X}$ and measurement matrices $\{\bs{H}^{(\ell)}\}_{\ell=1}^L$.
\item
Based on the proposed algorithm, we develop a joint channel and data (JCD) estimation method for massive amplify-and-forward (AF) relay communication, where the estimated payload data are utilized to aid the channel estimation. The simulation results confirm that our JCD method improves the performance of the pilot-only method, and validate the consistency of MSE performance of ML-BiGAMP and its SE.
\end{itemize}

%
%
%

The remainder of this work is organized as follows. Section II presents several examples of the multi-layer generalized bilinear inference problem (\ref{alg:system}). In Section III, we introduce the proposed ML-BiGAMP algorithm. In Section IV, we give the SE analysis of the ML-BiGAMP algorithm. In Section V, we apply the replica method to analyze the asymptotic MSE performance of the exact MMSE estimator. Finally, Section VI gives numeric simulations to validate the accuracy of these theoretic results.

\textit{Notations:} $\bs{A}$ denotes a matrix with $a_{ij}$ being its $(i,j)$-th element. $\|\bs{A}\|_{\text{F}}$ denotes the Frobenius norm. $\mathcal{N}(x|a,A)$ denotes a Gaussian distribution with mean $a$ and variance $A$:
$$
\mathcal{N}(x|a,A)=\frac{1}{\sqrt{2\pi A}}\exp \left[-\frac{(x-a)^2}{2A}\right].$$
$\mathcal{N}_{x|z}^{(\ell)}(a,A,b,B)=\mathcal{P}(x^{(\ell+1)}|z^{(\ell)})\mathcal{N}(z^{(\ell)}|a,A)\mathcal{N}(x^{(\ell+1)}|b,B)$, where $\mathcal{P}(x^{(\ell+1)}|z^{(\ell)})$ is the  transition distribution from $z^{(\ell)}$ to $x^{(\ell+1)}$.
$\text{D}\xi$ denotes Gaussian measure i.e., $\text{D}\xi=\mathcal{N}(\xi|0,1)\text{d}\xi$.

\section{Examples of Multi-Layer Generalized Bilinear Regression}
\label{Sec:Example}
For the model in (\ref{alg:system}), it is assumed that the transition distribution of each layer is componentwise, which is given by
\begin{align}
\nonumber
\mathcal{P}(\bs{X}^{(\ell+1)}|\bs{Z}^{(\ell)})
&=\int \delta\left(\bs{X}^{(\ell+1)}-\boldsymbol{\phi}^{(\ell)}(\bs{Z}^{(\ell)},\bs{W}^{(\ell)})\right)\mathcal{P}(\bs{W}^{(\ell)})\text{d}\bs{W}^{(\ell)},
\label{Equ:Transition}
\end{align}
where $\delta(\cdot)$ denotes Dirac delta function. Additionally, the componentwise mapping means $\mathcal{P}(\bs{X}^{(\ell+1)}|\bs{Z}^{(\ell)})=\prod_{m=1}^{N_{\ell+1}}\prod_{k=1}^{K}\mathcal{P}\left(x_{mk}^{(\ell+1)}|z_{mk}^{(\ell)}\right)$. The multi-layer generalized bilinear inference problem is to estimate the input signals $\{\bs{X}^{(\ell)}\}_{\ell=1}^L$ and measurement matrices $\{\bs{H}^{(\ell)}\}_{\ell=1}^L$ from the output $\bs{Y}$ of the model. In doing so, it is assumed that $\bs{X}$ and $\bs{H}^{(\ell)}$ are composed of random variables $\textsf{X}$ and $\textsf{H}^{(\ell)}$, respectively, which are drawn from the known distributions $\mathcal{P}(x)$ and $\mathcal{P}(h^{(\ell)})$, i.e.,
\begin{align}
\mathcal{P}(\bs{X})&=\prod_{n=1}^{N_1}\prod_{k=1}^K \mathcal{P}(x_{nk}),\\
\mathcal{P}(\bs{H}^{(\ell)})&=\prod_{m=1}^{N_{\ell+1}}\prod_{n=1}^{N_{\ell}}\mathcal{P}(h_{mn}^{(\ell)}).
\end{align}
We consider the \textit{large system limit}, in which the dimensions of the system go into infinity, i.e., $\forall \ell, N_{\ell},K\rightarrow \infty$ but the ratios $\alpha=\frac{N_1}{K}$ and $\beta_{\ell}=\frac{N_{\ell+1}}{N_{\ell}}$ are fixed and bounded. Actually, the model in (\ref{alg:system}) is a general model with many important problems as its special cases. We give a brief review in the following.

\subsection{Single-Layer Inference Problem}
When $L=1$, the multi-layer inference problem (\ref{alg:system}) reduces to a matrix factorization problem or generalized bilinear inverse problem, in which the target is to estimate the signal of interest $\bs{X}$ and the measurement matrix $\bs{H}$ from the observation $\bs{Y}$:
\begin{align}
\bs{Y}=\boldsymbol{\phi}(\bs{Z},\bs{W}), \quad \text{s.t.} \quad \bs{Z}=\bs{HX}.
\label{Equ:Single_System}
\end{align}
This degenerated model has a wide range of applications. One example is the joint channel and user data estimation \cite{wen2015bayes,zou2020low} considering a quantized massive MIMO communication system, in which the function $\boldsymbol{\phi}(\cdot)$ is particularized as $\bs{Y}=\textsf{Q}(\bs{HX}+\bs{W})$ with $\textsf{Q}(\cdot)$ being an uniform quantizer. More applications could be found in dictionary learning, blind matrix calibration, sparse principal component analysis (PCA) and blind source separation \cite{Kabashima2016Phase}. It is worthy of noting that when the function $\boldsymbol{\phi}(\cdot)$ is particularized as a linear function, i.e., $\bs{Y}=\bs{HX}+\bs{W}$, and the measure matrix is already known, the model is degenerated to multiple measurement vector (MMV) problem, which has been widely applied in compressed sensing \cite{kim2011belief,ziniel2012efficient,haghighatshoar2018multiple}, user activity detection in communication \cite{liu2018massive,liu2019generalized}, and direction of arrival (DOA) estimation \cite{tzagkarakis2010multiple}.


\subsection{Multi-Hop Relay Communication}
\begin{figure}[!t]
\centering
\includegraphics[width=0.55\textwidth]{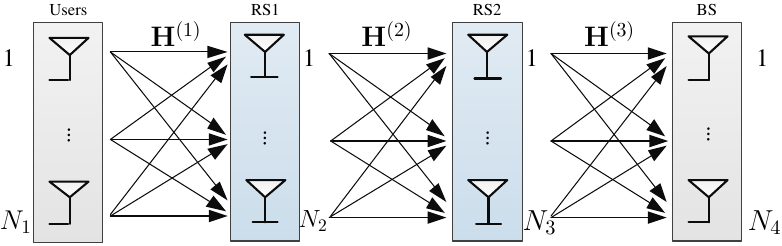}
\caption{Massive MIMO AF relay system.
}
\label{Fig:Relay_System}
\end{figure}
The multi-layer inference problem (\ref{alg:system}) can also be applied to multi-hop massive MIMO amplify-and-forward (AF) relay system \cite{yang2020bayes,wen2010asymptotic}, which has been regarded as an attractive solution to improve the quality of wireless communication. Fig.~\ref{Fig:Relay_System} shows a special case of multi-hop massive MIMO AF relay system in $L=3$. The multi-hop massive MIMO AF relay system can be modeled as
\begin{align}
\begin{cases}
\bs{X}^{(2)}=\textsf{Q}_{\textsf{c}}\left(\bs{H}^{(1)}\bs{X}^{(1)}+\bs{W}^{(1)}\right)\\
\bs{X}^{(3)}=\textsf{Q}_{\textsf{c}}\left(\rho^{(2)}\bs{H}^{(2)}\bs{X}^{(2)}+\bs{W}^{(2)}\right)\\
\quad \qquad \ \ \vdots\\
\ \ \ \bs{Y}=\textsf{Q}_{\textsf{c}}\left(\varrho^{(L)}\bs{H}^{(L)}\bs{X}^{(L)}+\bs{W}^{(L)}\right)
\end{cases},
\end{align}
where the matrices $\bs{H}^{(1)}$, $\{\bs{H}^{(\ell)}\}_{2}^{L-1}$, and $\bs{H}^{(L)}$ denote the channels from users to $1$st relay station (RS), $(\ell-1)$-th RS to $\ell$-th RS, and $(L-1)$-th RS to BS, respectively. $\{\bs{W}^{(\ell)}\}_{\ell=1}^L$ are the corresponding additive white Gaussian noises (AWGNs). $\{\varrho^{(\ell)}\}_{\ell=2}^L$ are amplification coefficient. $\textsf{Q}_{\textsf{c}}(\cdot)$ refers to a complex-valued quantizer including two separate real-valued quantizer $\textsf{Q}(\cdot)$. In \cite{yang2020bayes}, the authors considered a two-hop massive MIMO AF relay system with perfect channel information and  developed a EC based method to estimate the user data, which can be regarded as a special case of ML-VAMP in $L=2$.

\subsection{RIS-Aided Massive MIMO System}
A reconfigurable intelligent surfaces (RIS)-aided massive MIMO system \cite{he2019cascaded} is presented in Fig.~\ref{Fig:RIS_system}, where RIS includes $N_2$ low-cost passive elements and the BS is equipped  with $N_1$ antennas. Each user is equipped with $N_3$ antennas. In a coherent block (block length $K$), the received signal of the reference user can be expressed as\footnote{Here, we consider that the RIS not only reflects the signal, but also reflects nearby stray electromagnetic signal.}
\begin{align}
\bs{Y}=\bs{H}^{(2)}(\bs{S}\odot (\bs{H}^{(1)}\bs{X}^{(1)})+\bs{W}^{(1)})+\bs{W}^{(2)},
\end{align}
where $\odot$ represents  componentwise vector multiplication, $\bs{X}^{(1)}\in \mathbb{R}^{N_1\times K}$ is the transmitted signal, and $(\bs{W}^{(1)}, \bs{W}^{(2)})$ are additive noise with power $\sigma_w^2$.  $\bs{H}^{(1)}\in \mathbb{R}^{N_2\times N_1}$ and $\bs{H}^{(2)}\in \mathbb{R}^{N_3\times N_2}$ are the channels from BS to RIS and RIS to user, respectively. In addition, $\bs{S}$ is phase shift matrix and is known beforehand. Such system corresponds to the MMSE estimation of multi-layer generalized bilinear model in $L=2$. By defining $\bs{Z}^{(1)}=\bs{H}^{(1)}\bs{X}^{(1)}$ and $\bs{Z}^{(2)}=\bs{H}^{(2)}\bs{X}^{(2)}$, the transition distributions of the two layers are given by $\mathcal{P}(x^{(2)}_{mk}|z^{(1)}_{mk})=\mathcal{N}(x^{(2)}_{mk}|s_{mk}z^{(1)}_{mk},\sigma_w^2)$ and $\mathcal{P}(y_{pk}|z^{(2)}_{pk})=\mathcal{N}(y_{pk}|z^{(2)}_{pk},\sigma_w^2)$.  As the RIS only reflects the signal, then the model degenerates $\bs{Y}=\bs{H}^{(2)}(\bs{S}\odot (\bs{H}^{(1)}\bs{X}^{(1)}))+\bs{W}^{(2)}$. Accordingly, the transition distribution becomes
$\mathcal{P}(x^{(2)}_{mk}|z^{(1)}_{mk})=\delta(x^{(2)}_{mk}-s_{mk}z^{(1)}_{mk})$. Indeed, the Dirac delta function $\delta(x)$ can be regarded as the limit of standard Gaussian: $\lim_{v\rightarrow 0}(2\pi v)^{-\frac{1}{2}}\exp(-\frac{x^2}{2v})$, which is useful in realization.\\
\begin{figure}[!t]
\centering
\includegraphics[width=0.45\textwidth]{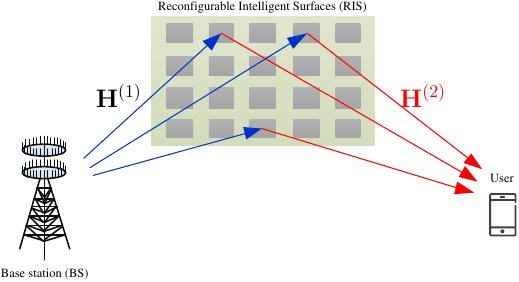}
\caption{A RIS-aided massive MIMO system.
}
\label{Fig:RIS_system}
\end{figure}

\subsection{Compressive Matrix Completion}
\label{Sec:CMC}
In matrix completion (MC) \cite{parker2014bilinearII}, only a fraction of entries of observation are valid. In other words, the observation in MC problem is generally sparse. To reduce the memory, we here consider a more practical scheme: compressive matrix completion, i.e., MC + \textit{compressive sampling}. The problem of compressive MC can be modeled as
\begin{align}
\begin{cases}
\bs{X}^{(2)}=\boldsymbol{f}(\bs{H}^{(1)}\bs{X}^{(1)}+\bs{W}^{(1)})\\
\ \ \  \bs{Y}=\bs{H}^{(2)}\bs{X}^{(2)}+\bs{W}^{(2)}
\end{cases},
\end{align}
where $\bs{H}^{(2)}\in \mathbb{R}^{N_3\times N_2}$ ($N_3\ll N_2$), and $\boldsymbol{f}$ is a componentwise mapping which is specified by
\begin{align}
\mathcal{P}(x^{(2)}_{mk}|z_{mk}^{(1)})=
\begin{cases}
\mathcal{N}(x^{(2)}_{mk}|z_{mk}^{(1)},\sigma_w^2)  & (m,k)\in \Omega\\
\mathbbm{1}_{y}  &(m,k)\notin \Omega
\end{cases},
\label{Equ:MC}
\end{align}
where $\Omega$ is a subset of valid entries of $\bs{X}^{(2)}$ and $\mathbbm{1}_y$ denotes a point mass at $y=0$. The goal is to recover a rank $N_1\ll \min (N_2,K)$ matrix $\bs{Z}^{(1)}=\bs{H}^{(1)}\bs{X}^{(1)}\in \mathbb{R}^{N_2\times K}$ from the observation $\bs{Y}\in \mathbb{R}^{N_3\times K}$.
For convenience, we here consider the rank $N_1$ is given. However, for the case of unknown $N_1$, similar to BiG-AMP \cite{parker2014bilinearII}, the proposed ML-BiGAMP can also be combined with rank selection method. The penalized log-likelihood is given by
\begin{align}
\hat{N}_1=\underset{N_1=1,\cdots, N_{\text{max}}}{\arg\max}\ 2\log p(\hat{\bs{X}}_{N_1}^{(2)}|\hat{\bs{Z}}_{N_1}^{(1)})-\eta(N_1),
\end{align}
where subscript ${N_1}$ indicates the restriction to rank $N_1$, $\hat{\bs{X}}_{N_1}^{(2)}$ and $\hat{\bs{Z}}_{N_1}^{(1)}=\hat{\bs{H}}_{N_1}^{(1)}\hat{\bs{X}}_{N_1}^{(1)}$ are provided by ML-BiGAMP, and $\eta(N_1)$ is penalty function (see \cite{parker2014bilinearII}). Specially, when compressive sampling is not considered, such compressive MC degenerates the classical MC. Note that compared to compressive sampling \cite{donoho2009message}, the prior of $\bs{X}^{(2)}$ is unknown in our compressive MC.


\section{ML-BiGAMP}
\subsection{Problem Formulation}
Considering the multi-layer generalized bilinear inference problem (\ref{alg:system}), all the input signals $\{\bs{X}^{(\ell)}\}_{\ell=1}^L$ and measurement matrices $\{\bs{H}^{(\ell)}\}_{\ell=1}^L$ of each layer should be estimated with the known distributions $\mathcal{P}(x)$ and $\mathcal{P}(h^{(\ell)})$. To address this joint estimation problem, we treat it under the framework of Bayesian inference, which provides several analytical and optimal estimators. Among them, we are interested in minimum mean square error (MMSE) estimator \cite[Chapter 10]{kay1993fundamentals}, which is optimal in MSE sense. The MMSE estimator of $\bs{X}^{(\ell)}$ and $\bs{H}^{(\ell)}$ are given by
\begin{align}
&\forall n, k,\ell \ :\ \hat{x}_{nk}^{(\ell)}=\mathbb{E}\left[x_{nk}^{(\ell)}|\bs{Y}\right],
\label{Eq_MMSE_x}\\
&\forall m,n,\ell: \ \hat{h}_{mn}^{(\ell)}=\mathbb{E}\left[h_{mn}^{(\ell)}|\bs{Y}\right],
\label{Eq_MMSE_h}
\end{align}
where the expectations are taken over the marginal distributions $\mathcal{P}(x_{nk}^{(\ell)}|\bs{Y})$ and $\mathcal{P}(h_{mn}^{(\ell)}|\bs{Y})$, respectively, which are the marginalization of $\mathcal{P}(\bs{X}^{(\ell)},\bs{H}^{(\ell)}|\bs{Y})$. The posterior distribution $\mathcal{P}(\bs{X}^{(\ell)},\bs{H}^{(\ell)}|\bs{Y})$ is written as
\begin{align}
\nonumber
\mathcal{P}(\bs{X}^{(\ell)},\bs{H}^{(\ell)}|\bs{Y})=&\frac{1}{\mathcal{P}(\bs{Y})}\int \prod_{l\ne \ell}^L\text{d}\bs{H}^{(l)}\prod_{l\ne \ell}^L\text{d}\bs{X}^{(l)}\\
&\quad \times   \left[\mathcal{P}(\bs{X})\prod_{\iota=1}^L\mathcal{P}(\bs{H}^{(\iota)})\mathcal{P}(\bs{X}^{(\iota+1)}|\bs{H}^{(\iota)},\bs{X}^{(\iota)})\right],
\label{JP1}
\end{align}
where $\mathcal{P}(\bs{Y})$ is the partition function. The MMSE estimators minimize the MSEs defined as
\begin{align}
\textsf{mse}(\bs{X}^{(\ell)})&=\frac{1}{N_{\ell}K}\mathbb{E}\left\{\|\hat{\bs{X}}^{(\ell)}-\bs{X}^{(\ell)}\|_{\text{F}}^2\right\},\\
\textsf{mse}(\bs{H}^{(\ell)})&=\frac{1}{N_{\ell+1}N_{\ell}}\mathbb{E}\left\{\|\hat{\bs{H}}^{(\ell)}-\bs{H}^{(\ell)}\|_{\text{F}}^2\right\},
\end{align}
where the expectations are taken over  $\mathcal{P}(\bs{X}^{(\ell)},\bs{Y})$ and $\mathcal{P}(\bs{H}^{(\ell)},\bs{Y})$, respectively. Additionally,  $\hat{\bs{H}}^{(\ell)}=\{\hat{h}_{mn}^{(\ell)}, \forall m,n\}$ and $\hat{\bs{X}}^{(\ell)}=\{\hat{x}_{nk}^{(\ell)}, \forall n,k\}$.

Actually, the exact MMSE estimator is generally prohibitive due to the high-dimensional integrals. Recent advances in signal processing \cite{donoho2009message}, \cite{2010arXiv1010.5141R} showed that the exact estimator can be efficiently approximated by the sum-product LBP, and a renowned solution for the single-layer case was BiG-AMP \cite{parker2014bilinear}. The multi-layer generalized bilinear regression problem is more general and complex than the single layer, and the technical challenge lies in the design of message passing in the middle layer. In this context, we propose multi-layer bilinear generalized approximate message passing (ML-BiGAMP) as an extension of the BiG-AMP to the multi-layer case.

\begin{algorithm}[!t]
\caption{ML-BiGAMP}
\label{alg:ML-BiGAMP}
{
\begingroup
\textbf{1.Initialization:} Choosing $\{Z_{mk}^{(\ell)}(1),V_{mk}^{(\ell)}(1)\}$, $\{\hat{h}_{mn}^{(\ell)}(1),v^{(h,\ell)}_{mn}(1)\}$, $\{\hat{x}_{nk}^{(\ell)}(1),v_{nk}^{(x,\ell)}(1)\}$.\\
\textbf{2.Output:} $\hat{\bs{X}}^{(\ell)}$, $\hat{\bs{H}}^{(\ell)}$.\\
\textbf{3.Iteration:} (for $t=1,\cdots,T$)\\
\For{$\ell=L,\cdots,1$}
{
 $\{\text{Module A}^{(\ell)}\}$
 \setlength\abovedisplayskip{0pt}
 \setlength\belowdisplayskip{0pt}
 \begin{align}
 \tilde{z}_{mk}^{(\ell)}(t)&=\mathbb{E}[\zeta_{mk}^{(\ell)}(t)]
 \tag{R1} \label{R1}\\
 \tilde{v}_{mk}^{(\ell)}(t)&=\text{Var}[\zeta_{mk}^{(\ell)}(t)]
 \tag{R2} \label{R2}\\
  \hat{s}_{mk}^{(\ell)}(t)&=(\tilde{z}_{mk}^{(\ell)}(t)-Z_{mk}^{(\ell)}(t))/(V_{mk}^{(\ell)}(t))
 \tag{R3} \label{R3}\\
  v^{(s,\ell)}_{mk}(t)&=(V_{mk}^{(\ell)}(t)-\tilde{v}_{mk}^{(\ell)}(t))/((V_{mk}^{(\ell)}(t))^2)
 \tag{R4} \label{R4}\\
  \Sigma_{nk}^{(x,\ell)}(t)&= \left(\sum\nolimits_{m=1}^{N_{\ell+1}}|\hat{h}_{mn}^{(\ell)}(t)|^2v_{mk}^{(s,\ell)}(t)\right)^{-1}
 \tag{R5} \label{R5}\\
  R_{nk}^{(x,\ell)}(t)&=\hat{x}_{nk}^{(\ell)}(t)\left[1-\Sigma_{nk}^{(x,\ell)}(t)\sum_{m=1}^{N_{\ell+1}}v^{(h,\ell)}_{mn}(t)v_{mk}^{(s,\ell)}(t)\right]+\Sigma_{nk}^{(x,\ell)}(t)\sum_{m=1}^{N_{\ell+1}}(\hat{h}_{mn}^{(\ell)}(t))^*\hat{s}^{(\ell)}_{mk}(t)
  \tag{R6} \label{R6}\\
  \Sigma_{mn}^{(h,\ell)}(t)&=\left(\sum_{k=1}^{K}|\hat{x}_{nk}^{(\ell)}(t)|^2v_{mk}^{(s,\ell)}(t)\right)^{-1}
  \tag{R7} \label{R7}\\
  R_{mn}^{(h,\ell)}(t)&=\hat{h}_{mn}^{(\ell)}(t)\left[1-\Sigma_{mn}^{(h,\ell)}(t)\sum_{k=1}^{K}v^{(x,\ell)}_{nk}(t)v^{(s,\ell)}_{mk}(t)\right]+\Sigma_{mn}^{(h,\ell)}(t)\sum_{k=1}^{K}(\hat{x}_{nk}^{(\ell)}(t))^{*}\hat{s}_{mk}^{(\ell)}(t)
  \tag{R8} \label{R8}
 \end{align}
}
\For{$\ell=1,\cdots,L$}
{
    $\{\text{Module B}^{(\ell)}\}$
    \setlength\abovedisplayskip{0pt}
    \setlength\belowdisplayskip{0pt}
    \begin{align}
    \hat{x}_{nk}^{(\ell)}(t+1)&=\mathbb{E}[\xi_{nk}^{(x,\ell)}(t+1)]
    \tag{R9}    \label{R9}\\
    v_{nk}^{(x,\ell)}(t+1)&=\text{Var}[\xi_{nk}^{(x,\ell)}(t+1)]
    \tag{R10}   \label{R10}\\
   \hat{h}_{mn}^{(\ell)}(t+1)&=\mathbb{E}[\xi_{mn}^{(h,\ell)}(t+1)]
   \tag{R11}    \label{R11}\\
    v_{mn}^{(h,\ell)}(t+1)&=\text{Var}[\xi_{mn}^{(h,\ell)}(t+1)]
    \tag{R12}   \label{R12}\\
    \overline{V}_{mk}^{(\ell)}(t+1)&=\sum\nolimits_{n=1}^{N_{\ell}}\left[
   |\hat{x}_{nk}^{(\ell)}(t+1)|^2v_{mn}^{(h,\ell)}(t+1)+|\hat{h}_{mn}^{(\ell)}(t+1)|^2v_{nk}^{(x,\ell)}(t+1)
   \right]
   \tag{R13}    \label{R13}\\
   \overline{Z}_{mk}^{(\ell)}(t+1)&=\sum\nolimits_{n=1}^{N_{\ell}}\hat{h}_{mn}^{(\ell)}(t+1)\hat{x}_{nk}^{(\ell)}(t+1)
   \tag{R14}    \label{R14}\\
   V_{mk}^{(\ell)}(t+1)&=\overline{V}_{mk}^{(\ell)}(t+1)+\sum_{n=1}^{N_{\ell}}v^{(h,\ell)}_{mn}(t+1)v^{(x,\ell)}_{nk}(t+1)
   \tag{R15}    \label{R15}\\
   Z_{mk}^{(\ell)}(t+1)&=\overline{Z}_{mk}^{(\ell)}(t+1)-\hat{s}_{mk}^{(\ell)}(t)\overline{V}_{mk}^{(\ell)}(t+1)
   \tag{R16} \label{R16}
  \end{align}
}
\endgroup
}
\end{algorithm}

\subsection{The ML-BiGAMP Algorithm}
\begin{figure}[!t]
\centering
\includegraphics[width=0.6\textwidth]{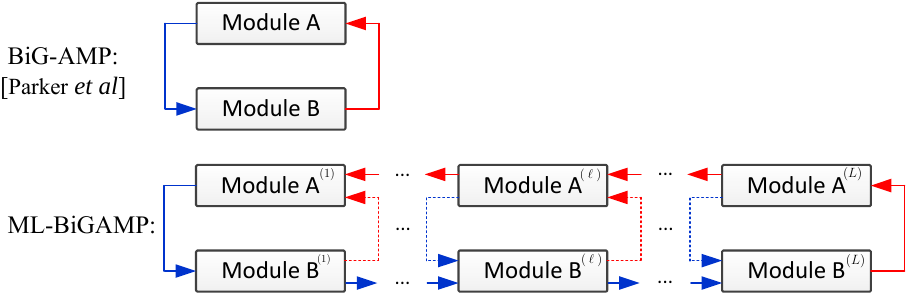}
\caption{The framework of ML-BiGAMP and BiG-AMP.
}
\label{Fig.ML-BiGAMP}
\end{figure}

The ML-BiGAMP algorithm described in Algorithm \ref{alg:ML-BiGAMP} operates in an iterative manner and thus organizes its message passing in two directions, one for the forward and the reverse. Per-iteration of the algorithm seen Fig.~\ref{Fig.ML-BiGAMP} works in a cyclic manner: $\text{Module A}^{(L)} \rightarrow \cdots \text{Module A} ^{(\ell)} \rightarrow \cdots  \text{Module A}^{(1)}\rightarrow \text{Module B}^{(1)} \rightarrow \cdots \text{Module B}^{(\ell)} \rightarrow \cdots  \text{Module B}^{(L)}$.

Module A$^{(\ell)}$ involves the \textit{scalar estimations} (R1)-(R2) and vector valued operations (R3)-(R8). In (R1)-(R2), the parameters $(\tilde{z}_{mk}^{(\ell)}(t),\tilde{v}_{mk}^{(\ell)}(t))$ represent the mean and variance of random variable (RV) $\zeta_{mk}^{(\ell)}(t)$ drawn by the approximate posterior distribution $\hat{\mathcal{P}}^t(z_{mk}^{(\ell)}|y)$ of $z_{mk}^{(\ell)}$, which for $\ell=L$ is expressed as,
\begin{align}
\zeta_{mk}^{(\ell)}(t)\sim
\begin{cases}
\frac{ \mathcal{P}(y_{mk}|z_{mk}^{(\ell)})\mathcal{N}(z_{mk}^{(\ell)}|Z_{mk}^{(\ell)}(t),V_{mk}^{(\ell)}(t))}{\int \mathcal{P}(y_{mk}|z^{(\ell)})\mathcal{N}(z^{(\ell)}|Z_{mk}^{(\ell)}(t),V_{mk}^{(\ell)}(t))\text{d}z^{(\ell)}} &\ell=L\\
\frac{\int \mathcal{N}_{x|z_{mk}}^{(\ell)}(Z_{mk}^{(\ell)}(t),V_{mk}^{(\ell)}(t),R_{mk}^{(x,\ell+1)}(t),\Sigma_{mk}^{(x,\ell+1)}(t))\text{d}x^{(\ell+1)}}
{\int \mathcal{N}_{x|z}^{(\ell)}(Z_{mk}^{(\ell)}(t),V_{mk}^{(\ell)}(t),R_{mk}^{(x,\ell+1)}(t),\Sigma_{mk}^{(x,\ell+1)}(t))\text{d}z^{(\ell)}\text{d}x^{(\ell+1)}}
& \ell<L
\end{cases}
\label{Equ:zeta},
\end{align}
where
$$\mathcal{N}_{x|z}^{(\ell)}(a,A,b,B)=\mathcal{P}(x^{(\ell+1)}|z^{(\ell)})\mathcal{N}(z^{(\ell)}|a,A)\mathcal{N}(x^{(\ell+1)}|b,B).$$
Note that the term $\mathcal{N}(z_{mk}^{(\ell)}|Z_{mk}^{(\ell)}(t), V_{mk}^{(\ell)}(t))$ is $t$-iteration approximate prior of $z_{mk}^{(\ell)}$ , i.e., $\hat{\mathcal{P}}^t(z_{mk}^{(\ell)})$; while $\mathcal{N}(x_{mk}^{(\ell+1)}|R_{mk}^{(x,\ell+1)}(t),\Sigma_{mk}^{(x,\ell+1)}(t))$ is $t$-iteration approximate likelihood function from $x_{mk}^{(\ell+1)}$ to observation, i.e., $\hat{\mathcal{P}}^t(y|x_{mk}^{(\ell+1)})$.

Similar to Module A$^{(\ell)}$, Module B$^{(\ell)}$ also includes \textit{scalar estimations} (R9)-(R12) and vector valued operations (R13)-(R16). The parameters $(\hat{x}_{nk}^{(\ell)}(t+1),v_{nk}^{(x,\ell)}(t+1))$ denote the mean and variance of RV $\xi_{nk}^{(x,\ell)}(t+1)$, which $\text{for}\ \ell=1$ follows
\begin{align}
\xi_{nk}^{(x,\ell)}(t+1)\sim
\begin{cases}
 \frac{\mathcal{P}(x_{nk}^{(\ell)})\mathcal{N}(x_{nk}^{(\ell)}|R^{(x,\ell)}_{nk}(t),\Sigma_{nk}^{(x,\ell)}(t))}{\int \mathcal{P}(x)\mathcal{N}(x|R^{(x,\ell)}_{nk}(t),\Sigma_{nk}^{(x,\ell)}(t))\text{d}x} &\ell=1\\
 \frac{\int \mathcal{N}_{x_{nk}|z}^{(\ell-1)}(Z_{nk}^{(\ell-1)}(t+1),V_{nk}^{(\ell-1)}(t+1),R^{(x,\ell)}_{nk}(t),\Sigma_{nk}^{(x,\ell)}(t))\text{d}z^{(\ell-1)}}
{\int \mathcal{N}_{x|z}^{(\ell-1)}(Z_{nk}^{(\ell-1)}(t+1),V_{nk}^{(\ell-1)}(t+1),R^{(x,\ell)}_{nk}(t),\Sigma_{nk}^{(x,\ell)}(t))\text{d}z^{(\ell-1)}x^{(\ell)}} &\ell>1
\end{cases}
\label{Equ:xi_x}.
\end{align}
Moreover, the parameters $\hat{h}_{mn}^{(\ell)}(t+1)$ and $v^{(h,\ell)}_{mn}(t+1)$ refer to the mean and variance of RV $\xi^{(h,\ell)}_{mn}(t+1)$ distributed as
\begin{align}
\xi^{(h,\ell)}_{mn}(t+1)\sim \frac{\mathcal{P}(h_{mn}^{(\ell)})\mathcal{N}(h_{mn}^{(\ell)}|R^{(h,\ell)}_{mn}(t),\Sigma_{mn}^{(h,\ell)}(t))}{\int \mathcal{P}(h)\mathcal{N}(h|R^{(h,\ell)}_{mn}(t),\Sigma_{mn}^{(h,\ell)}(t))\text{d}h},
\label{Equ:xi_h}
\end{align}
where the term $\mathcal{N}(h_{mn}^{(\ell)}|R^{(h,\ell)}_{mn}(t),\Sigma_{mn}^{(h,\ell)}(t))$ is $t$-iteration approximate likelihood function from $h_{mn}^{(\ell)}$ to observation, i.e., $\hat{\mathcal{P}}^t(y|h_{mn}^{(\ell)})$.

To derive the proposed ML-BiGAMP algorithm, we first use the factor graph to represent the joint posterior distribution $\mathcal{P}(\bs{X}^{(\ell)},\bs{H}^{(\ell)}|\bs{Y})$ in (\ref{JP1}), which includes variable nodes (sphere) and factor nodes (cube). Then the marginal posterior distribution can be approximated by sum-product loopy belief propagation (LBP), which is impractical in large system limit. To reduce the complexity of LBP, we simplify the messages from factor nodes to variable nodes by central limit theorem (CLT) and Taylor expansion. The messages from variable nodes to factor nodes are updated by Gaussian reproduction property. Besides, several new variables in belief distributions are defined to establish the relationship between belief distributions and the messages from variable nodes to factor nodes, where the Taylor expansion is applied again. By ignoring infinitesimals, the ML-BiGAMP is obtained. The detailed derivation of ML-BiGAMP is presented in Appendix \ref{Appendix:A}.
\\
\indent  Compared to BiG-AMP \cite{parker2014bilinear}, a major difference in our derivation of the ML-BiGAMP is that the factor node $\mathcal{P}(\bs{Y}|\bs{Z})$ in BiG-AMP only connects $\bs{X}$ and $\bs{H}$, which are all variable nodes of current layer, but in a multi-layer setup, it is generalized as $\mathcal{P}(\bs{X}^{(\ell+1)}|\bs{Z}^{(\ell)})$, which is a \textit{junction node} that connects  not only $\bs{X}^{(\ell)}$ and $\bs{H}^{(\ell)}$ in current layer, but also $\bs{X}^{(\ell+1)}$ in next layer. As a result, message update in multi-layer setup is more complex.  For example, as shown in Fig.~\ref{Fig.FG1-2layer}, the message (blue arrow) from factor node $\mathcal{P}(x_{mk}^{(2)}|z_{mk}^{(1)})$ to variable node $x_{nk}^{(1)}$ in Fig.~\ref{Fig.FG1-2layer} (b) should be updated by combining the message (red arrow) from the 2nd layer, while the message (blue arrow) in Fig.\ref{Fig.FG1-2layer} (a) is updated without adjacent layer.
\begin{figure*}[!t]
\centering
\includegraphics[width=0.8\textwidth]{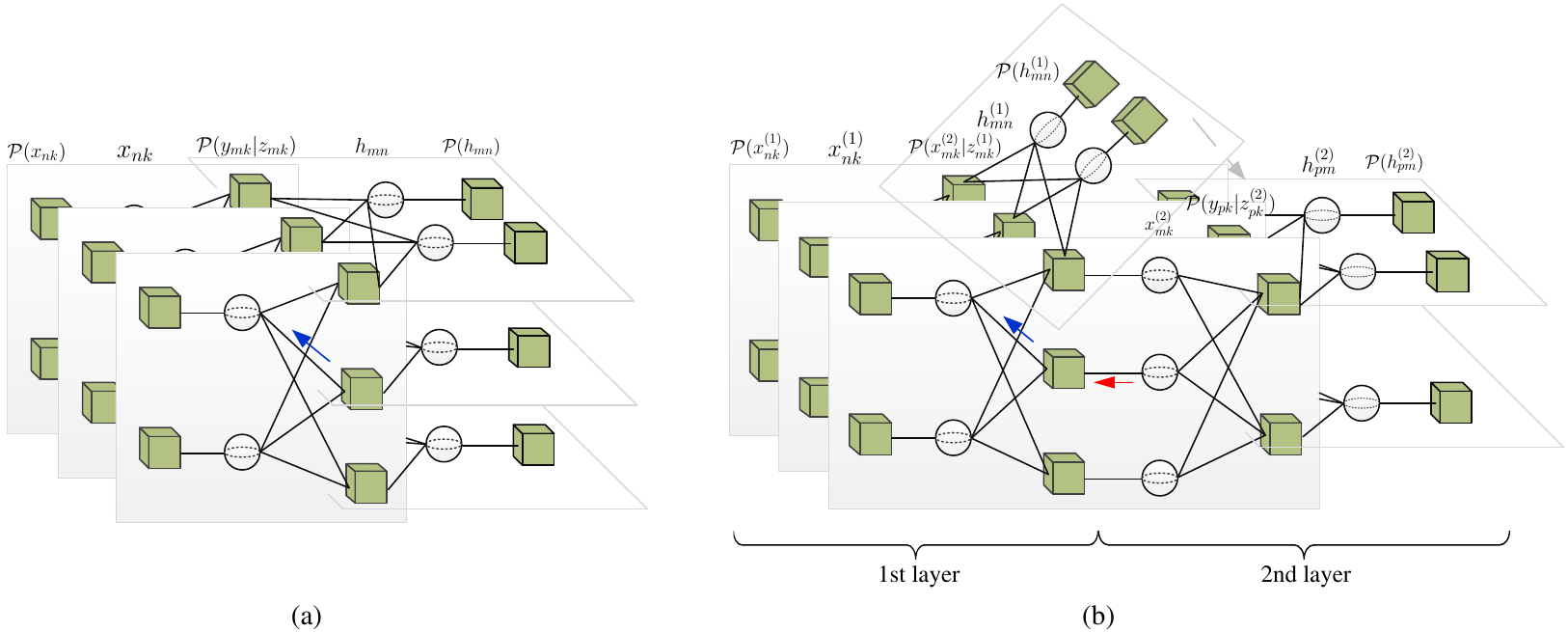}
\caption{(a) Factor graph of 1-layer model; (b) Factor graph of 2-layer model. The sphere denotes variable node while cube refers to factor node.
}
\label{Fig.FG1-2layer}
\end{figure*}

\subsection{Relation to Previous AMP-like Algorithms}
\begin{remark}
The ML-BiGAMP algorithm is a general algorithm, which degenerates smoothly to the existing AMP-like algorithms: BiG-AMP\cite{parker2014bilinear}, GAMP\cite{2010arXiv1010.5141R}, AMP\cite{2010arXiv1010.5141R}, as well as ML-AMP \cite{manoel2017multi}.
\end{remark}

 \textbf{($L=1$ and unknown $\bs{H}$)} By setting $L=1$, the ML-BiGAMP reduces to the BiG-AMP algorithm \cite[Table III]{parker2014bilinear}, where the RVs in (\ref{Equ:zeta}), (\ref{Equ:xi_x}), and (\ref{Equ:xi_h}) become
\begin{align}
\zeta_{mk}(t)&\sim \frac{ \mathcal{P}(y_{mk}|z_{mk})\mathcal{N}(z_{mk}|Z_{mk}(t),V_{mk}(t))}{\int \mathcal{P}(y_{mk}|z)\mathcal{N}(z|Z_{mk}(t),V_{mk}(t))\text{d}z},\\
\xi_{nk}^{(x)}(t+1)&\sim \frac{\mathcal{P}(x_{nk})\mathcal{N}(x_{nk}|R^{(x)}_{nk}(t),\Sigma_{nk}^{(x)}(t))}{\int \mathcal{P}(x)\mathcal{N}(x|R^{(x)}_{nk}(t),\Sigma_{nk}^{(x)}(t))\text{d}x},\\
\xi^{(h)}_{mn}(t+1)&\sim \frac{\mathcal{P}(h_{mn})\mathcal{N}(h_{mn}|R^{(h)}_{mn}(t),\Sigma_{mn}^{(h)}(t))}{\int \mathcal{P}(h)\mathcal{N}(h|R^{(h)}_{mn}(t),\Sigma_{mn}^{(h)}(t))\text{d}h}.
\end{align}

\textbf{($L=1$ and  known $\bs{H}$)} If the measurement matrix is further perfectly given, then  we have $\hat{h}_{mn}(t)=h_{mn}$ and $v^{(h)}_{mn}=0, \forall m,n$. Accordingly, the ML-BiGAMP algorithm  reduces to GAMP algorithm \cite[Algorithm 1]{2010arXiv1010.5141R} as below
\begin{subequations}
\begin{align}
\tilde{z}_{mk}(t)&=\mathbb{E}[\zeta_{mk}(t)],\\
\tilde{v}_{mk}(t)&=\text{Var}[\zeta_{mk}(t)],\\
\hat{s}_{mk}(t)&=(\tilde{z}_{mk}(t)-Z_{mk}(t))/(V_{mk}(t)),\\
v^{(s)}_{mk}(t)&=(V_{mk}(t)-\tilde{v}_{mk}(t))/((V_{mk}(t))^2),\\
\Sigma_{nk}^{(x)}(t)&= \left(\sum_{m}|h_{mn}|^2v_{mk}^{(s)}(t)\right)^{-1},\\
R_{nk}^{(x)}(t)&=\hat{x}_{nk}(t)+\Sigma_{nk}^{(x)}(t)\sum_{m}h_{mn}^*\hat{s}_{mk}(t),\\
\hat{x}_{nk}(t+1)&=\mathbb{E}[\xi_{nk}^{(x)}(t+1)],\\
v_{nk}^{(x)}(t+1)&=\text{Var}[\xi_{nk}^{(x)}(t+1)],\\
V_{mk}(t+1)&=\sum_n|h_{mn}|^2v_{nk}^{(x)}(t+1),\\
Z_{mk}(t+1)&=\sum_n h_{mn}\hat{x}_{nk}(t+1)-\hat{s}_{mk}(t)V_{mk}(t+1).
\end{align}
\end{subequations}

\textbf{($L=1$, known $\bs{H}$, and Gaussian transition)} Further, when the standard linear model is considered, where the transition distribution becomes $\mathcal{P}(y_{mk}|z_{mk})=\mathcal{N}(z_{mk}|y_{mk},\sigma_w^2)$, the ML-BiGAMP degenerates to the AMP algorithm \cite{donoho2009message}, where
\begin{align}
\Sigma_{nk}^{(x)}(t)&=\left(\sum_{m}\frac{|h_{mn}|^2}{\sigma_w^2+V_{mk}(t)}\right)^{-1},\\
R_{nk}^{(x)}(t)&=\hat{x}_{nk}(t)+\Sigma_{nk}^{(x)}(t)\sum_{m}\frac{h_{mn}^{*}(y_{mk}-Z_{mk}(t))}{\sigma_w^2+V_{mk}(t)}.
\end{align}

\textbf{($L\geq 1$ and known $\bs{H}$)} Besides, the ML-BiGAMP algorithm can also recover the ML-AMP algorithm \cite[(5)]{manoel2017multi}. For the case of $L\geq 1$ and known measurement matrix, we have $\hat{h}_{mn}^{(\ell)}(t)=h_{mn}$ and $v^{(h,\ell)}_{mn}=0, \forall m,n,\ell$. In the sequel, the ML-AMP algorithm can be obtained by plugging $\hat{h}^{(\ell)}_{mn}(t)=h_{mn}^{(\ell)}$ and $v_{mn}^{(h,\ell)}=0$ into ML-BiGAMP algorithm.

\subsection{Computational Complexity}
We now look at the ML-BiGAMP's computational complexity. As shown in Algorithm~\ref{alg:ML-BiGAMP}, the ML-BiGAMP algorithm involves two directions: reverse  and forward direction. Furthermore, there are linear steps and non-linear steps in both the forward and reverse directions.
\begin{itemize}
\item
The non-linear steps of the reverse direction refer to (R1)-(R2) in Algorithm \ref{alg:ML-BiGAMP}. The computation of the parameters $(\tilde{z}_{mk}^{(\ell)},\tilde{v}_{mk}^{(\ell)})$ does not change with the dimension.
\item
The linear steps of the reverse direction refer to (R3)-(R8) and their computational cost is dominated by the componentwise squares of $\hat{\bs{X}}^{(\ell)}$ in (R5) and $\hat{\bs{H}}^{(\ell)}$ in (R7).  The computational cost of the linear steps is $\mathcal{O}(N_{\ell+1}N_{\ell}K)$. As a result,  the total computational cost of reverse direction is $\mathcal{O}(N_{\ell+1}N_{\ell}K)$.
\item
Similarly, the computational cost of non-linear steps (R9)-(R12) in forward direction is $\mathcal{O}(N_{\ell}K)$.
Furthermore, the computational cost of linear steps of forward direction is dominated by componentwise squares of $\hat{\bs{H}}^{(\ell)}$ and $\hat{\bs{X}}^{(\ell)}$ in (R13), which is $\mathcal{O}(N_{\ell+1}N_{\ell}K)$.
\end{itemize}
Hence, the total computational cost of ML-BiGAMP is $\mathcal{O}(N_{\ell+1}N_{\ell}KLT)$ with $L$ and T being the number of layers and iteration numbers, respectively. By considering $K$ and $N_{\ell}$ with the same order and large system limit, the complexity of ML-BiGAMP is $\mathcal{O}(N_{\ell}^3)$, which is the same as BiG-AMP \cite{parker2014bilinear} and far less than ML-Mat-VAMP \cite{pandit2020inference} with $\mathcal{O}(N_{\ell}^4)$. Meanwhile, similar to BiG-AMP, the proposed ML-BiGAMP algorithm reduces the vector operation to a sequence of linear transforms and scalar estimation functions.

\subsection{Damping}
In practical applications, similar to other members in the AMP family, damping is applied to ensure convergence of the proposed ML-BiGAMP algorithm. Let $\beta(t)\in (0,1]$ denote the damping factor, then the following low-passed-filter damping is applied to the parameters $\hat{s}_{mk}^{(\ell)}(t), v_{mk}^{(s,\ell)}(t)$, $ \Sigma_{nk}^{(x,\ell)}(t), \Sigma_{mn}^{(h,\ell)}(t), \hat{x}_{nk}^{(\ell)}(t), \hat{h}_{mn}^{(\ell)}(t)$, and $V_{mk}^{(\ell)}(t)$:
\begin{align}
a(t) = \beta(t) a(t) + [1-\beta(t)] a(t-1),
\end{align}
where $a(t)$ is the parameter at $t$-iteration.
In particular, we used $\beta(t)=0.7$ for our multi-layer JCD simulations. In the case of single-layer with known measurement matrix, we only apply damping factor $\beta(t)=0.95$ to the parameters $V_{mk}^{(\ell)}(t)$ and $Z_{mk}^{(\ell)}(t)$.

\begin{algorithm}[!t]
\caption{State Evolution of ML-BiGAMP}
\label{alg:SE}
{
\begingroup
\textbf{Output:} $\textsf{mse}_{x}^{(\ell)}=\chi_x^{(\ell)}-q_x^{(\ell)}$, $\textsf{mse}_h^{(\ell)}=\chi_h^{(\ell)}-q_h^{(\ell)}$.\\
\For{$\ell=1,\cdots,L$}
{
    \begin{align*}
    \chi_x^{(\ell)}&=\begin{cases}
    \int x^2\mathcal{P}(x)\text{d}x   &\ell=1\\
    \int (x^{(\ell)})^2\mathcal{P}(x^{(\ell)}|z^{(\ell-1)})\mathcal{N}(z^{(\ell-1)}|0,\chi_z^{(\ell-1)})\text{d}z^{(\ell-1)}\text{d}x^{(\ell)} &\ell>1
    \end{cases}\\
    \chi_h^{(\ell)}&=\int (h^{(\ell)})^2\mathcal{P}(h^{(\ell)})\text{d}h^{(\ell)}\\
    \chi_z^{(\ell)}&=N_{\ell}\chi_x^{(\ell)}\chi_h^{(\ell)}
    \end{align*}
}
\For{$\ell=L,\cdots,1$}
{
\begin{align*}
q_z^{(\ell)}
&=\begin{cases}
q_z^{(L)}=\int \frac{\left[\int z^{(L)} \mathcal{P}(y|z^{(L)})\mathcal{N}(z^{(L)}|\sqrt{\chi_z^{(L)}-V^{(L)}}\xi,V^{(L)})\text{d}z^{(L)}\right]^2}
{\int  \mathcal{P}(y|z^{(L)})\mathcal{N}(z^{(L)}|\sqrt{\chi_z^{(L)}-V^{(L)}}\xi,V^{(L)})\text{d}z^{(L)}}\text{D}\xi\text{d}y
&\ell=L\\
q_z^{(\ell)}=\int \frac{\left[\int z^{(\ell)} \mathcal{N}_{x|z}^{(\ell)}(\sqrt{\chi_z^{(\ell)}-V^{(\ell)}}\xi,V^{(\ell)},\zeta,\Sigma^{(x,\ell+1)})\text{d}x^{(\ell+1)}\text{d}z^{(\ell)}\right]^2}{\int \mathcal{N}_{x|z}^{(\ell)}(\sqrt{\chi_z^{(\ell)}-V^{(\ell)}}\xi,V^{(\ell)},\zeta,\Sigma^{(x,\ell+1)})\text{d}x^{(\ell+1)}\text{d}z^{(\ell)}}\text{D}\xi\text{d}\zeta
 &\ell<L
\end{cases}\\
\Sigma^{(x,\ell)}
&=\frac{N_{\ell}(\chi_x^{(\ell)}\chi_h^{(\ell)}-q_x^{(\ell)}q_h^{(\ell)})^2}{\beta_{\ell}q_h^{(\ell)}(q_z^{(\ell)}-N_{\ell}q_x^{(\ell)}q_h^{(\ell)})}\\
\Sigma^{(h,\ell)}
&=\frac{\alpha \prod_{l=1}^{\ell-1}\beta_{l}N_{\ell}(\chi_x^{(\ell)}\chi_h^{(\ell)}-q_x^{(\ell)}q_h^{(\ell)})^2}{q_x^{(\ell)}(q_z^{(\ell)}-N_{\ell}q_x^{(\ell)}q_h^{(\ell)})}
\end{align*}
}
\For{$\ell=1,\cdots,L$}
{
  \begin{align*}
  q_x^{(\ell)}
  &=\begin{cases}
  \int \frac{\left[\int x\mathcal{P}(x)\mathcal{N}(x|\zeta,\Sigma^{(x,\ell)})\text{d}x\right]}
  {\int \mathcal{P}(x)\mathcal{N}(x|\zeta,\Sigma^{(x,\ell)})\text{d}x}\text{d}\zeta &\ell=1\\
  \int \frac{\left[\int x^{(\ell)} \mathcal{N}_{x|z}^{(\ell-1)}(\sqrt{\chi_z^{(\ell-1)}-V^{(\ell-1)}}\xi,V^{(\ell-1)},\zeta,\Sigma^{(x,\ell)})\text{d}x^{(\ell)}\text{d}z^{(\ell-1)}\right]^2}{\int \mathcal{N}_{x|z}^{(\ell-1)}(\sqrt{\chi_z^{(\ell-1)}-V^{(\ell-1)}}\xi,V^{(\ell-1)},\zeta,\Sigma^{(x,\ell)})\text{d}x^{(\ell)}\text{d}z^{(\ell-1)}}\text{D}\xi\text{d}\zeta &\ell>1
  \end{cases}\\
  q_h^{(\ell)}&=\int \frac{\left[\int h^{(\ell)}\mathcal{P}(h^{(\ell)})\mathcal{N}(h^{(\ell)}|\zeta,\Sigma^{(h,\ell)})\text{d}h^{(\ell)}\right]^2}
  {\int \mathcal{P}(h^{(\ell)})\mathcal{N}(h^{(\ell)}|\zeta,\Sigma^{(h,\ell)})\text{d}h^{(\ell)}}\text{d}\zeta\\
  V^{(\ell)}&=N_{\ell}(\chi_h^{(\ell)}\chi_x^{(\ell)}-q_h^{(\ell)}q_x^{(\ell)})
  \end{align*}
}
\endgroup
}
\end{algorithm}

\section{State Evolution}
In this section, we present the state evolution (SE) analysis for the ML-BiGAMP algorithm, which illustrates that the asymptotic MSE performance of the ML-BiGAMP algorithm can be fully characterized via a set of simple one-dimensional equations under the large system limit. Previous work pertaining to SE analysis for AMP-like algorithms was found in \cite{bayati2011dynamics}, in which the SE was mathematically rigorous.
In our derivation of SE analysis, we use some concepts (Definition 1, 2, and Assumption 1) from \cite{2010arXiv1010.5141R}. However, the SE of ML-BiGAMP presented here is different from SE of \cite{2010arXiv1010.5141R} in the following aspects. Firstly, \cite{2010arXiv1010.5141R} considered single layer with known measurement matrix, while we consider the multi-layer bilinear generalized model. Secondly, we give the detailed SE derivation (although heuristic) that features a special treatment towards the marginal density function: they are interpreted as the transitional probability of a Markov chain, c.f. (100) in the Appendix B, while the details of SE's proof of \cite{2010arXiv1010.5141R} are omitted. The SE analysis is extracted from the practical algorithm after averaging the observed signal and measurement matrices.  It is worthy of noting that the analysis is based on \textit{large system limit}, that is, when $\forall \ell, N_{\ell}, K\rightarrow \infty$ but the ratios
\begin{align}
\frac{N_1}{K}=\alpha, \ \frac{N_{\ell+1}}{N_{\ell}}=\beta_{\ell},
\end{align}
are fixed and finite.

\begin{proposition}
\label{Pro_SE}
In large system limit, by averaging the observation, the asymptotic MSE performance of the ML-BiGAMP algorithm can be fully characterized by a set of scalar equations termed state evolution shown in Algorithm \ref{alg:SE}.
\end{proposition}

\proof: See Appendix \ref{Appendix:B}.

\begin{remark}
Given the SE of ML-BiGAMP, we can recover the SEs of BiG-AMP\cite{parker2014bilinear}, GAMP \cite{2010arXiv1010.5141R}, AMP \cite{donoho2009message}, and ML-GAMP\cite{zou2020estimation}.
\end{remark}

\textbf{($L=1$ and unknown $\bs{H}$)} By setting $L=1$ and using $\beta=\beta_1$, $N=N_1$, and $M=N_2$, the SE equations of the ML-BiGAMP become
\begin{subequations}
\begin{align}
q_z&=\int \frac{\left[\int z\mathcal{P}(y|z)\mathcal{N}(z|\sqrt{Nq_xq_h}\xi,V)\text{d}z\right]^2}{\int \mathcal{P}(y|z)\mathcal{N}(z|\sqrt{Nq_xq_h}\xi,V)\text{d}z}\text{D}\xi\text{d}y,\\
\Sigma^{(x)}&=\frac{N(\chi_x\chi_h-q_xq_h)^2}{\beta q_h(q_z-Nq_xq_h)},\\
\Sigma^{(h)}&=\frac{\alpha N(\chi_x\chi_h-q_xq_h)^2}{q_x(q_z-Nq_xq_h)},\\
q_x&=\int \frac{\left[\int x\mathcal{P}(x)\mathcal{N}(x|\zeta,\Sigma^{(x)})\text{d}x\right]^2}{\int \mathcal{P}(x)\mathcal{N}(x|\zeta,\Sigma^{(x)})\text{d}x}\text{d}\zeta,\\
q_h&=\int \frac{\left[\int h\mathcal{P}(h)\mathcal{N}(h|\zeta,\Sigma^{(h)})\text{d}h\right]^2}{\int \mathcal{P}(h)\mathcal{N}(h|\zeta,\Sigma^{(h)})\text{d}h}\text{d}\zeta,\\
V&=N(\chi_x\chi_h-q_xq_h),
\end{align}
\end{subequations}
where $\chi_x=\int x^2\mathcal{P}(x)$, $\chi_h=\int h^2\mathcal{P}(h)\text{d}h$, and $\text{D}\xi=\mathcal{N}(\xi|0,1)\text{d}\xi$.

\textbf{($L=1$ and known $\bs{H}$)} If the measurement matrix $\bs{H}$ is perfectly given, we then have $\chi_h=q_h$, i.e., $\textsf{mse}_h=0$. Considering $h_{mn}^2$ with order $\mathcal{O}(\frac{1}{N_2})$, the following can be obtained
\begin{subequations}
\label{Equ:GLM_SE}
\begin{align}
\!\!\!\!q_z&=\int \frac{\left[\int z\mathcal{P}(y|z)\mathcal{N}\left(z|\sqrt{\frac{q_x}{\beta}}\xi,\frac{\chi_x-q_x}{\beta}\right)\text{d}z\right]^2}{\int \mathcal{P}(y|z)\mathcal{N}\left(z|\sqrt{\frac{q_x}{\beta}}\xi,\frac{\chi_x-q_x}{\beta}\right)\text{d}z}\text{D}\xi\text{d}y,\\
\!\!\!\!\Sigma^{(x)}&=\frac{(\chi_x-q_x)^2}{\beta(\beta q_z-q_x)},\\
\!\!\!\!q_x&=\int \frac{\left[\int x\mathcal{P}(x)\mathcal{N}(x|\zeta,\Sigma^{(x)})\text{d}x\right]^2}{\int \mathcal{P}(x)\mathcal{N}(x|\zeta,\Sigma^{(x)})\text{d}x}\text{d}\zeta.
\end{align}
\end{subequations}

\textbf{($L=1$, known $\bs{H}$, and Gaussian transition)} When we further consider the Gaussian transition distribution, i.e.,  $\mathcal{P}(y_{mk}|z_{mk})=\mathcal{N}(z_{mk}|y_{mk},\sigma_w^2)$, by Gaussian reproduction property\footnote{$\mathcal{N}(x|a,A)\mathcal{N}(x|b,B)=\mathcal{N}(0|a-b, A+B) \mathcal{N}(x|c,C)$ with $C=(A^{-1}+B^{-1})^{-1}$ and $c=C\cdot(\frac{a}{A}+\frac{b}{B})$.} and $\textsf{mse}_x=\chi_x-q_x$ the SE of ML-BiGAMP becomes
\begin{align}
\Sigma^{(x)}=\sigma_w^2+\frac{1}{\beta}\textsf{mse}_x(\Sigma^{(x)}).
\end{align}
It is found that the SE of the ML-BiGAMP algorithm in standard linear model setting is precisely equal to the SE of AMP \cite{bayati2011dynamics,guo2005randomly}.

\textbf{($L\geq 1$ and known $\bs{H}$)} If we consider the case of $L\geq1$ and known measurement matrices $\{\bs{H}^{(\ell)}\}_{\ell=1}^L$, then SE of ML-BiGAMP algorithm degenerates into the previous SE of ML-GAMP algorithm \cite{zou2020estimation} (including ML-AMP \cite{manoel2017multi} as its special case).

\section{Relation to Exact MMSE estimator}
The proposed algorithm is derived from the sum-product LBP followed by AMP approximation, and it is well-known that the sum-product LBP generally provides a good approximation to MMSE estimator \cite{winn2005variational}. The MMSE estimator is known as Bayes-optimal in MSE sense but is infeasible in practice due to multiple integrals. In this section, we establish that the asymptotic MSE predicted by ML-BiGAMP'SE agrees perfectly with the MMSE estimator predicted by replica method. The key strategy of analyzing MSE of MMSE estimator is through averaging free energy
\begin{align}
\mathcal{F}=\lim_{N_{1}\rightarrow \infty}\frac{1}{N_1^2}\mathbb{E}_{\bs{Y}}\{\log \mathcal{P}(\bs{Y})\},
\label{Equ:FE}
\end{align}
where $\mathcal{P}(\bs{Y})$ is partition function. The analysis is based on large system limit and we simply apply $N_1\rightarrow \infty$ to denote the large system limit. Actually, even in large system limit the computation of (\ref{Equ:FE}) is difficult due to the expectation of the logarithm of $\mathcal{P}(\bs{Y})$. Using the note\footnote{The following formula is applied from right to left
$$
\lim_{\tau\rightarrow 0}\frac{\partial }{\partial \tau}\log \mathbb{E}\{\Theta^{\tau}\}=\lim_{\tau\rightarrow 0}\frac{\mathbb{E}\{\Theta^{\tau}\log \Theta\}}{\mathbb{E}\{\Theta^{\tau}\}}=\mathbb{E}\{\log \Theta \},
$$
where $\Theta$ is any positive random variable.
}, it can be facilitated by rewriting $\mathcal{F}$ as
\begin{align}
\mathcal{F}=\lim_{N_1\rightarrow \infty}\frac{1}{N_1^2}\lim_{\tau\rightarrow 0}\frac{\partial }{\partial \tau}\log \mathbb{E}_{\bs{Y}}\{\mathcal{P}^{\tau}(\bs{Y})\}.
\end{align}

To ease the statement, we firstly calculate the free energy considering a representative two-layer model, and it leads to the saddle point equations. By replica symmetry assumption, the fixed point equations can be obtained by solving the saddle point equations. Finally, we extend the results of the two-layer model into multi-layer regime with similar procedures where the Proposition \ref{Pro_De} and Proposition \ref{Pro_Optimal} can be obtained.


\subsection{Performance Analysis}
\begin{proposition}[Decoupling principle]
\label{Pro_De}
In large system limit, by replica method, the input output of the multi-layer generalized bilinear model is decoupled into a bank of scalar AWGN channels  w.r.t. the input signal $\bs{X}$ and measurement matrices $\{\bs{H}^{(\ell)}\}_{\ell=1}^L$
\begin{align}
y_{x}&=x+w_{x},
\label{SE_x}\\
y_{h^{(\ell)}}&=h^{(\ell)}+w_{h^{(\ell)}},
\label{SE_h}
\end{align}
where $w_{x}\sim \mathcal{N}(w_x|0,(2\hat{q}_x)^{-1})$, $x\sim \mathcal{P}(x)$, $h^{(\ell)}\sim \mathcal{P}(h^{(\ell)})$, and $w_{h^{(\ell)}}\sim \mathcal{N}(w_{h^{(\ell)}}|0, (2\hat{q}_h^{(\ell)})^{-1})$. The parameters $\hat{q}_x$ and $\hat{q}_h^{(\ell)}$ are from the fixed point equations in (\ref{TSE_qh})-(\ref{TSE_qx}) of the exact MMSE estimator, for $\ell=1,\cdots, L$,
\begin{subequations}
\begin{align}
\chi_h^{(\ell)}&=\int (h^{(\ell)})^2\mathcal{P}(h^{(\ell)}){\mathrm{d}}h^{(\ell)},
\label{TSE_qh}\\
\hat{q}_x^{(\ell)}&=\frac{\beta_{\ell}q_h^{(\ell)}}{2}\frac{q_z^{(\ell)}-N_{\ell}q_x^{(\ell)}q_h^{(\ell)}}{N_{\ell}(\chi_x^{(\ell)}\chi_h^{(\ell)}-q_x^{(\ell)}q_h^{(\ell)})^2},\\
\hat{q}_h^{(\ell)}&=\frac{q_x^{(\ell)}}{2\alpha\prod_{l=1}^{\ell-1}\beta_{l}}
\frac{q_z^{(\ell)}-N_{\ell}q_x^{(\ell)}q_h^{(\ell)}}{N_{\ell}(\chi_x^{(\ell)}\chi_h^{(\ell)}-q_x^{(\ell)}q_h^{(\ell)})^2},\\
q_h^{(\ell)}
&=\int \frac{[\int h^{(\ell)}\mathcal{P}(h^{(\ell)})\mathcal{N}(h^{(\ell)}|\zeta,\frac{1}{2\hat{q}_h^{(\ell)}}){\rm{d}}h^{(\ell)}]^2}
{\int \mathcal{P}(h^{(\ell)})\mathcal{N}(h^{(\ell)}|\zeta,\frac{1}{2\hat{q}_h^{(\ell)}}){\rm{d}}h^{(\ell)}}{\rm{d}}\zeta\\
\chi_x^{(\ell)}&=
\begin{cases}
\int x^2\mathcal{P}(x){\rm{d}}x &\ell=1\\
\chi_x^{(\ell)}=\int (x^{(\ell)})^2\mathcal{P}(x^{(\ell)}|z^{(\ell-1)})\mathcal{N}(z^{(\ell-1)}|0, \chi_z^{\ell-1}){\rm{d}}z^{(\ell-1)}{\rm{d}}x^{(\ell)} &\ell>1
\end{cases}\\
q_z^{(\ell)}&=
\begin{cases}
\int\frac{\left[\int z^{(\ell)}\mathcal{P}(y|z^{(\ell)})\mathcal{N}\left(z^{(\ell)}|\sqrt{N_{\ell}q_x^{(\ell)}q_h^{(\ell)}}\xi,N_{\ell}(\chi_x^{(\ell)}\chi_h^{(\ell)}-q_x^{(\ell)}q_h^{(\ell)})\right){\rm{d}}z^{(\ell)}\right]^2}
{\int \mathcal{P}(y|z^{(\ell)})\mathcal{N}\left(z^{(\ell)}|\sqrt{N_{\ell}q_x^{(\ell)}q_h^{(\ell)}}\xi,N_{\ell}(\chi_x^{(\ell)}\chi_h^{(\ell)}-q_x^{(\ell)}q_h^{(\ell)})\right){\rm{d}}z^{(\ell)}}{\rm{D}}\xi{\rm{d}}y & \ell=L\\
\int \frac{\left[\int z^{(\ell)}\mathcal{N}_{x|z}^{(\ell)}\left(\sqrt{N_{\ell}q_h^{(\ell)}q_x^{(\ell)}}\xi,N_{\ell}(\chi_h^{(\ell)}\chi_x^{(\ell)}-q_h^{(\ell)}q_x^{(\ell)}),\zeta,\frac{1}{2\hat{q}_x^{(\ell+1)}}\right){\rm{d}}x^{(\ell+1)}{\rm{d}}z^{(\ell)}\right]^2}
{\int \mathcal{N}_{x|z}^{(\ell)}\left(\sqrt{N_{\ell}q_h^{(\ell)}q_x^{(\ell)}}\xi,N_{\ell}(\chi_h^{(\ell)}\chi_x^{(\ell)}-q_h^{(\ell)}q_x^{(\ell)}),\zeta,\frac{1}{2\hat{q}_x^{(\ell+1)}}\right){\rm{d}}x^{(\ell+1)}{\rm{d}}z^{(\ell)}}{\rm{D}}\xi{\rm{d}}\zeta
&\ell<L
\end{cases}\\
q_x^{(\ell)}&=
\begin{cases}
\int \frac{\left[\int x\mathcal{P}(x)\mathcal{N}(x|\zeta,\frac{1}{2\hat{q}_x^{(1)}}){\rm{d}}x\right]^2}
{\int \mathcal{P}(x) \mathcal{N}(x|\zeta,\frac{1}{2\hat{q}_x^{(1)}}){\rm{d}}x}{\rm{d}}\zeta  &\ell=1\\
\int \frac{\left[\int x^{(\ell)}\mathcal{N}_{x|z}^{(\ell-1)}(\sqrt{N_{\ell-1}q_h^{(\ell-1)}q_x^{(\ell-1)}}\xi,N_{\ell-1}(\chi_h^{(\ell-1)}\chi_x^{(\ell-1)}-q_h^{(\ell-1)}q_x^{(\ell-1)}),\zeta,\frac{1}{2\hat{q}_x^{(\ell)}}){\rm{d}}z^{(\ell-1)}{\rm{d}}x^{(\ell)}\right]^2}
{\int \mathcal{N}_{x|z}^{(\ell-1)}(\sqrt{N_{\ell-1}q_h^{(\ell-1)}q_x^{(\ell-1)}}\xi,N_{\ell-1}(\chi_h^{(\ell-1)}\chi_x^{(\ell-1)}-q_h^{(\ell-1)}q_x^{(\ell-1)}),\zeta,\frac{1}{2\hat{q}_x^{(\ell)}}){\rm{d}}z^{(\ell-1)}{\rm{d}}x^{(\ell)}}{\rm{D}}\xi{\rm{d}}\zeta
&\ell>1
\end{cases}
\label{TSE_qx}
\end{align}
\end{subequations}
\end{proposition}
\textit{Proof}: See Appendix \ref{Appendix:C}.

\begin{proposition}[Optimality]
\label{Pro_Optimal}
In large system limit, the SE equations of the proposed ML-BiGAMP algorithm depicted in Algorithm \ref{alg:SE} match perfectly the fixed point equations in (\ref{TSE_qh})-(\ref{TSE_qx}) of the exact MMSE estimator as predicted by replica method under setting
\begin{align}
\Sigma^{(x,\ell)}= \frac{1}{2\hat{q}_x^{(\ell)}},\quad \Sigma^{(h,\ell)}= \frac{1}{2\hat{q}_h^{(\ell)}}.
\end{align}
\end{proposition}
The Proposition \ref{Pro_Optimal} indicates that the proposed ML-BiGAMP algorithm can attain the MSE performance of the exact MMSE estimator, which is Bayes-optimal but is generally computationally intractable except all priors and transition distributions being Gaussian.

\subsection{Parameters of Proposition \ref{Pro_De}}
Based on Proposition \ref{Pro_De}, the MSE performances of $\bs{X}$ and $\bs{H}^{(\ell)}$ can be fully characterized by the scalar AWGN channel (\ref{SE_x}) and (\ref{SE_h}), while the former should be run for a sufficiently large number of iterations (independent of the system dimensions). We note that under certain inputs, when the signal-to-noise ratio (SNR) related parameters $1/(2\hat{q}_x)$ and $1/(2\hat{q}_h^{(\ell)})$ are given, the analytical expression of MSE of MMSE estimator is possible.

For the model (\ref{SE_x}) and (\ref{SE_h}),  we get the MMSE estimators:
\begin{align}
\hat{x}&=\mathbb{E}\{x|y_{x}\}=\int x\mathcal{P}(x|y_{x})\text{d}x,\\
\hat{h}^{(\ell)}&=\mathbb{E}\{h^{(\ell)}|y_{h^{(\ell)}}\}=\int h^{(\ell)}\mathcal{P}(h^{(\ell)}|y_{h^{(\ell)}})\text{d}h^{(\ell)}.
\end{align}
The MSEs of those MMSE estimators are given by
\begin{align}
\nonumber
\textsf{mse}_x&=\mathbb{E}_{x,y_{x}}\{\left(x-\mathbb{E}\{x|y_{x}\}\right)^2\}\\
&=\chi_x-q_x,\\
\nonumber
\textsf{mse}_{h}^{(\ell)}&=\mathbb{E}_{h^{(\ell)},y_{\textsf{h}^{(\ell)}}}\left\{\left(h^{(\ell)}-\mathbb{E}\{h^{(\ell)}|y_{h^{(\ell)}}\}\right)^2\right\}\\
&=\chi_h^{(\ell)}-q_h^{(\ell)},
\end{align}
where $\chi_x=\mathbb{E}\{x^2\}$, $q_x=\mathbb{E}_{y_{x}}\{(\mathbb{E}\{x|y_{x}\})^2\}$, $\chi_h^{(\ell)}=\mathbb{E}\{(h^{(\ell)})^2\}$, and $q_h^{(\ell)}=\mathbb{E}_{y_{h^{(\ell)}}}\{(\mathbb{E}\{h^{(\ell)}|y_{h^{(\ell)}}\})^2\}$.

Below we only give a belief review of the MSE of the MMSE estimator of $x$, and that of  $h^{(\ell)}$ can be obtained with similar steps.

\textit{Example 1 (Gaussian input):} For the Gaussian input $x\sim \mathcal{N}(x|0,\sigma_x^2)$, the MSE of the MMSE estimator for the scalar channel (\ref{SE_x}) can be obtained by Gaussian reproduction property
\begin{align}
\textsf{mse}_x=\frac{\sigma_x^2}{1+2\hat{q}_x\sigma_x^2}.
\end{align}

\textit{Example 2 (constellation-like input):} Considering the quadrature phase shift keying (QPSK) constellation symbol, the MSE of the MMSE estimator for scalar channel (\ref{SE_x}) is given by \cite{lozano2006optimum}
\begin{align}
\textsf{mse}_x=1-\int \tanh\left(2\hat{q}_x+\sqrt{2\hat{q}_x}\zeta\right)\text{D}\zeta.
\end{align}
The corresponding symbol error rate (SER) w.r.t. $\bs{X}$ can be evaluated through the scalar AWGN channel (\ref{SE_x}), which is given by \cite{wen2015bayes}
\begin{align}
\textsf{SER}_{x}=2\mathcal{Q}\left(\sqrt{2\hat{q}_x}\right)-\left[\mathcal{Q}\left(\sqrt{2\hat{q}_x}\right)\right]^2,
\end{align}
where $\mathcal{Q}(x)=\int_x^{+\infty}\text{D}z$ is the Q-function.

\textit{Example 3 (Bernoulli-Gaussian input):} The Bernoulli-Gaussian input, i.e., $\mathcal{P}(x)=\rho\mathcal{N}(x|0,\rho^{-1})+(1-\rho)\delta(x)$, is common in the recovery of sparse signal. In this case, the MSE of the MMSE estimator for  the scalar channel (\ref{SE_x}) can be obtained by Gaussian reproduction property and convolution formula
\begin{align}
\textsf{mse}_x&=1-\frac{2\hat{q}_x\rho}{\rho+2\hat{q}_x}\int \frac{\zeta^2}{\rho+(1-\rho)\sqrt{\frac{\rho+2\hat{q}_x}{\rho}}\exp \left(-\frac{\hat{q}_x\zeta^2}{\rho}\right)}\text{D}\zeta.
\end{align}

\textit{Example 4 (Gaussian mixture input):} In \cite{wen2014channel}, the channel of massive MIMO system considering the pilot contamination is modeled as Gaussian mixture, i.e., $\mathcal{P}(x)=\sum_{i=1}^\kappa\rho_i\mathcal{N}(x|0,\sigma_i^2)$, where $\rho_i$ and $\sigma_i^2$ are the mixing probability and the power of the $i$-th Gaussian mixture component, respectively. To implement channel estimation, a message passing based method is developed. For Gaussian mixture input, the MSE of MMSE estimator of the scalar AWGN channel (\ref{SE_x}) is given by
\begin{align}
\textsf{mse}_x
=\sum_{i=1}^{\kappa}\rho_i\sigma_i^2-\int \frac{\left[\sum_{i=1}^{\kappa}\frac{\rho_i\sigma_i^2y}{\sigma_i^2+(2\hat{q}_x)^{-1}}\mathcal{N}(y|0,\sigma_i^2+(2\hat{q}_x)^{-1})\right]^2}{\sum_{i=1}^{\kappa}\rho_i\mathcal{N}(y|0,\sigma_i^2+(2\hat{q}_x)^{-1})}\text{d}y.
\end{align}

\section{Simulation and Discussion}
In this section\footnote{The codes of simulations can be found in \textit{branches} from the link: https://github.com/QiuyunZou/ML-BiGAMP.}, we firstly develop a joint channel and data (JCD) estimation method based on the proposed algorithm for massive MIMO AF relay system. Secondly, we give the application of ML-BiGAMP in compressive matrix completion. Finally, we present the numerical simulations to validate the consistency of the ML-BiGAMP algorithm and its SE in different settings (prior or layer). Here we only consider random Gaussian measurement matrices, but the proposed algorithm and its SE empirically hold in more generalized regions such as discrete uniform distribution, Bernoulli Gaussian, and Gaussian mixture, etc.

\begin{figure}[!t]
\centering
\includegraphics[width=0.6\textwidth]{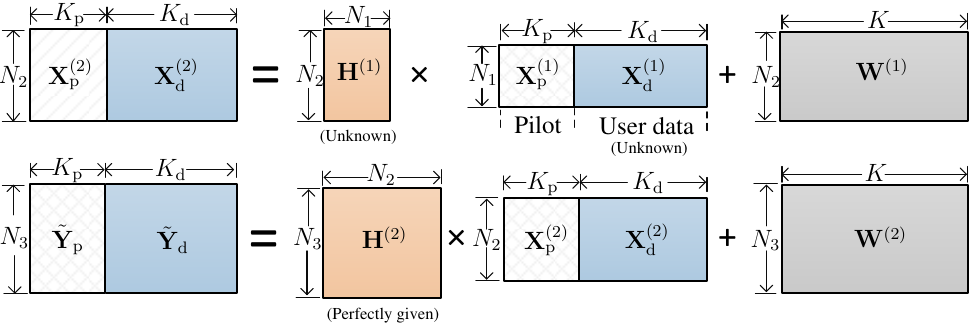}
\caption{Signal model for massive MIMO AF relay system.
}
\label{Fig:JCD_System}
\end{figure}

\subsection{JCD Method Based on Proposed Algorithm}
As shown in Fig.~\ref{Fig:JCD_System} and also described in Section \ref{Sec:Example}, the massive MIMO AF relay system can be modeled as \footnote{For simplification, we consider the relay antennas equipped with $\infty$-bit ADCs and this system can be combined as a single-layer model with non-white noise. In fact, it is the ML-BiGAMP that can be applied to the case of relay with low-precision ADCs.}
\begin{align}
\begin{cases}
\bs{X}^{(2)}=\bs{H}^{(1)}\bs{X}^{(1)}+\bs{W}^{(1)}\\
\ \ \ \bs{Y}=\textsf{Q}_{\textsf{c}}\left(\bs{H}^{(2)}\bs{X}^{(2)}+\bs{W}^{(2)}\right)
\end{cases}.
\label{Equ_Relay_System}
\end{align}
 To estimate the channel $\bs{H}^{(1)}$, the original signal $\bs{X}^{(1)}$ is divided into two parts. The first $K_{\text{p}}$ symbols of the block of $K$ symbols serve as the pilot sequences, while the remaining $K_{\text{d}}=K-K_{\text{p}}$ symbols are data transmission, i.e., $\bs{X}^{(1)}=[\bs{X}_{\text{p}}^{(1)},\bs{X}_\text{d}^{(1)}]$, where both $\bs{X}_{\text{p}}^{(1)}$ and $\bs{X}_{\text{d}}^{(1)}$ are quadrature phase shift keying (QPSK) symbol. As a toy model, we assume that the channel $\bs{H}^{(2)}$ in the second layer is perfectly known, but be aware that our ML-BiGAMP algorithm is generally applicable to those cases of an unknown $\bs{H}^{(2)}$. $\bs{W}^{(1)}$ and $\bs{W}^{(2)}$ refer to additive white
Gaussian noise and it is assumed that they have the same power $\sigma_w^2$. The signal-to-noise ratio (SNR) is defined as $1/\sigma_w^2$. Additionally, $\textsf{Q}_{\textsf{c}}(\cdot)$ represents a low-resolution complex-valued quantizer including two separable real-valued quantizer $\textsf{Q}(\cdot)$, i.e.,
\begin{align}
\forall m,k: \quad y_{mk}=\textsf{Q}(\Re(\tilde{y}_{mk}))+\mathbb{J}\textsf{Q}(\Im(\tilde{y}_{mk})),
\label{Equ:Q1}
\end{align}
where $\mathbb{J}^2=-1$, $\tilde{\bs{Y}}=\{\tilde{y}_{mk},\forall m,k\}=\bs{H}^{(2)}\bs{X}^{(2)}+\bs{W}^{(2)}$ and $\textsf{Q}(\cdot):\mathbb{R}\mapsto \mathcal{R}_B$ with $\mathcal{R}_B$ being the set of $B$-bits ADCs defined as $\mathcal{R}_B=\{(-\frac{1}{2}+b)\triangle; \quad b=-2^{B-1}+1,\cdots,2^{B-1}\}$ and $\triangle$ being an uniform quantization step size. For the output $y$ of ADCs, its input $\tilde{y}$ is assigned within the range of $(q^{\text{low}}(y),q^{\text{up}}(y)]$, which reads \cite{wen2015bayes}
\begin{align}
q^{\text{low}}(y)
&=
\begin{cases}
y-\frac{\triangle}{2}, &\text{if}\  y>\min \mathcal{R}_B\\
-\infty  &\text{otherwise}
\end{cases},\\
q^{\text{up}}(y)
&=
\begin{cases}
y+\frac{\triangle}{2}, &\text{if}\ y<\max\mathcal{R}_B\\
+\infty   &\text{otherwise}
\end{cases}.
\label{Equ:Q2}
\end{align}
Accordingly, the transition distribution from $\bs{Z}^{(2)}$ to $\bs{Y}$ of this quantization model reads
\begin{align*}
\mathcal{P}(y|z^{(2)})
&=\Psi\left(\Re(y)|\Re(z^{(2)}),\frac{\sigma_w^2}{2}\right)\Psi\left(\Im(y)|\Im(z^{(2)}),\frac{\sigma_w^2}{2}\right)
\end{align*}
where $\Psi(y|v,c^2)=\Phi(\frac{q^{\text{up}}(y)-v}{c})-\Phi(\frac{q^{\text{low}}(y)-v}{c})$ with $\Phi(x)=\int_{-\infty}^x\mathcal{N}(t|0,1)\text{d}t$. Furthermore, the main technical challenges in particularizing our algorithm and SE to the specific quantization model are the computation of parameters $(\tilde{z}_{mk}^{(2)},\tilde{v}_{mk}^{(2)})$ in practical algorithm  and $q_z^{(2)}$ in SE equations. The expressions of $(\tilde{z}_{mk}^{(2)},\tilde{v}_{mk}^{(2)})$ are given in \cite[(23)-(25)]{wen2015bayes}. The evaluation of $q_z^{(2)}$ can be found in \cite{wang2017bayesian}
\begin{align}
\nonumber
q_z^{(2)}
&=Mq_h^{(2)}q_x^{(2)}+\frac{(V^{(2)})^2}{\sigma_w^2+V^{(2)}}\sum_{y\in \mathcal{R}_B}\int \frac{\left[\phi(\varsigma_1(y,\xi))-\phi(\varsigma_2(y,\xi))\right]^2}{\Phi(\varsigma_1(y,\xi))-\Phi(\varsigma_2(y,\xi))}\text{D}\xi,
\end{align}
where $\varsigma_1(y,\xi)=\frac{q^{\text{up}}(y)-\sqrt{(\chi_z^{(2)}-V^{(2)})/2}\xi}{\sqrt{(\sigma_w^2+V^{(2)})/2}}$, $\varsigma_2(y,\xi)=\frac{q^{\text{low}}(y)-\sqrt{(\chi_z^{(2)}-V^{(2)})/2}\xi}{\sqrt{(\sigma_w^2+V^{(2)})/2}}$, and $\phi(x)=\mathcal{N}(x|0,1)$.

 In Fig.~\ref{Fig:Exp1}, the dimensions of the system  are set as $(K_{\text{p}},K_{\text{d}},N_1, N_2, N_3)=(100, 400, 50, 200, 400)$ and $\text{SNR}_1=\text{SNR}_2=1/\sigma_w^2=8$dB. As depicted in Fig.~\ref{Fig:Exp1}, the ML-BiGAMP and its SE converge very quickly within 12$\sim$15 iterations. More importantly, the normalized MSE (NMSE) performance of $\bs{X}_{\text{d}}^{(1)}$ ($\|\hat{\bs{X}}_{\text{d}}^{(1)}-\bs{X}_{\text{d}}^{(1)}\|_{\text{F}}^2/\|\bs{X}_{\text{d}}^{(1)}\|_{\text{F}}^2$ ) of ML-BiGAMP algorithm agrees with its SE ($\textsf{mse}_{x}$) perfectly in this two layer setting, where $\|\cdot\|_{\text{F}}$ denotes the Frobenius norm.

In Fig.~\ref{Fig:Exp2}, we present the bit error rate (BER) versus SNR  plot in terms of pilot-only, JCD, and perfect-CSI method. The pilot-only method involves two phases: train phase and data phase. In train phase, the pilot $\bs{X}_{\text{p}}^{(1)}$ is transmitted to estimate channel $\bs{H}^{(1)}$ using the proposed ML-BiGAMP algorithm. In data phase, the data $\bs{X}_{\text{d}}$ is detected using the proposed ML-BiGAMP algorithm based on the estimated channel. The JCD method is to jointly estimate channel and data using the proposed algorithm. In perfect-CSI (channel status information) method, the channel $\bs{H}^{(1)}$ is perfectly given and $\bs{X}_{\text{d}}^{(1)}$ is detected using the proposed algorithm. The dimensions of the system are set as $(K_{\text{p}}, K_{\text{d}}, N_1, N_2, N_3)=(100, 400, 50, 150, 300)$ and $\text{SNR}_1=\text{SNR}_2=1/\sigma_w^2$. As can be seen from Fig.~\ref{Fig:Exp2}, the JCD has a huge advantage over the pilot-only method. Meanwhile, there is small gap between JCD and perfect-CSI method, especially in $B=\infty$.

In Fig.~\ref{Fig:Exp3}, we present the influence of pilot length on NMSE performance of $\bs{H}^{(1)}$ of JCD and pilot-only method by varying $K_{\text{p}}/K$ from $0.1$ to $0.99$. The dimensions of the system are $(K,N_1,N_2,N_3)=(500, 50, 150, 300)$ and SNR is $\text{SNR}_1=\text{SNR}_2=5\text{dB}$. As shown in Fig.~\ref{Fig:Exp3}, the performance of JCD method is better than the pilot-only method, especially in low $K_p/K$. A straightforward ideal to reduce the gap between JCD and pilot-only method is increasing the pilot length.

\begin{figure}[!t]
\centering
\includegraphics[width=0.48\textwidth]{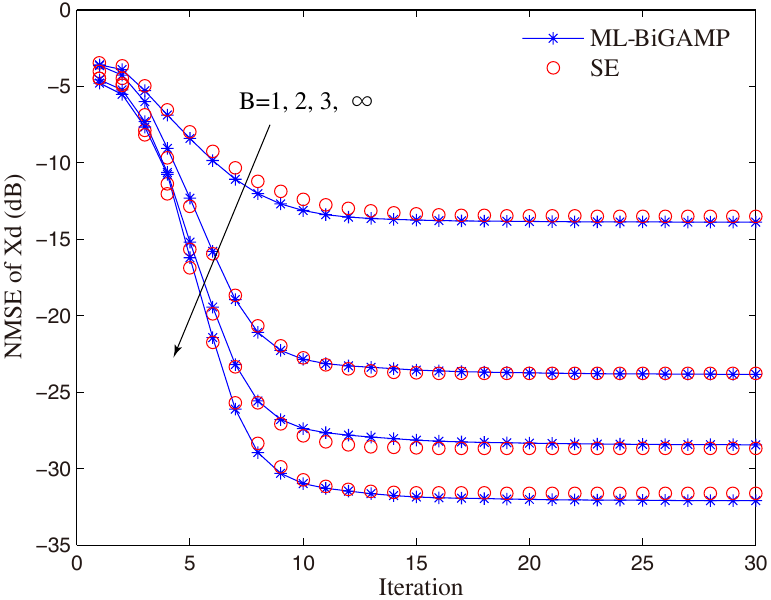}
\caption{Per-iteration behavior of ML-BiGAMP and its SE in two-hop AF relay communication.}
\label{Fig:Exp1}
\vspace{+0.4cm}
\centering
\includegraphics[width=0.48\textwidth]{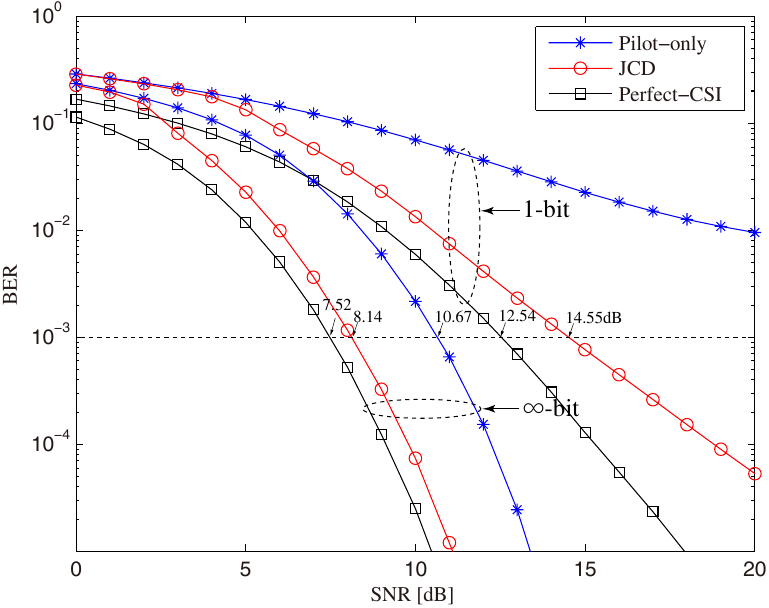}
\caption{BER behavior of pilot-only, JCD, and perfect-CSI method in two-hop AF relay communication.
}
\label{Fig:Exp2}
\end{figure}

\begin{figure}[!t]
\centering
\includegraphics[width=0.48\textwidth]{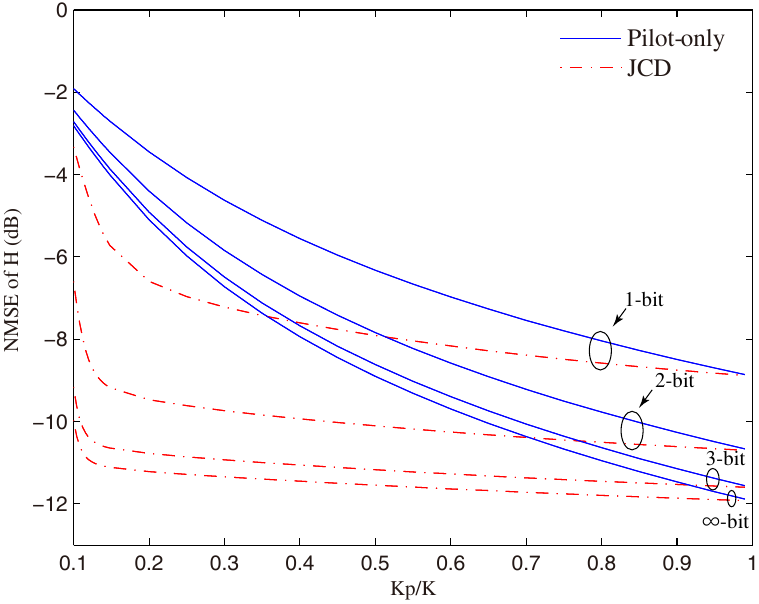}
\caption{NMSE w.r.t. $\bs{H}^{(1)}$ of JCD and pilot-only scheme  versus the pilot ratio $K_{\text{p}}/K$ for different bit quantizer.
}
\label{Fig:Exp3}
\vspace{+0.4cm}
\centering
\includegraphics[width=0.48\textwidth]{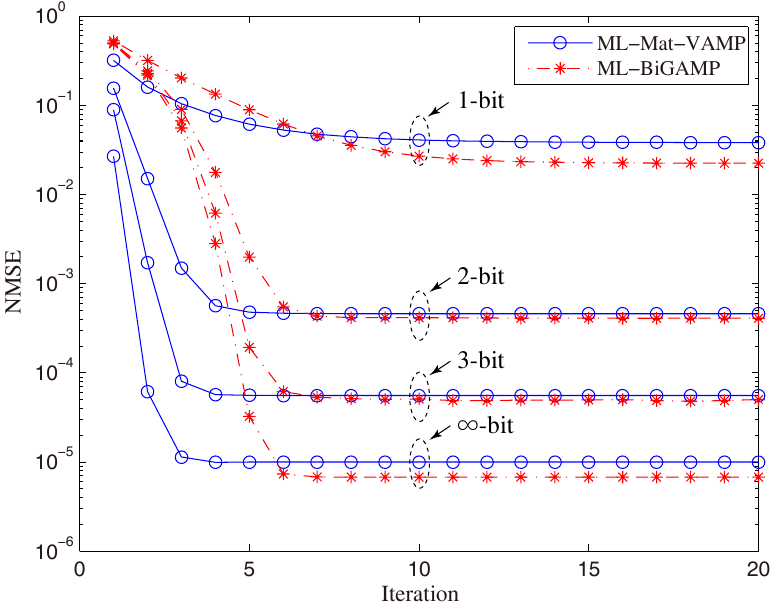}
\caption{NMSE performance of ML-BiGAMP and ML-Mat-VAMP in 3-layer model.
}
\label{Fig:Expb}
\end{figure}

In Fig.~\ref{Fig:Expb}, we compare the ML-BiGAMP with the competing ML-Mat-VAMP in 3-layer model: $\bs{X}^{(2)}=\bs{H}^{(1)}\bs{X}^{(1)}+\bs{W}^{(1)}, \bs{X}^{(3)}=\bs{H}^{(2)}\bs{X}^{(2)}+\bs{W}^{(2)}, \bs{Y}=\textsf{Q}_c(\bs{H}^{(3)}\bs{X}^{(2)}+\bs{W}^{(3)})$. The system dimensions are set as $(N_1, N_2, N_3, N_4, K)=(200, 400, 600, 800, 1)$. The noise power of each layer is considered to be equal and the SNR is set as $1/\sigma_w^2=15$dB. Besides, it is assumed that the measurement matrix of each layer is known.  As can be seen from Fig.~\ref{Fig:Expb}, at the case of 2-bit and 3-bit, the NMSE performance of ML-BiGAMP coincides with ML-Mat-VAMP almost. While, at the case of 1-bit and $\infty$-bit, ML-BiGAMP has a slight advantage over ML-Mat-VAMP on NMSE performance. Also, the convergence speed of ML-Mat-VAMP is faster than ML-BiGAMP, but it has to pay more computational cost.

\subsection{Compressive Matrix Completion}
As described in Section \ref{Sec:CMC}, the compressive matrix completion (MC) problem is formalized as
\begin{align}
\begin{cases}
\bs{X}^{(2)}=\boldsymbol{f}(\bs{H}^{(1)}\bs{X}^{(1)}+\bs{W}^{(1)})\\
\ \ \ \bs{Y}=\bs{H}^{(2)}\bs{X}^{(2)}+\bs{W}^{(2)}
\end{cases},
\end{align}
where $\boldsymbol{f}$ is a componentwise mapping which is specified by (\ref{Equ:MC}). In this problem, only a fraction $\epsilon=\frac{|\Omega|}{N_2 K}$ of entries of $\bs{X}^{(2)}$ are valid, where $\Omega$ is the set of valid entries of $\bs{X}^{(2)}$. In addition, it is assumed that $\bs{W}^{(1)}$ and $\bs{W}^{(2)}$ have the same power $\sigma_w^2$, and $\bs{H}^{(1)}$ and $\bs{X}^{(1)}$ are drawn from Gaussian distribution with zero mean and unit variance. For any $(m,k)\notin \Omega$, several quantities in Algorithm \ref{alg:ML-BiGAMP} becomes
$\tilde{z}_{mk}^{(1)}(t)=Z^{(1)}_{mk}(t)$, $\tilde{v}_{mk}^{(1)}(t)=V_{mk}^{(1)}(t)$, $\hat{x}_{mk}^{(2)}(t)=0$, and $v_{mk}^{(x,2)}(t)=0$. \\
\indent In Fig.~\ref{Fig:ExpCMC}, we show NMSE performance of $\bs{Z}$ defined as $\|\hat{\bs{H}}^{(1)}\hat{\bs{X}}^{(1)}-\bs{Z}^{(1)}\|_{\text{F}}^2/\|\bs{Z}^{(1)}\|_{\text{F}}^2$ over a grid of sampling ratios $\epsilon =\frac{|\Omega|}{N_2 K}$ and rank $N_1$. The dimensions of system are set as $(K, N_2, N_3)=(1000, 1000, 500)$ and the SNR is $\text{SNR}=1/\sigma_w^2=50\text{dB}$. The ``success'' (white grid) is defined as NMSE$<-50$dB. As we can seen from this figure, more valid entries or smaller rank $N_1$ can improve the NMSE performance of $\bs{Z}^{(1)}$.

\subsection{Validation for SE Using More Degenerated Cases }
In Fig.~\ref{Fig:Exp4}$\sim$\ref{Fig:Exp7}, we consider the model (\ref{alg:system}) in $L=1$ and $K=1$, i.e., $\bs{Y}=\boldsymbol{\phi}(\bs{HX}+\bs{W})$, where $\bs{H}$ is Gaussian random matrix and is perfectly given. Further, the deterministic and element-wise mapping $\boldsymbol{\phi}(\cdot)$ is particularized as quantization function defined by (\ref{Equ:Q1})-(\ref{Equ:Q2}).

In Fig.~\ref{Fig:Exp4}$\sim$\ref{Fig:Exp5}, to be specific, the application in compressive sensing (Bernoulli-Gaussian prior) is considered  by varying the sparse rate $\rho$ and the precision of ADCs. The SNR is defined as $1/\sigma_w^2$ and it is set as $12$dB. The dimensions of the system are $(N_2,N_1)=(512,1024)$, i.e., measurement ratio $N_2/N_1=0.5$.  In addition, the NMSE of $\bs{X}$ is defined as $\|\hat{\bs{X}}-\bs{X}\|^2/\|\bs{X}\|^2$. As can be seen from Fig.~\ref{Fig:Exp4} and Fig.~\ref{Fig:Exp5}, the SEs agree perfectly with the algorithm in all settings.

Meanwhile, the application of ML-BiGAMP in communication (QPSK symbols) is depicted in Fig.~\ref{Fig:Exp6}$\sim$\ref{Fig:Exp7} by varying the measurement ratio $N_2/N_1$ and the precision of ADCs. The SNR of them is set as 9dB. It can also be seen that the SEs predict the NMSE performance of the algorithm in all settings.

\begin{figure}[!t]
\centering
\includegraphics[width=0.48\textwidth]{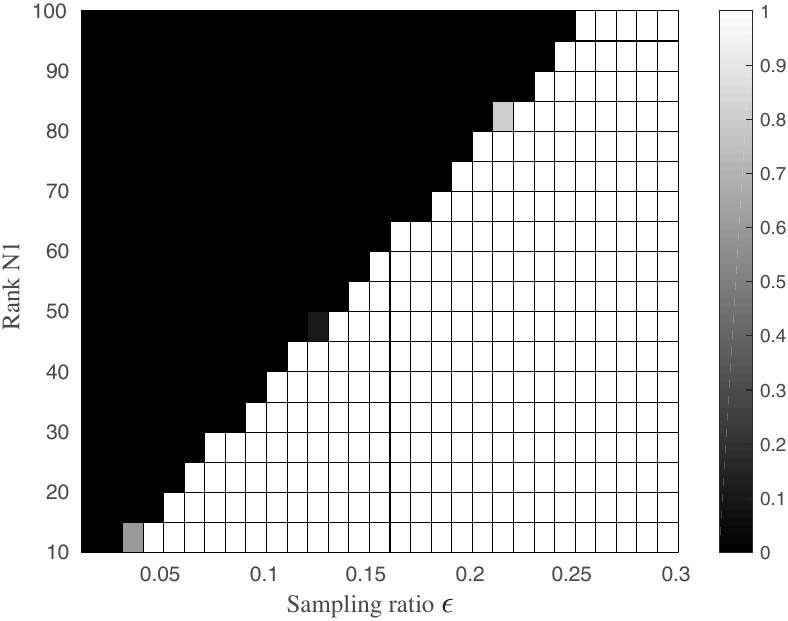}
\caption{Success rate over a grid of sampling ratios $\epsilon=\frac{|\Omega|}{N_2K}$ and ranks $N_1$. Here, ``success'' is defined as NMSE$<-50$dB.
}
\label{Fig:ExpCMC}
\vspace{+0.4cm}
\centering
\includegraphics[width=0.48\textwidth]{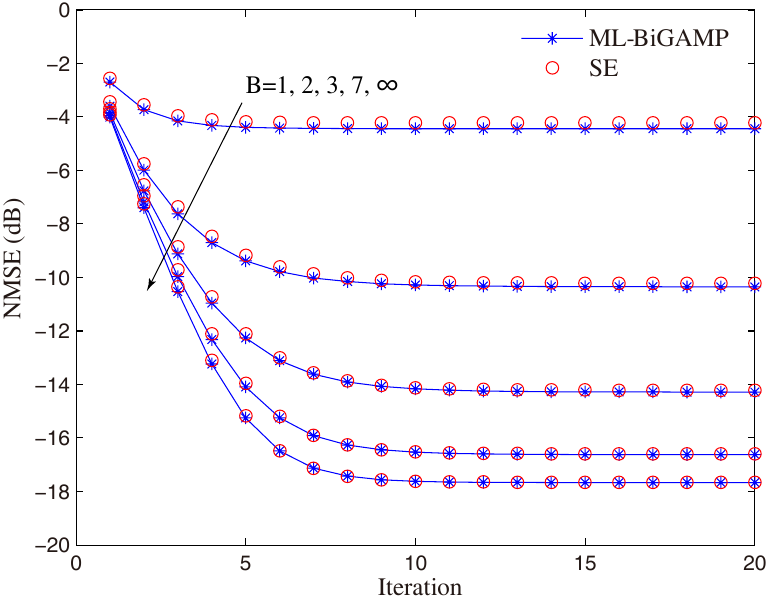}
\caption{Per-iteration behavior of ML-BiGAMP and its SE in compressive sensing (sparse rate $\rho=0.1$).
}
\label{Fig:Exp4}
\end{figure}

\begin{figure}[!t]
\centering
\vspace{+0.3cm}
\includegraphics[width=0.48\textwidth]{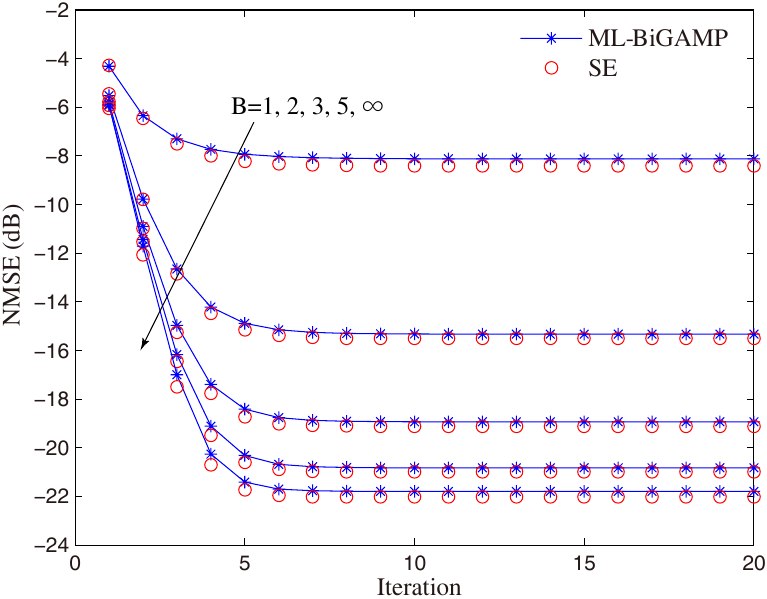}
\caption{Per-iteration behavior of ML-BiGAMP and its SE in compressive sensing (sparse rate $\rho=0.05$).
}
\label{Fig:Exp5}
\vspace{+0.3cm}
\centering
\includegraphics[width=0.48\textwidth]{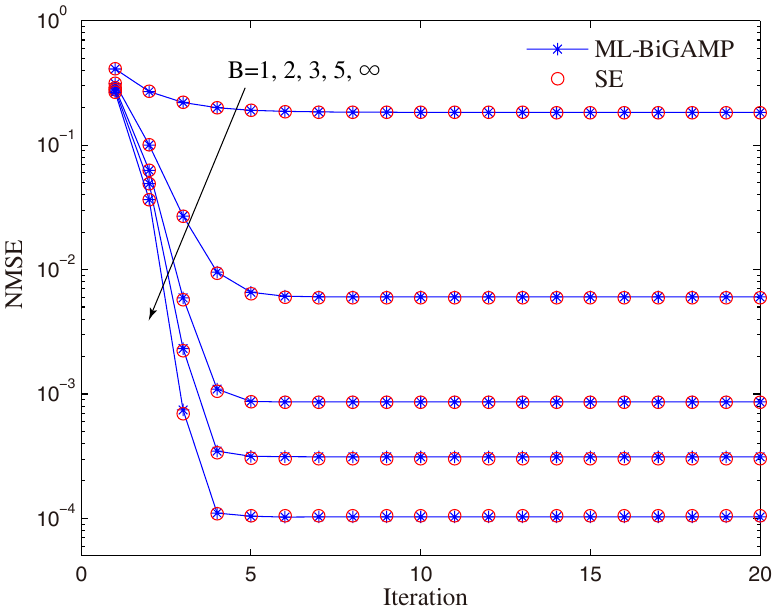}
\caption{Per-iteration behavior of ML-BiGAMP and its SE in communication ($N_2/N_1=2$).
}
\label{Fig:Exp6}
\end{figure}

\begin{figure}[!t]
\centering
\includegraphics[width=0.48\textwidth]{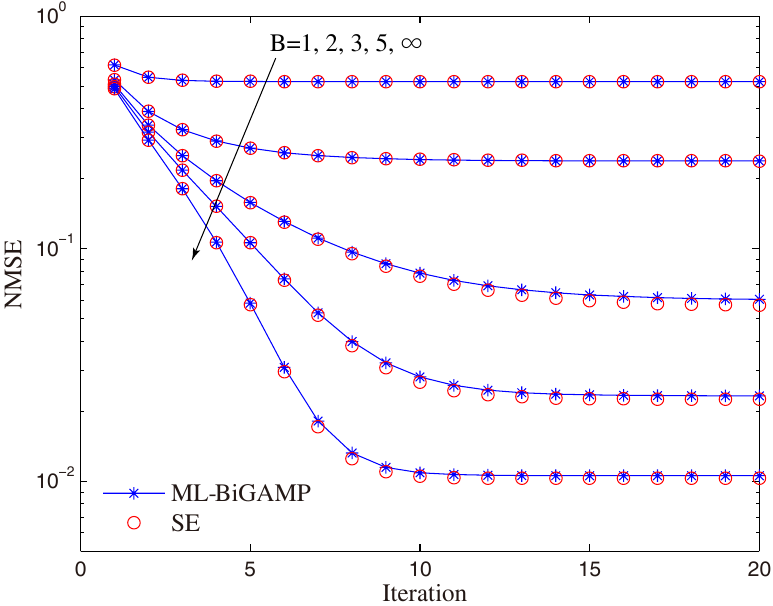}
\caption{Per-iteration behavior of ML-BiGAMP and its SE in communication ($N_2/N_1=1$).
}
\label{Fig:Exp7}
\end{figure}

\section{Conclusion}

In this paper, we studied the multi-layer generalized bilinear inference problem (\ref{alg:system}), where the goal is to recover each layer's input signal $\bs{X}^{(\ell)}$ and the measurement matrix $\bs{H}^{(\ell)}$ from the ultimate observation $\bs{Y}$. To this end, we have extended the BiG-AMP \cite{parker2014bilinear}, originally designed for a single layer, to develop a new algorithm termed multi-layer BiG-AMP (ML-BiGAMP). The new algorithm approximates the general sum-product LBP by performing AMP approximation in the high-dimensional limit and thus has a substantial reduction in its computational complexity as compared to competing methods. We also demonstrated that, in large system limit, the asymptotic MSE performance of ML-BiGAMP could be fully characterized via  its state evolution, i.e., a set of one-dimensional equations. The state evolution further revealed that its fixed point equations agreed perfectly with those of the exact MMSE estimator as predicted via the replica method. Given the fact that the MMSE estimator is optimal in MSE sense and that it is infeasible in high-dimensional practice, our ML-BiGAMP is attractive because it could achieve the same Bayes-optimal MSE performance with only a complexity of $\mathcal{O}(N^3)$. To illustrate the usefulness as well as to validate our theoretical analysis and prediction, we designed a new detector based on ML-BiGAMP that jointly estimates the channel fading and the data symbols with high precision, considering a two-hop AF relay communication system.

\begin{appendices}
\section{Derivation of ML-BiGAMP}
\label{Appendix:A}

\begin{figure}[!t]
\centering
\includegraphics[width=1\textwidth]{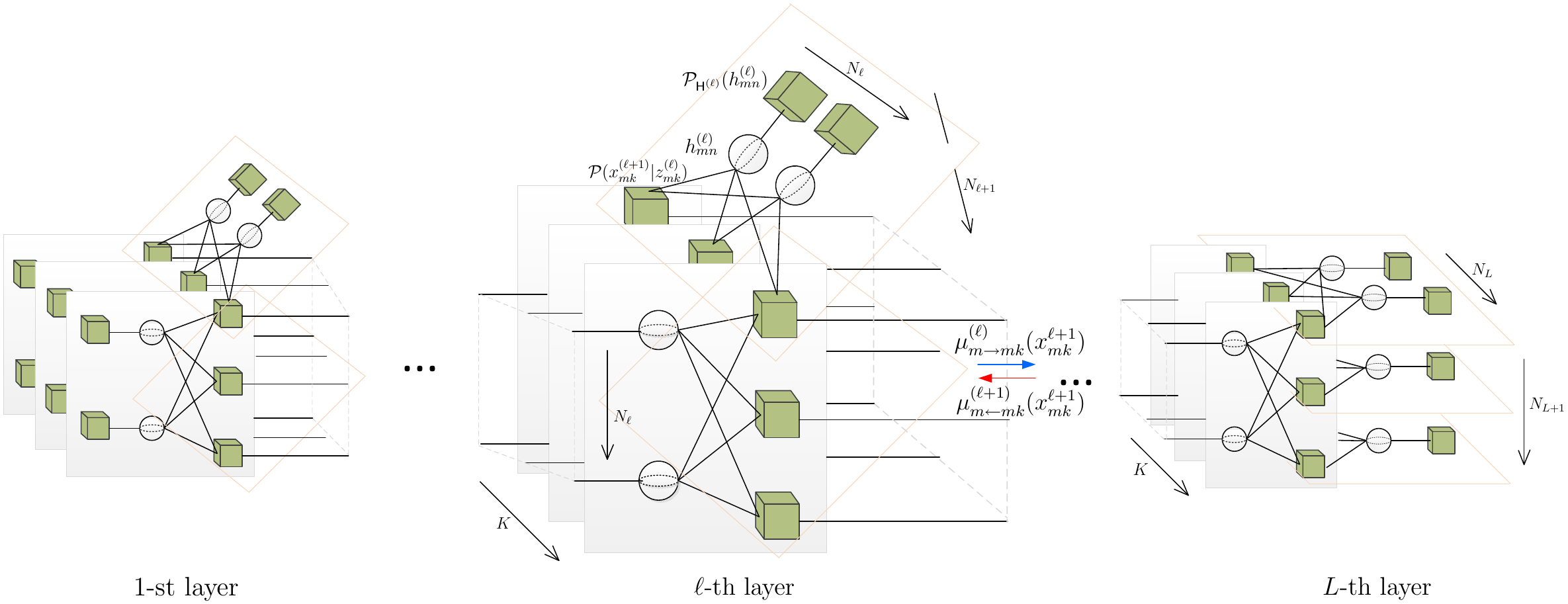}
\caption{The factor graph of multi-layer generalized bilinear inferences problem with unknown measurement matrices, where the cubes denote the factor nodes, the spheres denote the variable nodes, and the messages deliver via the edges between factor nodes and variable nodes.
}
\label{fig:FG}
\end{figure}

The factor graph of multi-layer generalized bilinear problems is presented in Fig.~\ref{fig:FG}. We then address the following messages defined in Table~\ref{Table:Definition}.
\begin{align}
\nonumber
\mu_{n\leftarrow mk}^{(\ell)}(x_{nk}^{(\ell)},t)&\propto \int \mathcal{P}\left(x_{mk}^{(\ell+1)}|\sum_{s=1}^{N_{\ell}}h_{ms}^{(\ell)}x_{sk}^{(\ell)}\right)\mu_{m\leftarrow mk}^{(\ell+1)}(x_{mk}^{(\ell+1)},t)\prod_{s=1}^{N_{\ell}}\mu_{k\leftarrow ms}^{(\ell)}(h_{ms}^{(\ell)},t),\\
&\qquad \quad  \times \prod_{r\ne n}^{N_{\ell}}\mu^{(\ell)}_{r\rightarrow mk}(x_{rk}^{(\ell)},t)\text{d}h^{(\ell)}_{ms}\text{d}x_{rk}^{(\ell)}\text{d}x_{mk}^{(\ell+1)},\\
\mu_{n\rightarrow  mk}^{(\ell)}(x_{nk}^{(\ell)},t+1)&\propto \mu_{n\rightarrow nk}^{(\ell-1)}(x_{nk}^{(\ell)},t+1)\prod_{s\ne m}^{N_{\ell+1}}\mu_{n\leftarrow sk}^{(\ell)}(x_{nk}^{(\ell)},t),\\
\nonumber
\mu_{k\rightarrow mn}^{(\ell)}(h_{mn}^{(\ell)},t)&\propto \int \mathcal{P}\left(x_{mk}^{(\ell+1)}|\sum_{s=1}^{N_{\ell}}h_{ms}^{(\ell)}x_{sk}^{(\ell)}\right) \mu_{m\leftarrow mk}^{(\ell+1)}(x_{mk}^{(\ell+1)},t)\prod_{r=1}^{N_{\ell}}\mu^{(\ell)}_{r\rightarrow mk}(x_{rk}^{(\ell)},t)\\
&\qquad \quad \times  \prod_{s\ne n}^{N_{\ell}}\mu_{k\leftarrow ms}^{(\ell)}(h_{ms}^{(\ell)},t)\text{d}h^{(\ell)}_{ms}\text{d}x_{rk}^{(\ell)}\text{d}x_{mk}^{(\ell+1)},\\
\mu_{k\leftarrow  mn}^{(\ell)}(h_{mn}^{(\ell)},t+1)&\propto \mathcal{P}(h_{mn}^{(\ell)})\prod_{s\ne k}^K \mu_{s\rightarrow mn}^{(\ell)}(h_{mn}^{(\ell)},t),
\end{align}
where
\begin{align}
\mu_{m\leftarrow mk}^{(\ell+1)}(x_{mk}^{(\ell+1)},t)&\propto  \prod_{p=1}^{N_{\ell+2}}\mu_{m\leftarrow pk}^{(\ell+1)}(x_{mk}^{(\ell+1)},t),
\label{Equ:Mu-mk}\\
\nonumber
\mu_{n\rightarrow nk}^{(\ell-1)}(x_{nk}^{(\ell)},t+1)&\propto  \int \mathcal{P}\left(x_{nk}^{(\ell)}|\sum_{s=1}^{N_{\ell-1}}h_{ns}x_{sk}\right)\prod_{s=1}^{N_{\ell-1}}\mu_{k\leftarrow ns}^{(\ell-1)}(h_{ns}^{(\ell-1)},t+1)\\
&\qquad \quad \times \prod_{r=1}^{N_{\ell-1}}\mu_{r\rightarrow nk}^{(\ell-1)}(x_{rk}^{(\ell-1)},t+1)\text{d}x_{rk}\text{d}h_{ns}.
\label{Equ:Mnnk}
\end{align}
Specially, when $\ell=L$, there is $\mu_{m\leftarrow mk}^{(L+1)}(x_{mk}^{(L+1)},t)=1$ whereas when $\ell=1$, we have $\mu_{n\rightarrow nk}^{(\ell-1)}(x_{nk}^{(\ell)},t)=\mathcal{P}(x_{nk})$.

\begin{table}[!t]
\centering
\caption{Sum-product Message definitions}
\label{Table:Definition}
\begin{tabular}{|l|l|l|l|l|}
  \hline
  $\mu_{n\leftarrow mk}^{(\ell)}(x_{nk}^{(\ell)},t)$        &message from $\mathcal{P}(x_{mk}^{(\ell+1)}|\cdot)$ to $x_{nk}$\\
  \hline
  $\mu_{n\rightarrow mk}^{(\ell)}(x_{nk}^{(\ell)},t)$       &message from  $x_{nk}$ to $\mathcal{P}(x_{mk}^{(\ell+1)}|\cdot)$\\
  \hline
  $\mu_{m\leftarrow mk}^{(\ell+1)}(x_{mk}^{(\ell+1)},t)$    &message from  $x_{mk}^{(\ell+1)}$ in $(\ell+1)$-th layer  to $\mathcal{P}(x_{mk}^{(\ell+1)}|\cdot)$\\
  \hline
  $\mu_{n\rightarrow nk}^{(\ell-1)}(x_{nk}^{(\ell)},t)$     &message from  $\mathcal{P}(x_{nk}^{(\ell)}|\cdot)$ in $(\ell-1)$-th layer  to $x_{nk}^{(\ell)}$\\
  \hline
  $\mu_{k\rightarrow mn}^{(\ell)}(h_{mn}^{(\ell)},t)$       &message from  $\mathcal{P}(x_{mk}^{(\ell+1)}|\cdot)$ to $h_{mn}^{(\ell)}$\\
  \hline
  $\mu_{k\leftarrow  mn}^{(\ell)}(h_{mn}^{(\ell)},t)$       &message from   $h_{mn}^{(\ell)}$ to $\mathcal{P}(x_{mk}^{(\ell+1)}|\cdot)$ \\
  \hline
  $\mu_{nk}^{(\ell)}(x_{nk}^{(\ell)},t)$                    &belief distribution at $x_{nk}$ \\
  \hline
  $\mu_{mn}^{(\ell)}(h_{mn}^{(\ell)},t)$                    &belief distribution at $h_{mn}^{(\ell)}$ \\
  \hline
\end{tabular}
\end{table}

Accordingly, the belief distributions (approximate posterior distribution) of $x_{nk}^{(\ell)}$ and $a_{mn}^{(\ell)}$ are respectively given by
\begin{align}
\mu_{nk}^{(\ell)}(x_{nk}^{(\ell)},t+1)&=\frac{\mu_{n\rightarrow nk}^{(\ell-1)}(x_{nk}^{(\ell)},t)\prod_{m=1}^{N_{\ell+1}}\mu_{n\leftarrow mk}^{(\ell)}(x_{nk}^{(\ell)},t)}{\int \mu_{n\rightarrow nk}^{(\ell-1)}(x_{nk}^{(\ell)},t)\prod_{m=1}^{N_{\ell+1}}\mu_{n\leftarrow mk}^{(\ell)}(x_{nk}^{(\ell)},t)\text{d}x_{nk}^{(\ell)}},\\
\mu_{mn}^{(\ell)}(h_{mn}^{(\ell)},t+1)&=\frac{\mathcal{P}_{\textsf{H}^{(\ell)}}(h_{mn}^{(\ell)})\prod_{k=1}^K \mu_{k\rightarrow mn}^{(\ell)}(h_{mn}^{(\ell)},t)}{\int \mathcal{P}_{\textsf{H}^{(\ell)}}(h_{mn}^{(\ell)})\prod_{k=1}^K \mu_{k\rightarrow mn}^{(\ell)}(h_{mn}^{(\ell)},t)\text{d}h_{mn}^{(\ell)}}.
\end{align}
We denote the mean and variance of $\mu_{nk}^{(\ell)}(x_{nk}^{(\ell)},t)$ as $\hat{x}_{nk}^{(\ell)}(t)$ and $v^{(x,\ell)}_{nk}(t)$ respectively. Meanwhile, we denote the mean and variance of $\mu_{mn}^{(\ell)}(h_{mn}^{(\ell)},t)$ as $\hat{h}_{mn}^{(\ell)}$ and $v_{mn}^{(h,\ell)}(t)$, respectively. Note that $\hat{x}_{nk}(t)$ and $\hat{h}_{mn}^{(\ell)}(t)$ are the approximate MMSE estimators of $x_{nk}$ and $h_{mn}$ in $t$-th iteration, respectively.

\subsection{Approximate factor-to-variable messages}
We begin at simplifying the factor-to-variable message $\mu_{n\leftarrow mk}^{(\ell)}(x_{nk}^{(\ell)},t)$
\begin{align}
\nonumber
\mu_{n\leftarrow mk}^{(\ell)}(x_{nk}^{(\ell)},t)
&\propto \int \mathcal{P}\left(x_{mk}^{(\ell+1)}|z_{mk}^{(\ell)}\right) \mathbb{E}\left[\delta \left(z_{mk}^{(\ell)}-h_{mn}^{(\ell)}x_{nk}^{(\ell)}-\sum_{s\ne n}^{N_{\ell}} h_{ms}^{(\ell)}x_{sk}^{(\ell)}\right)\right]\\
&\qquad \times \mu_{m\leftarrow mk}^{(\ell+1)}(x_{mk}^{(\ell+1)},t)\text{d}z_{mk}^{(\ell)}\text{d}x_{mk}^{(\ell+1)},
\end{align}
where the expectation is taken over the distribution $\prod_{s=1}^{N_{\ell}}\mu_{k\leftarrow ms}^{(\ell)}(h_{ms}^{(\ell)},t) \prod_{r\ne n}^{N_{\ell}}\mu^{(\ell)}_{r\rightarrow mk}(x_{rk}^{(\ell)},t)$. We associate random variable (RV) $\xi_{mk}^{(z,\ell)}(t)$ with $z_{mk}^{(\ell)}$, associate RV $\xi^{(h,\ell)}_{k\leftarrow ms}(t)$ with $h_{ms}^{(\ell)}$ following $\mu_{k\leftarrow ms}^{(\ell)}(h_{ms}^{(\ell)},t)$, and associate RV $\xi_{r\rightarrow mk}^{(x,\ell)}(t)$ with $x_{rk}^{(\ell)}$ following $\mu^{(\ell)}_{r\rightarrow mk}(x_{rk}^{(\ell)},t)$. Then applying PDF-to-RV lemma \footnote{
Let ${\boldsymbol{w}}\in \mathbb{R}^p$ and ${u} \in \mathbb{R}^1$ be two RVs, and $g:\mathbb{R}^p\rightarrow \mathbb{R}$ be a generic mapping. Then, ${u}=g({\boldsymbol{w}})$ if and only if the PDF $\mathcal{P}_{{u}}(u)\propto \int \delta(u-g(\boldsymbol{w}))\mathcal{P}_{{\boldsymbol{w}}}(\boldsymbol{w})\text{d}\boldsymbol{w}$.}  yields
\begin{align}
\xi_{mk}^{(z,\ell)}(t)=x_{nk}^{(\ell)}\xi^{(h,\ell)}_{k\leftarrow mn}(t)+\sum_{s\ne n}\xi^{(h,\ell)}_{k\leftarrow ms}(t)\xi_{s\rightarrow mk}^{(x,\ell)}(t).
\end{align}
In large system limits,  the central limit theorem (CLT) allows us to handle $\xi_{mk}^{(z,\ell)}(t)$ as Gaussian distribution with mean and variance respectively given by
\begin{align}
\mathbb{E}[\xi_{mk}^{(z,\ell)}(t)]&=x_{nk}^{(\ell)}\hat{h}_{k\leftarrow mn}^{(\ell)}(t)+Z_{mk\backslash n}^{(\ell)}(t),\\
\text{Var}[\xi_{mk}^{(z,\ell)}(t)]&=|x_{nk}^{(\ell)}|^2v_{k\leftarrow mn}^{(h,\ell)}(t)+V_{mk\backslash n}^{(\ell)}(t),
\end{align}
where
\begin{align}
Z_{mk\backslash n}^{(\ell)}(t)&=\sum_{s\ne n}\hat{h}_{k\leftarrow ms}^{(\ell)}(t)\hat{x}_{s\rightarrow mk}^{(\ell)}(t),\\
V_{mk\backslash n}^{(\ell)}(t)&=\sum_{s\ne n}v_{k\leftarrow ms}^{(h,\ell)}(t)v_{s\rightarrow mk}^{(x,\ell)}(t)+|\hat{h}_{k\leftarrow ms}^{(\ell)}(t)|^2v_{s\rightarrow mk}^{(x,\ell)}(t)+|\hat{x}_{s\rightarrow mk}^{(\ell)}(t)|^2v_{k\leftarrow ms}^{(h,\ell)}(t),
\end{align}
with $\hat{h}_{k\leftarrow ms}^{(\ell)}(t)$ and $v_{k\leftarrow ms}^{(h,\ell)}(t)$ being the mean and variance of RV $\xi^{(h,\ell)}_{k\leftarrow ms}(t)$, respectively, and $\hat{x}_{s\rightarrow mk}^{(\ell)}(t)$ and $v_{s\rightarrow mk}^{(x,\ell)}(t)$ being the mean and variance of RV $\xi_{s\rightarrow mk}^{(x,\ell)}(t)$, respectively.

By Gaussian approximation, the message $\mu_{n\leftarrow mk}^{(\ell)}(x_{nk}^{(\ell)},t)$ is simplified as
\begin{align}
\nonumber
\mu_{n\leftarrow mk}^{(\ell)}(x_{nk}^{(\ell)},t)
&\propto \int  \mathcal{N}\left(z_{mk}^{(\ell)}|x_{nk}^{(\ell)}\hat{h}_{k\leftarrow mn}^{(\ell)}(t)+Z_{mk\backslash n}^{(\ell)}(t), |x_{nk}^{(\ell)}|^2v_{k\leftarrow mn}^{(h,\ell)}(t)+V_{mk\backslash n}^{(\ell)}(t)\right)\\
&\qquad \times \mathcal{P}\left(x_{mk}^{(\ell+1)}|z_{mk}^{(\ell)}\right)\mu_{m\leftarrow mk}^{(\ell+1)}(x_{mk}^{(\ell+1)},t)\text{d}z_{mk}^{(\ell)}\text{d}x_{mk}^{(\ell+1)}.
\end{align}
It is found that the parameters $Z_{mk\backslash n}^{(\ell)}(t)$ only has a slight differ from each others. The similar situation also exists in the parameter $V_{mk\backslash n}^{(\ell)}(t)$. To further simplify the message $\mu_{n\leftarrow mk}^{(\ell)}(x_{nk}^{(\ell)},t)$, we define
\begin{align}
H_{mk}^{(\ell)}(a,A)
&=\log  \int \mathcal{P}\left(x_{mk}^{(\ell+1)}|z_{mk}^{(\ell)}\right)\mathcal{N}(z_{mk}^{(\ell)}|a,A)\mu_{m\leftarrow mk}^{(\ell+1)}(x_{mk}^{(\ell+1)},t)\text{d}x_{mk}^{(\ell+1)},\\
Z_{mk}^{(\ell)}(t)&=\sum_{n=1}^{N_{\ell}}\hat{h}_{k\leftarrow mn}^{(\ell)}(t)\hat{x}_{n\rightarrow mk}^{(\ell)}(t)
\label{Equ:Z},\\
V_{mk}^{(\ell)}(t)&=\sum_{n=1}^{N_{\ell}}v_{k\leftarrow mn}^{(h,\ell)}(t)v_{n\rightarrow mk}^{(x,\ell)}(t)+|\hat{h}_{k\leftarrow mn}^{(\ell)}(t)|^2v_{n\rightarrow mk}^{(x,\ell)}(t)+|\hat{x}_{n\rightarrow mk}^{(\ell)}(t)|^2v_{k\leftarrow mn}^{(h,\ell)}(t),
\label{Equ:V}
\end{align}
and then obtain
\begin{align}
\nonumber
&\log \mu_{n\leftarrow mk}^{(\ell)}(x_{nk}^{(\ell)},t)\\
&=\text{const}+H_{mk}^{(\ell)}\left(x_{nk}^{(\ell)}\hat{h}_{k\leftarrow mn}^{(\ell)}(t)+Z_{mk\backslash n}^{(\ell)}(t), |x_{nk}^{(\ell)}|^2v_{k\leftarrow mn}^{(h,\ell)}(t)+V_{mk\backslash n}^{(\ell)}(t)\right)\\
\nonumber
&=\text{const}+H_{mk}^{(\ell)}\left(Z_{mk}^{(\ell)}(t)+\hat{h}_{k\leftarrow mn}^{(\ell)}(t)(x_{nk}^{(\ell)}-\hat{x}_{n\rightarrow mk}^{(\ell)}(t))\right., \\ &\left.V_{mk}^{(\ell)}(t)+v_{k\leftarrow mn}^{(h,\ell)}(t)(|x_{nk}^{(\ell)}|^2-|\hat{x}_{n\rightarrow mk}^{(\ell)}(t)|^2)-v_{k\leftarrow mn}^{(h,\ell)}(t)v_{n\rightarrow mk}^{(x,\ell)}(t)-|\hat{h}_{k\leftarrow mn}^{(\ell)}(t)|^2v_{n\rightarrow mk}^{(x,\ell)}(t)\right)\\
&\approx \text{const}+H_{mk}^{(\ell)}\left(Z_{mk}^{(\ell)}(t)+\hat{h}_{k\leftarrow mn}^{(\ell)}(t)(x_{nk}^{(\ell)}-\hat{x}_{nk}^{(\ell)}(t)),V_{mk}^{(\ell)}(t)+v_{k\leftarrow mn}^{(h,\ell)}(t)(|x_{nk}^{(\ell)}|^2-|\hat{x}_{n\rightarrow mk}^{(\ell)}(t)|^2)\right),
\end{align}
where we use $\hat{x}_{nk}^{(\ell)}(t)$ to replace $\hat{x}_{n\rightarrow mk}^{(\ell)}(t)$, since $\mu_{nk}^{(\ell)}(x_{nk}^{(\ell)},t)$ is slightly different from $\mu_{n\rightarrow  mk}^{(\ell)}(x_{nk}^{(\ell)},t)$ and further $\hat{x}_{n\rightarrow mk}^{(\ell)}(t)$ has the same order as $\hat{x}_{nk}^{(\ell)}(t)$. Besides, the item $v_{k\leftarrow mn}^{(h,\ell)}(t)v_{n\rightarrow mk}^{(x,\ell)}(t)+|\hat{h}_{k\leftarrow mn}^{(\ell)}(t)|^2v_{n\rightarrow mk}^{(x,\ell)}(t)$ is ignored due to infinitesimal items $v_{k\leftarrow mn}^{(h,\ell)}(t)v_{n\rightarrow mk}^{(x,\ell)}(t)$, $|\hat{h}_{k\leftarrow mn}^{(\ell)}(t)|^2$. The remaining variance entries are found in Table~\ref{Table:oder}.

\begin{table}[!t]
\centering
\caption{ML-BiGAMP variable scalings in the large system limit \cite{parker2014bilinear}}
\label{Table:oder}
\begin{tabular}{|l|l||l|l||l|l||}
  \hline
  $\tilde{z}_{mk}^{(\ell)}(t)$                & $\mathcal{O}(1)$ & $\tilde{v}^{(z,\ell)}_{mk}(t)$    &$\mathcal{O}(1)$  & $\hat{x}_{n\rightarrow mk}^{(\ell)}(t)-\hat{x}_{nk}^{(\ell)}(t)$ & $\mathcal{O}(\frac{1}{\sqrt{N_{\ell}}})$ \\
  $\hat{x}_{n\rightarrow mk}^{(\ell)}(t)$   & $\mathcal{O}(1)$ & $v^{(x,\ell)}_{n\rightarrow mk}(t)$ &$\mathcal{O}(1)$  & $|\hat{x}_{n\rightarrow mk}^{(\ell)}(t)|^2-|\hat{x}_{nk}^{(\ell)}(t)|^2$ & $\mathcal{O}(\frac{1}{\sqrt{N_{\ell}}})$ \\
  $\hat{x}_{nk}^{(\ell)}(t)$   & $\mathcal{O}(1)$ & $v_{nk}^{(x,\ell)}(t)$ &$\mathcal{O}(1)$  & $v^{(x,\ell)}_{n\rightarrow mk}(t)-v^{(x,\ell)}_{nk}(t)$ & $\mathcal{O}(\frac{1}{\sqrt{N_{\ell}}})$ \\
  $\hat{h}_{k\leftarrow mn}^{(\ell)}(t)$   & $\mathcal{O}(\frac{1}{\sqrt{N_{\ell}}})$ & $v_{k\leftarrow mn}^{(h,\ell)}(t)$ &$\mathcal{O}(\frac{1}{N_{\ell}})$  & $\hat{h}_{k\leftarrow mn}^{(\ell)}(t)-\hat{h}_{mn}^{(\ell)}(t)$ & $\mathcal{O}(\frac{1}{N_{\ell}})$ \\
  $\hat{h}_{ mn}^{(\ell)}(t)$   & $\mathcal{O}(\frac{1}{\sqrt{N_{\ell}}})$ & $v_{mn}^{(h,\ell)}(t)$ &$\mathcal{O}(\frac{1}{N_{\ell}})$  & $|\hat{h}_{k\leftarrow mn}^{(\ell)}(t)|^2-|\hat{h}_{mn}^{(\ell)}(t)|^2$ & $\mathcal{O}(\frac{1}{(N_{\ell})^{3/2}})$ \\
  $Z_{mk}^{(\ell)}(t)$   & $\mathcal{O}(1)$ & $V_{mk}^{(\ell)}(t)$ &$\mathcal{O}(1)$  & $v_{k\leftarrow mn}^{(h,\ell)}(t)-v_{mn}^{(h,\ell)}(t)$ & $\mathcal{O}(\frac{1}{(N_{\ell})^{3/2}})$ \\
  $R_{nk}^{(x,\ell)}(t)$   & $\mathcal{O}(1)$ & $\Sigma_{nk}^{(x,\ell)}(t)$ &$\mathcal{O}(1)$  &  &  \\
  $R_{mn}^{(h,\ell)}(t)$   & $\mathcal{O}(\frac{1}{\sqrt{N_{\ell}}})$ & $\Sigma_{mn}^{(h,\ell)}(t)$ &$\mathcal{O}(\frac{1}{N_{\ell}})$  & & \\
  $\hat{s}_{mk}^{(\ell)}(t)$     & $\mathcal{O}(1)$ & $v_{mk}^{(s,\ell)}(t)$ &$\mathcal{O}(1)$  & & \\
  \hline
\end{tabular}
\end{table}

We further apply Taylor series expansion\footnote{
$f(x+\triangle x,y+\triangle y)\approx f(x,y)+\triangle xf'(x,y)+\triangle y\dot{f}(x,y)+\frac{|\triangle x|^2}{2}f''(x,y)+o.$
} to logarithm of message $\mu_{n\leftarrow mk}^{(\ell)}(x_{nk}^{(\ell)},t)$
\begin{align}
\nonumber
&\log \mu_{n\leftarrow mk}^{(\ell)}(x_{nk}^{(\ell)},t)\\
\nonumber
&\approx \text{const}+H_{mk}^{(\ell)}\left(Z_{mk}^{(\ell)}(t),V_{mk}^{(\ell)}(t)\right)\\
\nonumber
&\quad +\hat{h}_{k\leftarrow mn}^{(\ell)}(t)(x_{nk}^{(\ell)}-\hat{x}_{nk}^{(\ell)}(t))H_{mk}'^{(\ell)}\left(Z_{mk}^{(\ell)}(t),V_{mk}^{(\ell)}(t)\right)\\
\nonumber
&\quad +\frac{|\hat{h}_{k\leftarrow mn}^{(\ell)}(t)|^2|x_{nk}^{(\ell)}-\hat{x}_{nk}^{(\ell)}(t)|^2}{2}H_{mk}''^{(\ell)}\left(Z_{mk}^{(\ell)}(t),V_{mk}^{(\ell)}(t)\right)\\
&\quad +v_{k\leftarrow mn}^{(h,\ell)}(t)(|x_{nk}^{(\ell)}|^2-|\hat{x}_{n\rightarrow mk}^{(\ell)}(t)|^2)\dot{H}_{mk}^{(\ell)}\left(Z_{mk}^{(\ell)}(t),V_{mk}^{(\ell)}(t)\right)\\
\nonumber
&=\text{const}+x_{nk}^{(\ell)}\left[\hat{h}_{k\leftarrow mn}^{(\ell)}(t)H_{mk}'^{(\ell)}\left(Z_{mk}^{(\ell)}(t),V_{mk}^{(\ell)}(t)\right)+|\hat{h}_{mn}^{(\ell)}(t)|^2\hat{x}_{nk}^{(\ell)}(t)H_{mk}''^{(\ell)}\left(Z_{mk}^{(\ell)}(t),V_{mk}^{(\ell)}(t)\right)\right]\\
&\quad +|x_{nk}^{(\ell)}|^2\left[
\frac{1}{2}|\hat{h}_{ mn}^{(\ell)}(t)|^2H_{mk}''^{(\ell)}\left(Z_{mk}^{(\ell)}(t),V_{mk}^{(\ell)}(t)\right)+
v_{mn}^{(h,\ell)}(t)\dot{H}_{mk}^{(\ell)}\left(Z_{mk}^{(\ell)}(t),V_{mk}^{(\ell)}(t)\right)\right],
\end{align}
where $H_{mk}'^{(\ell)}(\cdot)$ and $H_{mk}''^{(\ell)}(\cdot)$ are first and second order  partial derivation w.r.t. first argument and $\dot{H}_{mk}^{(\ell)}(\cdot)$ is first order partial derivation w.r.t. second argument.

With the facts\footnote{
Defining the mean and variance of distribution $\mathcal{P}(x|a,A)=\frac{f(x)\mathcal{N}(x|a,A)}{\int f(x)\mathcal{N}(x|a,A)\text{d}x}$ as $\mathbb{E}[x]$ and $\text{Var}[x]$, where $f(x)$ is bound and non-negative  function, we have$
\frac{\partial \log f(x)\mathcal{N}(x|a,A)}{\partial a}=\frac{\mathbb{E}[x]-a}{A}
, \ \frac{\partial^2 \log f(x)\mathcal{N}(x|a,A)}{\partial a^2}=\frac{\text{Var}[x]-A}{A^2}$, and $
\frac{\partial \log f(x)\mathcal{N}(x|a,A)}{\partial A}=\frac{1}{2}\left[\left|\frac{\partial \log \mathcal{P}(x|a,A)}{\partial a}\right|^2+\frac{\partial^2\log \mathcal{P}(x|a,A)}{\partial a^2}\right]
$
}
, the message $\mu_{n\leftarrow mk}^{(\ell)}(x_{nk}^{(\ell)},t)$ is approximated by following Gaussian distribution
\begin{align}
\nonumber
\mu_{n\leftarrow mk}^{(\ell)}(x_{nk}^{(\ell)},t)
&\approx \mathcal{N}_c\left(x_{n\ell}^{(\ell)}|
\frac{\hat{h}_{k\leftarrow mn}^{(\ell)}(t)\hat{s}_{m\ell}(t)+|\hat{h}_{mn}^{(\ell)}(t)|^2\hat{x}_{nk}^{(\ell)}(t)v^{(s,\ell)}_{mk}(t)}
{|\hat{h}_{mn}^{(\ell)}(t)|^2v^{(s,\ell)}_{mk}(t)-v_{mn}^{(h,\ell)}(t)(|\hat{s}_{mk}^{(\ell)}(t)|^2-v^{(s,\ell)}_{mk}(t))},
\right.\\
&\qquad \qquad \quad  \left.\frac{1}{|\hat{h}_{mn}^{(\ell)}(t)|^2v^{(s,\ell)}_{mk}(t)-v_{mn}^{(h,\ell)}(t)(|\hat{s}_{mk}^{(\ell)}(t)|^2-v^{(s,\ell)}_{mk}(t))}\right),
\end{align}
where
\begin{align}
\hat{s}_{mk}^{(\ell)}(t)&=H_{mk}'^{(\ell)}\left(Z_{mk}^{(\ell)}(t),V_{mk}^{(\ell)}(t)\right)=\frac{\tilde{z}_{mk}^{(\ell)}(t)-Z_{mk}^{(\ell)}(t)}{V_{mk}^{(\ell)}(t)},\\
v^{(s,\ell)}_{mk}(t)&=-H_{mk}''^{(\ell)}\left(Z_{mk}^{(\ell)}(t),V_{mk}^{(\ell)}(t)\right)=\frac{1}{V_{mk}^{(\ell)}}\left(1-\frac{\tilde{v}_{mk}^{(\ell)}(t)}{V_{mk}^{(\ell)}(t)}\right),
\end{align}
with $\tilde{z}_{mk}^{(\ell)}(t)$ and $\tilde{v}_{mk}^{(\ell)}(t)$ defined as the mean and variance of random variable (RV) $\zeta_{mk}^{(\ell)}(t)$ drawn by
\begin{align}
\zeta_{mk}^{(\ell)}(t) \sim \frac{\int \mathcal{P}\left(x_{mk}^{(\ell+1)}|z_{mk}^{(\ell)}\right)\mathcal{N}(z_{mk}^{(\ell)}|Z_{mk}^{(\ell)}(t),V_{mk}^{(\ell)}(t))\mu_{m\leftarrow mk}^{(\ell+1)}(x_{mk}^{(\ell+1)},t)\text{d}x_{mk}^{(\ell+1)}}
{ \int \mathcal{P}\left(x_{mk}^{(\ell+1)}|z_{mk}^{(\ell)}\right)\mathcal{N}(z_{mk}^{(\ell)}|Z_{mk}^{(\ell)}(t),V_{mk}^{(\ell)}(t))\mu_{m\leftarrow mk}^{(\ell+1)}(x_{mk}^{(\ell+1)},t)\text{d}x_{mk}^{(\ell+1)}\text{d}z_{mk}^{(\ell)}}.
\label{Equ:posz}
\end{align}

Note that the message $\mu_{m\leftarrow mk}^{(\ell+1)}(x_{mk}^{(\ell+1)},t)$ in (\ref{Equ:Mu-mk}) is the  product of a large number of Gaussian distributions. Based on the Gaussian reproduction property\footnote{
$\mathcal{N}(x|a,A)\mathcal{N}(x|b,B)=\mathcal{N}(0|a-b,A+B)\mathcal{N}(x|c,C)$ with $C=(A^{-1}+B^{-1})^{-1}$ and $c=C\cdot(\frac{a}{A}+\frac{b}{B})$.
}, we obtain
\begin{align}
\mu_{m\leftarrow mk}^{(\ell+1)}(x_{mk}^{(\ell+1)},t)&\propto  \mathcal{N}_c(x_{mk}^{(\ell+1)}|R_{mk}^{(x,\ell+1)}(t),\Sigma_{mk}^{(x,\ell+1)}(t)),
\end{align}
where
\begin{align}
\Sigma_{mk}^{(x,\ell+1)}(t)&=\left(\sum_{p=1}^{N_{\ell+2}}\frac{1}{v^{(x,\ell+1)}_{m\leftarrow pk}(t)}\right)^{-1},\\
R_{mk}^{(x,\ell+1)}(t)&=\Sigma_{mk}^{(x,\ell+1)}(t)\left(\sum_{p=1}^{N_{\ell+2}}\frac{\hat{x}_{m\leftarrow pk}^{(\ell+1)}(t)}{v^{(x,\ell+1)}_{m\leftarrow pk}(t)}\right),
\end{align}
with $\hat{x}_{m\leftarrow pk}^{(\ell)}(t)$ and $v^{(x,\ell)}_{m\leftarrow pk}(t)$ being the mean and variance of $\mu^{(\ell+1)}_{m\leftarrow pk}(x_{mk}^{(\ell+1)},t)$ respectively.

We then update the expression of $\tilde{z}_{mk}^{(\ell)}(t)$ and $\tilde{v}_{mk}^{(\ell)}(t)$
\begin{align}
\tilde{z}_{mk}^{(\ell)}(t)&=\mathbb{E}\left[\zeta_{mk}^{(\ell)}(t)\right],\\
\tilde{v}_{mk}^{(\ell)}(t)&=\text{Var}\left[\zeta_{mk}^{(\ell)}(t)\right],
\end{align}
where the expectation is taken over
\begin{align}
\zeta_{mk}^{(\ell)}(t)\sim
\frac{\int \mathcal{P}\left(x_{mk}^{(\ell+1)}|z_{mk}^{(\ell)}\right)\mathcal{N}(z_{mk}^{(\ell)}|Z_{mk}^{(\ell)}(t),V_{mk}^{(\ell)}(t))\mathcal{N}_c(x_{mk}^{(\ell+1)}|R_{mk}^{(x,\ell+1)}(t),\Sigma_{mk}^{(x,\ell+1)}(t))\text{d}x_{mk}^{(\ell+1)}}
{ \int \mathcal{P}\left(x_{mk}^{(\ell+1)}|z_{mk}^{(\ell)}\right)\mathcal{N}(z_{mk}^{(\ell)}|Z_{mk}^{(\ell)}(t),V_{mk}^{(\ell)}(t))\mathcal{N}_c(x_{mk}^{(\ell+1)}|R_{mk}^{(x,\ell+1)}(t),\Sigma_{mk}^{(x,\ell+1)}(t))\text{d}x_{mk}^{(\ell+1)}\text{d}z_{mk}^{(\ell)}}.
\label{Pz_A}
\end{align}
Specially, as $\ell=L$, we have $\mu_{m\leftarrow mk}^{(\ell+1)}(x_{mk}^{(\ell+1)},t)=1$ and further
\begin{align}
\zeta_{mk}^{(\ell)}(t)\sim \frac{ \mathcal{P}\left(y_{mk}|z_{mk}^{(\ell)}\right)\mathcal{N}(z_{mk}^{(\ell)}|Z_{mk}^{(\ell)}(t),V_{mk}^{(\ell)}(t))}
{ \int \mathcal{P}\left(y_{mk}|z_{mk}^{(\ell)}\right)\mathcal{N}(z_{mk}^{(\ell)}|Z_{mk}^{(\ell)}(t),V_{mk}^{(\ell)}(t))\text{d}z_{mk}^{(\ell)}},
\label{Pz_B}
\end{align}

Similar to simplifying $\mu_{n\leftarrow mk}^{(\ell)}(x_{nk}^{(\ell)},t)$, we approximate the message $\mu_{k\rightarrow mn}^{(\ell)}(h_{mn}^{(\ell)},t)$ as below
\begin{align}
\nonumber
\mu_{k\rightarrow mn}^{(\ell)}(h_{mn}^{(\ell)},t)
&\approx  \mathcal{N}_c\left(h_{mn}^{(\ell)}|
\frac{\hat{x}_{n\rightarrow  mk}^{(\ell)}(t)\hat{s}_{m\ell}(t)+|\hat{x}_{nk}^{(\ell)}(t)|^2\hat{h}_{mn}^{(\ell)}(t)v^{(s,\ell)}_{mk}(t)}
{|\hat{x}_{nk}^{(\ell)}(t)|^2v_{mk}^{(s,\ell)}(t)-v_{nk}^{(x,\ell)}(t)(|\hat{s}_{mk}^{(\ell)}(t)|^2-v^{(s,\ell)}_{mk}(t))},
\right.\\
&\qquad \qquad \quad  \left.\frac{1}{|\hat{x}_{nk}^{(\ell)}(t)|^2v_{mk}^{(s,\ell)}(t)-v_{nk}^{(x,\ell)}(t)(|\hat{s}_{mk}^{(\ell)}(t)|^2-v^{(s,\ell)}_{mk}(t))}\right).
\end{align}

For message $\mu_{n\rightarrow nk}^{(\ell-1)}(x_{nk}^{(\ell)},t+1)$ in (\ref{Equ:Mnnk}), we have
\begin{align}
\mu_{n\rightarrow nk}^{(\ell-1)}(x_{nk}^{(\ell)},t+1)&\propto \int \mathcal{P}\left(x_{nk}^{(\ell)}|z_{nk}^{(\ell-1)}\right)
\mathbb{E}\left[\delta\left(z_{nk}^{(\ell-1)}-\sum_{r=1}^{N_{\ell-1}}h_{nr}^{(\ell-1)}x_{rk}^{(\ell-1)}\right)\right]\text{d}z_{nk}^{(\ell-1)},
\end{align}
with expectation over $\prod_{r=1}^{N_{\ell-1}}\mu_{k\leftarrow nr}^{(\ell-1)}(h_{nr}^{(\ell-1)},t+1)\prod_{r=1}^{N_{\ell-1}}\mu_{r\rightarrow nk}^{(\ell-1)}(x_{rk}^{(\ell-1)},t+1)$. Applying PDF-to-RV lemma and CLT, we get
\begin{align}
\mu_{n\rightarrow nk}^{(\ell-1)}(x_{nk}^{(\ell)},t+1)\approx \int \mathcal{P}\left(x_{nk}^{(\ell)}|z_{nk}^{(\ell-1)}\right)\mathcal{N}_c(z_{nk}^{(\ell-1)}|Z_{nk}^{(\ell-1)}(t+1),V_{nk}^{(\ell-1)}(t+1))\text{d}z_{nk}^{(\ell-1)},
\end{align}
where the definitions of $Z_{nk}^{(\ell-1)}(t+1)$ and $V_{nk}^{(\ell-1)}(t+1)$ are found in (\ref{Equ:Z}) and (\ref{Equ:V}) respectively.

\subsection{Approximate variable-to-factor node messages}
We now move to the simplifying of messages from variable node to factor node. By Gaussian reproduction lemma, the Gaussian product item in message $\mu_{n\rightarrow  mk}^{(\ell)}(x_{nk}^{(\ell)},t+1)$ is as blow
\begin{align}
\prod_{s\ne m}^{N_{\ell+1}}\mu_{n\leftarrow sk}^{(\ell)}(x_{nk}^{(\ell)},t)\propto \mathcal{N}\left(x_{nk}|R_{nk\backslash m}^{(x,\ell)}(t),\Sigma_{nk\backslash m}^{(x,\ell)}(t)\right),
\end{align}
where
\begin{align}
\Sigma_{nk\backslash m}^{(x,\ell)}(t)
&=\left(\sum_{r\ne m}\frac{1}{v_{n\leftarrow rk}^{(x,\ell)}(t)}\right)^{-1}\\
&=\left(\sum_{r\ne m}|\hat{h}_{rn}^{(\ell)}(t)|^2v_{rk}^{(s,\ell)}(t)-v_{rn}^{(h,\ell)}(t)(|\hat{s}_{rk}^{(\ell)}(t)|^2-v^{(s,\ell)}_{rk}(t))\right)^{-1},\\
R_{nk\backslash m}^{(x,\ell)}(t)
&=\Sigma_{nk\backslash m}^{(x,\ell)}(t)\left(\sum_{r\ne m}\frac{\hat{x}_{n\leftarrow rk}^{(\ell)}(t)}{v_{n\leftarrow rk}^{(x,\ell)}(t)}\right)\\
&=\frac{\sum_{r\ne m}\hat{h}_{k\leftarrow rn}^{(\ell)}(t)\hat{s}_{r\ell}(t)+|\hat{h}_{rn}^{(\ell)}(t)|^2\hat{x}_{nk}^{(\ell)}(t)v^{(s,\ell)}_{rk}(t)}
{\sum_{r\ne m}|\hat{h}_{rn}^{(\ell)}(t)|^2v_{rk}^{(s,\ell)}(t)-v_{rn}^{(h,\ell)}(t)(|\hat{s}_{rk}^{(\ell)}(t)|^2-v^{(s,\ell)}_{rk}(t))}\\
\nonumber
&=\hat{x}_{nk}^{(\ell)}(t)\frac{\sum_{r\ne m}|\hat{h}_{rn}^{(\ell)}(t)|^2v^{(s,\ell)}_{rk}(t)}
{\sum_{r\ne m}|\hat{h}_{rn}^{(\ell)}(t)|^2v_{rk}^{(s,\ell)}(t)-v_{rn}^{(h,\ell)}(t)(|\hat{s}_{rk}^{(\ell)}(t)|^2-v^{(s,\ell)}_{rk}(t))}\\
&\qquad +\frac{\sum_{r\ne m}\hat{h}_{k\leftarrow rn}^{(\ell)}(t)\hat{s}_{r\ell}(t)}
{\sum_{r\ne m}|\hat{h}_{rn}^{(\ell)}(t)|^2v_{rk}^{(s,\ell)}(t)-v_{rn}^{(h,\ell)}(t)(|\hat{s}_{rk}^{(\ell)}(t)|^2-v^{(s,\ell)}_{rk}(t))}\\
&=\hat{x}_{nk}^{(\ell)}(t)\left[1+\Sigma_{nk\backslash m}^{(x,\ell)}(t)\sum_{r\ne m}v_{rn}^{(h,\ell)}(t)(|\hat{s}_{rk}^{(\ell)}(t)|^2-v^{(s,\ell)}_{rk}(t))\right]+\Sigma_{nk\backslash m}^{(x,\ell)}(t)\sum_{r\ne m}\hat{h}_{k\leftarrow rn}^{(\ell)}(t)\hat{s}_{rk}(t).
\end{align}

For easy of notation, we define
\begin{align}
g_{nk}^{(\ell)}(a,A)
&=\frac{1}{C}  \int x_{nk}^{(\ell)} \mathcal{P}\left(x_{nk}^{(\ell)}|z_{nk}^{(\ell-1)}\right)\mathcal{N}_c(z_{nk}^{(\ell-1)}|Z_{nk}^{(\ell-1)}(t+1),V_{nk}^{(\ell-1)}(t+1))\mathcal{N}(x_{nk}^{(\ell)}|a,A)\text{d}z_{nk}^{(\ell-1)}\text{d}x_{nk}^{(\ell)},
\end{align}
where $C$ is a normalization constant.  Accordingly, the mean and variance of $\mu_{n\rightarrow  mk}^{(\ell)}(x_{nk}^{(\ell)},t+1)$ are given
\begin{align}
\hat{x}_{n\rightarrow mk}^{(\ell)}(t+1)&=g_{nk}^{(\ell)}\left(R_{nk\backslash m}^{(x,\ell)}(t),\Sigma_{nk\backslash m}^{(x,\ell)}(t)\right),\\
v^{(x,\ell)}_{n\rightarrow mk}(t+1)&=\Sigma_{nk\backslash m}^{(x,\ell)}(t)g_{nk}'^{(\ell)}\left(R_{nk\backslash m}^{(x,\ell)}(t),\Sigma_{nk\backslash m}^{(x,\ell)}(t)\right),
\end{align}
where the last equation holds by the property of exponential family
 \footnote{
Given a distribution $\mathcal{P}(x|a,A)=\frac{f(x)\mathcal{N}(x|a,A)}{\int f(x)\mathcal{N}(x|a,A)\text{d}x}$, we have $\frac{\partial }{\partial a}\int x\mathcal{P}(x|a,A)\text{d}x=\frac{1}{A}\int (x-\hat{x})^2\mathcal{P}(x|a,A)\text{d}x$ with $\hat{x}=\int x\mathcal{P}(x|a,A)\text{d}x$.
} and $g_{nk}'^{(\ell)}(R_{nk\backslash m}^{(x,\ell)}(t),\Sigma_{nk\backslash m}^{(x,\ell)}(t))$ is the partial derivation w.r.t. the first argument.

One could see that there is only slight difference between $\mu_{n\rightarrow  mk}^{(\ell)}(x_{nk}^{(\ell)},t+1)$ and belief distribution $\mu_{nk}^{(\ell)}(x_{nk}^{(\ell)},t+1)$. To fix this gap, we define
\begin{align}
\Sigma_{nk}^{(x,\ell)}(t)&=\left(\sum_{r=1}^{N_{\ell+1}}|\hat{h}_{rn}^{(\ell)}(t)|^2v_{rk}^{(s,\ell)}(t)-v_{rn}^{(h,\ell)}(t)(|\hat{s}_{rk}^{(\ell)}(t)|^2-v^{(s,\ell)}_{rk}(t))\right)^{-1},\\
R_{nk}^{(x,\ell)}(t)&=\hat{x}_{nk}^{(\ell)}(t)\left[1+\Sigma_{nk}^{(x,\ell)}(t)\sum_{r=1}^{N_{\ell+1}}v_{rn}^{(h,\ell)}(t)(|\hat{s}_{rk}^{(\ell)}(t)|^2-v^{(s,\ell)}_{rk}(t))\right]+\Sigma_{nk}^{(x,\ell)}(t)\sum_{r=1}^{N_{\ell+1}}\hat{h}_{k\leftarrow rn}^{(\ell)}(t)\hat{s}_{rk}(t),
\end{align}
Accordingly,  we define RV $\xi_{nk}^{(x,\ell)}(t+1)$ following  $\mu_{nk}^{(\ell)}(x_{nk}^{(\ell)},t+1)$ i.e.,
\begin{align}
\nonumber
&\xi_{nk}^{(x,\ell)}(t+1)\sim\\
 &\frac{\int \mathcal{P}\left(x_{nk}^{(\ell)}|z_{nk}^{(\ell-1)}\right)\mathcal{N}(z_{nk}^{(\ell-1)}|Z_{nk}^{(\ell-1)}(t+1),V_{nk}^{(\ell)}(t+1))\mathcal{N}(x_{nk}^{(\ell)}|R_{nk}^{(x,\ell)}(t),\Sigma_{nk}^{(x,\ell)}(t))\text{d}z_{nk}^{(\ell-1)}}
{ \int \mathcal{P}\left(x_{mk}^{(\ell)}|z_{nk}^{(\ell-1)}\right)\mathcal{N}(z_{nk}^{(\ell-1)}|Z_{nk}^{(\ell-1)}(t+1),V_{nk}^{(\ell)}(t+1))\mathcal{N}(x_{nk}^{(\ell)}|R_{nk}^{(x,\ell)}(t),\Sigma_{nk}^{(x,\ell)}(t))\text{d}x_{nk}^{(\ell)}\text{d}z_{nk}^{(\ell-1)}}.
\end{align}
Specially, for $\ell=1$,  it becomes
\begin{align}
\xi_{nk}^{(x,1)}(t+1)\sim \frac{\mathcal{P}(x_{nk})\mathcal{N}(x_{nk}|R_{nk}^{(x,1)}(t),\Sigma_{nk}^{(x,1)}(t))}
{\int \mathcal{P}(x_{nk})\mathcal{N}(x_{nk}|R_{nk}^{(x,1)}(t),\Sigma_{nk}^{(x,1)}(t)) \text{d}x}.
\end{align}
The mean and variance of RV $\xi_{nk}^{(x,\ell)}(t+1)$ can be represented as
\begin{align}
\hat{x}_{nk}^{(\ell)}(t+1)&=g_{nk}^{(\ell)}(R_{nk}^{(x,\ell)}(t),\Sigma_{nk}^{(x,\ell)}(t)),\\
v_{nk}^{(x,\ell)}(t+1)&=\Sigma_{nk}^{(x,\ell)}(t)g_{nk}'^{(\ell)}(R_{nk}^{(x,\ell)}(t),\Sigma_{nk}^{(x,\ell)}(t)).
\end{align}

Using first-order Taylor series expansion we have
\begin{align}
\hat{x}_{n\rightarrow mk}^{(\ell)}(t+1)
&\approx g_{nk}^{(\ell)}(R_{nk}^{(x,\ell)}(t),\Sigma_{nk}^{(x,\ell)}(t))
-\Sigma_{nk}^{(x,\ell)}(t)\hat{h}_{mn}^{(\ell)}(t)\hat{s}_{mk}^{(\ell)}(t)g_{nk}'^{(\ell)}(R_{nk}^{(x,\ell)}(t),\Sigma_{nk}^{(x,\ell)}(t))\\
&=\hat{x}_{mk}^{(\ell)}(t+1)
-\hat{h}_{mn}^{(\ell)}(t)\hat{s}_{mk}^{(\ell)}(t)v_{nk}^{(x,\ell)}(t+1),
\label{Equ:X}
\end{align}
where the item $v_{mn}^{(h,\ell)}(t)(|\hat{s}_{mk}^{(\ell)}(t)|^2-v^{(s,\ell)}_{mk}(t))$ is ignored since $v_{mn}^{(h,\ell)}(t)$ is $\mathcal{O}(1/N_{\ell})$ and the item $\hat{h}_{mn}^{(\ell)}(t)$ is replaced by $\hat{h}_{k\leftarrow mn}^{(\ell)}(t)$ since $\hat{h}_{mn}^{(\ell)}(t)$ has the same order as $\hat{h}_{k\leftarrow mn}^{(\ell)}(t)$.

Likewise, applying first-order Taylor series expansion to  $v^{(x,\ell)}_{n\rightarrow mk}(x_{nk}^{(\ell)},t+1)$ and ignoring the high order items, we have
\begin{align}
v^{(x,\ell)}_{n\rightarrow mk}(t+1)
&\approx v^{(x,\ell)}_{nk}(t+1).
\label{Equ:XV}
\end{align}

Similarly, the message $\mu_{k\rightarrow mn}(h_{mn}^{(\ell)},t)$ is approximated with the mean and variance as
\begin{align}
\hat{h}^{(\ell)}_{k\leftarrow  mn}(t+1)&\approx  \hat{h}_{mn}(t+1)-\hat{x}_{nk}^{(\ell)}(t)\hat{s}_{mk}^{(\ell)}(t)v_{mn}^{(h,\ell)}(t+1),
\label{Equ:H}\\
v^{(h,\ell)}_{k\leftarrow mn}(t+1)&\approx v^{(h,\ell)}_{mn}(t+1),
\label{Equ:HV}
\end{align}
where $\hat{h}_{mn}(t+1)$ and $v^{(h,\ell)}_{mn}(t+1)$ are the mean and variance of RV $\xi_{mn}^{(h,\ell)}(t+1)$ following $\mu_{mn}^{(\ell)}(h_{mn}^{(\ell)},t+1)$
\begin{align}
\xi^{(h,\ell)}_{nk}(t+1)\sim \frac{ \mathcal{P}(h_{mn}^{(\ell)})\mathcal{N}\left(h_{mn}^{(\ell)}|R_{mn}^{(h,\ell)}(t),\Sigma_{mn}^{(h,\ell)}(t)\right)}
{\int \mathcal{P}(h_{mn}^{(\ell)})\mathcal{N}\left(h_{mn}^{(\ell)}|R_{mn}^{(h,\ell)}(t),\Sigma_{mn}^{(h,\ell)}(t)\right)\text{d}h_{mn}^{(\ell)}},
\end{align}
where the following definitions are applied
\begin{align}
\Sigma_{mn}^{(h,\ell)}(t)&=\left(\sum_{k=1}^{K}|\hat{x}_{nk}^{(\ell)}(t)|^2v_{mk}^{(s,\ell)}(t)-v_{nk}^{(x,\ell)}(t)(|\hat{s}_{mk}^{(\ell)}(t)|^2-v^{(s,\ell)}_{mk}(t))\right)^{-1},\\
\nonumber
R_{mn}^{(h,\ell)}(t)&=\hat{h}_{mn}^{(\ell)}(t)\left[1+\Sigma_{mn}^{(h,\ell)}(t)\sum_{k=1}^Kv_{nk}^{(x,\ell)}(t)(|\hat{s}_{mk}^{(\ell)}(t)|^2-v^{(s,\ell)}_{mk}(t))\right]\\
&\qquad +\Sigma_{mn}^{(h,\ell)}(t)\sum_{k=1}^{K}\hat{x}_{n\rightarrow mk}^{(\ell)}(t)\hat{s}_{mk}(t).
\end{align}

Summarizing those approximated messages constructs the relaxed belief propagation.  However, there still exist $\mathcal{O}(N_{\ell+1}N_{\ell})$ parameters in each iterations. One way to reduce the number of those parameters is to update the previous steps by the approximated results of $(\hat{h}^{(\ell)}_{k\leftarrow  mn}(t+1),v^{(h,\ell)}_{k\leftarrow  mn}(t+1))$ and $(\hat{x}_{n\rightarrow mk}^{(\ell)}(t+1), v_{n\rightarrow mk}^{(x,\ell)}(t+1))$.

\subsection{Close to loop}
Substituting (\ref{Equ:X}) and (\ref{Equ:H}) into (\ref{Equ:Z}) yields
\begin{align}
\nonumber
Z_{mk}^{(\ell)}(t)&=\sum_{n=1}^{N_{\ell}}\left(\hat{h}_{mn}(t)+\hat{x}_{nk}^{(\ell)}(t-1)\hat{s}_{mk}^{(\ell)}(t-1)v_{mn}^{(h,\ell)}(t)\right)\\
&\qquad \times \left(\hat{x}_{mk}^{(\ell)}(t)
+\hat{h}_{mn}^{(\ell)}(t-1)\hat{s}_{mk}^{(\ell)}(t-1)v_{nk}^{(x,\ell)}(t)\right)\\
\nonumber
&=\underbrace{\sum_{n=1}^{N_{\ell}}\hat{h}_{mn}(t)\hat{x}_{mk}^{(\ell)}(t)}_{\overset{\triangle}{=}\overline{Z}_{mk}^{(\ell)}(t)}-\hat{s}_{mk}^{(\ell)}(t-1)\sum_{n=1}^{N_{\ell}}\left[
\hat{x}_{mk}^{(\ell)}(t)\hat{x}_{nk}^{(\ell)}(t-1)v_{mn}^{(h,\ell)}(t)+
\hat{h}_{mn}^{(\ell)}(t)
\hat{h}_{mn}^{(\ell)}(t-1)v_{nk}^{(x,\ell)}(t)
\right]\\
&\qquad +|\hat{s}_{mk}^{(\ell)}(t-1)|^2\sum_{n=1}^{N_{\ell}}\hat{h}_{mn}^{(\ell)}(t-1)v_{nk}^{(x,\ell)}(t)\hat{x}_{nk}^{(\ell)}(t-1)v_{mn}^{(h,\ell)}(t)\\
&\approx \overline{Z}_{mk}^{(\ell)}(t)+\hat{s}_{mk}^{(\ell)}(t-1)\underbrace{\sum_{n=1}^{N_{\ell}}\left[
|\hat{x}_{nk}^{(\ell)}(t)|^2v_{mn}^{(h,\ell)}(t)+|\hat{h}_{mn}^{(\ell)}(t)|^2v_{nk}^{(x,\ell)}(t)
\right]}_{\overset{\triangle}{=}\overline{V}_{mk}^{(\ell)}(t)},
\end{align}
where we use $|\hat{x}_{nk}^{(\ell)}(t)|^2$ to replace $\hat{x}_{mk}^{(\ell)}(t)\hat{x}_{nk}^{(\ell)}(t-1)$, apply $|\hat{h}_{mn}^{(\ell)}(t)|^2$ to replace $\hat{h}_{mn}^{(\ell)}(t)
\hat{h}_{mn}^{(\ell)}(t-1)$, and neglect the infinitesimal terms relative to the remaining terms.

Next we plug (\ref{Equ:XV}) and (\ref{Equ:HV}) into (\ref{Equ:V}) and get
\begin{align}
\nonumber
V_{mk}^{(\ell)}(t)
&=\sum_{n=1}^{N_{\ell}}v_{mn}^{(h,\ell)}(t)v_{nk}^{(x,\ell)}(t)+v_{nk}^{(x,\ell)}(t)\sum_{n=1}^{N_{\ell}}[\hat{h}_{mn}(t)+\hat{x}_{nk}^{(\ell)}(t-1)\hat{s}_{mk}^{(\ell)}(t-1)v_{mn}^{(h,\ell)}(t)]^2\\
&\qquad +v_{mn}^{(h,\ell)}(t)\sum_{n=1}^{N_{\ell}}[\hat{x}_{mk}^{(\ell)}(t)
+\hat{h}_{mn}^{(\ell)}(t-1)\hat{s}_{mk}^{(\ell)}(t-1)v_{nk}^{(x,\ell)}(t)]^2\\
\nonumber
&=\overline{V}_{mk}^{(\ell)}(t)+\sum_{n=1}^{N_{\ell}}v_{mn}^{(h,\ell)}(t)v_{nk}^{(x,\ell)}(t)\\
\nonumber
&
\qquad -2\hat{s}^{(\ell)}_{mk}(t-1)\sum_{n=1}^{N_{\ell}}\left[v_{nk}^{(x,\ell)}(t)\hat{h}_{mn}(t)\hat{x}_{nk}^{(\ell)}(t-1)v_{mn}^{(h,\ell)}(t)+v_{mn}^{(h,\ell)}(t)\hat{x}_{mk}^{(\ell)}(t)
\hat{h}_{mn}^{(\ell)}(t-1)v_{nk}^{(x,\ell)}(t)\right]\\
\nonumber
&
\qquad +|\hat{s}^{(\ell)}_{mk}(t-1)|^2\sum_{n=1}^{N_{\ell}}\left[v_{nk}^{(x,\ell)}(t)|\hat{x}_{nk}^{(\ell)}(t-1)|^2(v_{mn}^{(h,\ell)}(t))^2\right.\\
&\qquad \qquad \qquad \qquad \qquad +\left.v_{mn}^{(h,\ell)}(t)|\hat{h}_{mn}^{(\ell)}(t-1)|^2(v_{nk}^{(x,\ell)}(t))^2\right]\\
&\approx \overline{V}_{mk}^{(\ell)}(t)+\sum_{n=1}^{N_{\ell}}v_{mn}^{(h,\ell)}(t)v_{nk}^{(x,\ell)}(t),
\end{align}
where only $\mathcal{O}(1)$ items are remained.

We then simplify  $\Sigma_{mn}^{(h,\ell)}(t)$ and $\Sigma_{nk}^{(x,\ell)}(t)$ as
\begin{align}
\Sigma_{nk}^{(x,\ell)}(t)&\approx \left(\sum_{m=1}^{N_{\ell+1}}|\hat{h}_{mn}^{(\ell)}(t)|^2v_{mk}^{(s,\ell)}(t)\right)^{-1},\\
\Sigma_{mn}^{(h,\ell)}(t)&\approx\left(\sum_{k=1}^{K}|\hat{x}_{nk}^{(\ell)}(t)|^2v_{mk}^{(s,\ell)}(t)\right)^{-1},
\end{align}
where the items $\sum_{m=1}^{N_{\ell+1}}v_{mn}^{(h,\ell)}(t)(|\hat{s}_{mk}^{(\ell)}(t)|^2-v^{(s,\ell)}_{mk}(t))$ and  $\sum_{k=1}^{K}v_{nk}^{(x,\ell)}(t)(|\hat{s}_{mk}^{(\ell)}(t)|^2-v^{(s,\ell)}_{mk}(t))$ are neglected (detail please finds in Appendix \ref{Appendix:D}). When keeping those items yields message passing related \cite{Kabashima2016Phase}.

With approximations above, we simplify $R_{mn}^{(h,\ell)}(t)$ and $R_{nk}^{(x,\ell)}(t)$ as below
\begin{align}
R_{mn}^{(h,\ell)}(t)&=\hat{h}_{mn}^{(\ell)}(t)\left[1-\Sigma_{mn}^{(h,\ell)}(t)\sum_{k=1}^{K}v^{(x,\ell)}_{nk}(t)v^{(s,\ell)}_{mk}(t)\right]+\Sigma_{mn}^{(h,\ell)}(t)\sum_{k=1}^{K}\hat{x}_{nk}^{(\ell)}(t)\hat{s}_{mk}(t),\\
R_{nk}^{(x,\ell)}(t)&=\hat{x}_{nk}^{(\ell)}(t)\left[1-\Sigma_{nk}^{(x,\ell)}(t)\sum_{m=1}^{N_{\ell+1}}v^{(h,\ell)}_{mn}(t)v_{mk}^{(s,\ell)}(t)\right]+\Sigma_{nk}^{(x,\ell)}(t)\sum_{m=1}^{N_{\ell+1}}\hat{h}_{mn}^{(\ell)}(t)\hat{s}_{mk}(t).
\end{align}

\section{Proof for Proposition 1}
\label{Appendix:B}

\subsection{Simplification to ML-BiGAMP}
The scalar-variance ML-BiGAMP is the pre-condition to derive ML-BiGAMP'SE, where element-wise variances are replaced by scalar variances to reduce the memory and complexity of ML-BiGAMP. To obtain this algorithm, we assume
\begin{align}
v^{(x,\ell)}_{nk}(t)&\approx \frac{1}{N_{\ell}K}\sum_{n=1}^{N_{\ell}}\sum_{k=1}^Kv_{nk}^{(x,\ell)}(t)=\overline{v^{(x,\ell)}(t)},
\label{Equ:vx}\\
v^{(h,\ell)}_{mn}(t)&\approx \frac{1}{N_{\ell+1}N_{\ell}}\sum_{m=1}^{N_{\ell+1}}\sum_{n=1}^{N_{\ell}}v_{mn}^{(h,\ell)}(t)=\overline{v^{(h,\ell)}(t)},
\label{Equ:vrh}\\
\tilde{v}_{mk}^{(\ell)}(t)&\approx \frac{1}{N_{\ell+1}K}\sum_{m=1}^{N_{\ell+1}}\sum_{k=1}^K \tilde{v}_{mk}^{(\ell)}(t)=\overline{\tilde{v}^{(\ell)}(t)}.
\label{Equ:vz}
\end{align}
Based on the approximations above, we simplify the variance parameters in Algorithm 1 as below

\begin{align}
\overline{V}_{mk}^{(\ell)}(t)
&\approx \frac{\overline{v^{(h,\ell)}(t)}}{K}\sum_{n=1}^{N_{\ell}}\sum_{k=1}^{K}|\hat{x}_{nk}^{(\ell)}(t)|^2 +
\frac{\overline{v^{(x,\ell)}(t)}}{N_{\ell+1}}\sum_{m=1}^{N_{\ell+1}}\sum_{n=1}^{N_{\ell}}|\hat{h}_{mn}^{(\ell)}(t)|^2=\overline{V^{(\ell)}(t)},
\\
 V_{mk}^{(\ell)}(t)&\approx \overline{V^{(\ell)}(t)}+N_{\ell}\overline{v^{(x,\ell)}(t)}\cdot \overline{v^{(h,\ell)}(t)}=V^{(\ell)}(t),
 \label{Equ:V}\\
 v_{mk}^{(s,\ell)}(t)&\approx \frac{V^{(\ell)}(t)-\overline{\tilde{v}^{(\ell)}(t)}}{(V^{(\ell)}(t))^2}=v^{(s,\ell)}(t),
 \label{Equ:vs}\\
\Sigma_{nk}^{(x,\ell)}(t)&\approx \left(\frac{v^{(s,\ell)}(t)}{N_{\ell}}\sum_{m=1}^{N_{\ell+1}}\sum_{n=1}^{N_{\ell}}|\hat{h}_{mn}^{(\ell)}(t)|^2\right)^{-1}=\Sigma^{(x,\ell)}(t),
\label{Equ:Sigmax}\\
\Sigma_{mn}^{(h,\ell)}(t)&\approx
\left(\frac{v^{(s,\ell)}(t)}{N_{\ell}}\sum_{n=1}^{N_{\ell}}\sum_{k=1}^K|\hat{x}_{nk}^{(\ell)}(t)|^2\right)^{-1}=\Sigma^{(h,\ell)}(t).
\label{Equ:Sigmah}
\end{align}

To close the loop, we apply those variance parameters to rewrite the mean parameters in Algorithm~1
\begin{align}
R_{nk}^{(x,\ell)}(t)&=\hat{x}_{nk}^{(\ell)}(t)\left[1-N_{\ell+1}\Sigma^{(x,\ell)}(t)v^{(s,\ell)}(t)\overline{v^{(h,\ell)}(t)}\right]  +\Sigma^{(x,\ell)}(t)\sum\nolimits_{m=1}^{N_{\ell+1}}(\hat{h}_{mn}^{(\ell)}(t))^*\hat{s}^{(\ell)}_{mk}(t),\\
R_{mn}^{(h,\ell)}(t)&=\hat{h}_{mn}^{(\ell)}(t)\left[1-K\Sigma^{(h,\ell)}(t)v^{(s,\ell)}(t)\overline{v^{(x,\ell)}(t)}\right] +\Sigma^{(h,\ell)}(t)\sum\nolimits_{k=1}^{K}(\hat{x}_{nk}^{(\ell)}(t))^{*}\hat{s}_{mk}^{(\ell)}(t).
\end{align}
Those simplification results together with the remaining parameters in Algorithm 1 construct the scalar-variance ML-BiGAMP algorithm.

\subsection{Derivation of SE}
Before giving derivation, we introduce the following concepts.
\begin{definition}[Pseudo-Lipschitz function]
For any $k\geq 1$, a function $\varphi(\cdot):\mathbb{R}^p\mapsto \mathbb{R}$ ($p\geq 1$) is pseudo-Lipschitz of order k, if there exists a constant $C>0$ such that for any $\bs{x},\bs{y}\in \mathbb{R}^p$,
\begin{align}
\left|\varphi(\bs{x})-\varphi(\bs{y})\right|\leq C\left(1+\left|\bs{x}\right|^{k-1}+\left|\bs{y}\right|^{k-1}\right)\|\bs{x}-\bs{y}\|.
\end{align}
\end{definition}

\begin{definition}
\label{Asp_1}
Let $\bs{x}=\{\bs{x}_n(N)\}_{n=1}^N$ be a block vector sequence set with $\bs{x}_n(N)\in \mathbb{R}^p$ $(p\geq 1)$. Given $k\geq 1$, $\bs{x}$ converges empirically a random variable $\boldsymbol{\textsf{X}}$ on $\mathbb{R}^p$  with $k$-th order moments if \\
\noindent (i) $\mathbb{E}|\textsf{X}|^k<\infty$; and \\
\noindent (ii) for any pseudo-Lipschitz continuous function $\varphi(\cdot)$ of order $k$,
\begin{align}
\lim_{N\rightarrow \infty}\frac{1}{N}\sum_{n=1}^N\varphi(\bs{x}_n(N))-\mathbb{E}\left\{\varphi(\boldsymbol{\textsf{X}})\right\}\overset{\rm{a.s.}}{\longrightarrow} 0.
\end{align}
Thus, the empirical mean of the components $\varphi(\bs{x}_n(N))$ converges to the expectation $\mathbb{E}\{\varphi(\textsf{X})\}$.  For ease of notation, we write it as $\lim_{N\rightarrow \infty} \{\bs{x}_n(N)\}_{n=1}^N\overset{PL(k)}{=}\textsf{X}$.
\end{definition}

\begin{assumption}
We assume that the mean related parameters $\{y_{mk}, Z_{mk}^{(\ell)}(t),z_{mk}^{(\ell)}$, $R^{(x,\ell)}_{nk}(t),x_{nk}^{(\ell)}, R^{(h,\ell)}_{mn}(t), h_{mn}^{(\ell)}\}$ converge empirically to the following RVs with $2$nd order moments
\begin{align}
\lim_{K,N_{\ell}\rightarrow \infty}\{y_{mk}, Z_{mk}^{(\ell)}(t),z_{mk}^{(\ell)},R^{(x,\ell)}_{nk}(t),x_{nk}^{(\ell)}, R^{(h,\ell)}_{mn}(t), h_{mn}^{(\ell)}\}
\overset{PL(2)}{=} \{\textsf{Y}, \textsf{Z}^{(\ell)}(t),\textsf{z}^{(\ell)}, \textsf{R}^{(x,\ell)}(t),\textsf{X}^{(\ell)}, \textsf{R}^{(h,\ell)}(t),\textsf{H}^{(\ell)}\}.
\label{Equ:em_con}
\end{align}
\end{assumption}

Based on this assumption, we first calculate the asymptotic MSE of iteration-$t$ $\hat{\bs{X}}^{(\ell)}(t)$, for $1<\ell\leq L$, defined as
\begin{align}
\textsf{mse}(\bs{X}^{(\ell)},t)=\lim_{N_{\ell},K\rightarrow \infty}\frac{1}{N_{\ell}K}\|\hat{\bs{X}}^{(\ell)}(t)-\bs{X}^{(\ell)}\|_{\text{F}}^2.
\end{align}
We particularize the \textit{pseudo-Lipchitz} continuous function $g^{(x)}(\cdot)$ and $\varphi^{(x)}(\cdot)$ as
\begin{align}
g^{(x)}(Z_{nk}^{(\ell-1)}(t), R^{(x,\ell)}_{nk}(t-1))&=\hat{x}_{nk}^{(\ell)}(t),\\
\varphi^{(x)}(Z_{nk}^{(\ell-1)}(t), R^{(x,\ell)}_{nk}(t-1))&=v^{(x,\ell)}_{nk}(t),
\end{align}
where $\hat{x}_{nk}^{(\ell)}(t)$ and $v^{(x,\ell)}_{nk}(t)$ are the mean and variance of the approximate posterior distribution $\hat{\mathcal{P}}^t(x_{nk}^{(\ell)}|y)$ found in (\ref{Equ:xi_x}).

Pertaining to the asymptotic MSE of $\hat{\bs{X}}(t)$, we describe the following proposition.

\setcounter{proposition}{+3}
\begin{proposition}
In large system limit, the asymptotic MSE of iteration-$t$ estimator $\hat{\bs{X}}^{(\ell)}(t)$ is identical to $\overline{v^{(x,\ell)}(t)}$  and $\mathbb{E}_{\textsf{Z}^{(\ell-1)}(t), \textsf{R}^{(x,\ell)}(t-1)}\left\{\varphi^{(x)}(\textsf{Z}^{(\ell-1)}(t), \textsf{R}^{(x,\ell)}(t-1))\right\}$ almost sure.
\end{proposition}
\begin{proof}
To prove this proposition, we write
\begin{align}
\textsf{mse}(\bs{X}^{(\ell)}, t)
&=\lim_{N_{\ell},K\rightarrow \infty}\frac{1}{N_{\ell}K}\sum_{n=1}^{N_{\ell}}\sum_{k=1}^K(\hat{x}_{nk}^{(\ell)}(t)-x_{nk}^{(\ell)})^2\\
&\overset{(a)}{=}\mathbb{E}_{\textsf{Z}^{(\ell-1)}(t), \textsf{R}^{(x,\ell)}(t-1)} \left\{\left(g^{(x)}(\textsf{Z}^{(\ell-1)}(t), \textsf{R}^{(x,\ell)}(t-1))-\textsf{X}^{(\ell)}\right)^2\right\}\\
&\overset{(b)}{=}\mathbb{E}_{\textsf{Z}^{(\ell-1)}(t), \textsf{R}^{(x,\ell)}(t-1)} \left\{\varphi^{(x)}(\textsf{Z}^{(\ell-1)}(t), \textsf{R}^{(x,\ell)}(t-1))\right\}\\
&\overset{(c)}{=}\overline{v^{(x,\ell)}(t)},
\end{align}
where $(a)$ and $(b)$ holds by the empirical convergence to random variables in (\ref{Equ:em_con}), and $(c)$  holds by the following steps
\begin{align}
\overline{v^{(x,\ell)}(t)}
&=\frac{1}{N_{\ell}K}\sum_{n=1}^{N_{\ell}}\sum_{k=1}^K\varphi^{(x)}(Z_{nk}^{(\ell-1)}(t), R_{nk}^{(x,\ell)}(t-1)))\\
&=\mathbb{E}_{\textsf{Z}^{(\ell-1)}(t), \textsf{R}^{(x,\ell)}(t-1)}\left\{\varphi^{(x)}(\textsf{Z}^{(\ell-1)}(t), \textsf{R}^{(x,\ell)}(t-1))\right\}.
\end{align}
\end{proof}

For the case $\ell=1$, the similar result can be obtained. As a result, we have
\begin{align}
\textsf{mse}(\bs{X}^{(\ell)}, t)=\overline{\hat{v}^{(x,\ell)}(t)}=
\begin{cases}
\mathbb{E}_{\textsf{Z}^{(\ell-1)}(t),\textsf{R}^{(x,\ell)}(t-1)}\{\varphi^{(x)}(\textsf{Z}^{(\ell-1)}(t), \textsf{R}^{(x,\ell)}(t-1))\}  &\ell>1\\
\mathbb{E}_{\textsf{R}^{(x,\ell)}(t-1)}\{\varphi^{(x)}( \textsf{R}^{(x,\ell)}(t-1))\}   &\ell=1
\end{cases}.
\end{align}

Pertaining to the asymptotic MSE of $\hat{\bs{H}}^{(\ell)}(t)$ and $\hat{\bs{Z}}^{(\ell)}(t)$, we define
\begin{align}
\textsf{mse}(\bs{Z}^{(\ell)},t)&=\lim_{N_{\ell+1},K\rightarrow \infty}\frac{1}{N_{\ell+1}K}\|\tilde{\bs{Z}}^{(\ell)}(t)-\bs{Z}^{(\ell)}\|_{\text{F}}^2,\\
\textsf{mse}(\bs{H}^{(\ell)},t)&=\lim_{N_{\ell+1},N_{\ell}\rightarrow \infty}\frac{1}{N_{\ell+1}N_{\ell}}\|\hat{\bs{H}}^{(\ell)}(t)-\bs{H}^{(\ell)}\|_{\text{F}}^2.
\end{align}
Similar to $\bs{X}^{(\ell)}(t)$, the follows can be obtained
\begin{align}
\textsf{mse}(\bs{H}^{(\ell)}, t)&=\overline{\hat{v}^{(h,\ell)}(t)}=\mathbb{E}_{\textsf{H}^{(\ell)},\textsf{R}^{(h,\ell)}(t)}\{\varphi^{(h)}(\textsf{H}^{(\ell)},\textsf{R}^{(h,\ell)}(t))\},\\
\textsf{mse}(\bs{Z}^{(\ell)}, t)&=\overline{\tilde{v}^{(\ell)}(t)}=
\begin{cases}
\mathbb{E}_{\textsf{Z}^{(\ell)}(t),\textsf{R}^{(x,\ell+1)}(t)}\{\varphi^{(z)}(\textsf{Z}^{(\ell)}(t), \textsf{R}^{(x,\ell+1)}(t))\}  &\ell<L\\
\mathbb{E}_{\textsf{Z}^{(\ell)}(t),\textsf{Y}}\{\varphi^{(z)}(\textsf{Z}^{(\ell)}(t),\textsf{Y})\}   &\ell=L
\end{cases}.
\end{align}
where
\begin{align}
&\varphi^{(h)}(h_{mn}^{(\ell)},R_{mn}^{(h,\ell)}(t))=v_{mn}^{(h,\ell)}(t),\\
&\varphi^{(z)}(Z_{mk}^{(\ell)}(t), R_{mk}^{(x,\ell+1)}(t))=\tilde{v}_{mk}^{(\ell)}(t).
\end{align}

We move to giving the step-by-step derivation of the asymptotic MSEs of those MMSE estimators. For simplification, we omit iteration $t$ in the following derivation.

\textbf{Step 1}: We first compute $\overline{\tilde{v}^{(\ell)}}$, for $1\leq \ell<L$,
\begin{align}
\overline{\tilde{v}^{(\ell)}}&=\mathbb{E}_{\textsf{Z}^{(\ell)},\textsf{R}^{(x,\ell+1)}}\{\varphi^{(z)}(\textsf{Z}^{(\ell)},\textsf{R}^{(x,\ell+1)})\}\\
&=\mathbb{E}_{\textsf{Z}^{(\ell)},\textsf{R}^{(x,\ell+1)}}\{\mathbb{E}\{|z^{(\ell)}|^2\}-|\mathbb{E}\{z^{(\ell)}\}|^2\},
\label{Equ:tildev}
\end{align}
where the inner expectation is taken over the approximate posterior distribution $\hat{\mathcal{P}}(z^{(\ell)}|y)$ in (\ref{Equ:zeta})
\begin{align}
\hat{\mathcal{P}}(z^{(\ell)}|y)=\frac{
\int \mathcal{N}_{x^{(\ell+1)}|z^{(\ell)}}(Z^{(\ell)}, V^{(\ell)},R^{(x,\ell+1)}, \Sigma^{(x,\ell+1)})\text{d}x^{(\ell+1)}
}{
\int \mathcal{N}_{x^{(\ell+1)}|z^{(\ell)}}(Z^{(\ell)}, V^{(\ell)},R^{(x,\ell+1)}, \Sigma^{(x,\ell+1)})\text{d}z^{(\ell)}\text{d}x^{(\ell+1)}
},
\end{align}
By the \textit{Markov property}, the joint distribution the random variables (RVs) $(\textsf{Z}^{(\ell)},\textsf{z}^{(\ell)},\textsf{X}^{(\ell+1)},\textsf{R}^{(x,\ell+1)})$ can be represented as
\begin{align}
\mathcal{P}(Z^{(\ell)},&z^{(\ell)},x^{(\ell+1)},R^{(x,\ell+1)})
=\mathcal{P}(Z^{(\ell)})\mathcal{P}(z^{(\ell)}|Z^{(\ell)})\mathcal{P}(x^{(\ell+1)}|z^{(\ell)})\mathcal{P}(R^{(x,\ell+1)}|x^{(\ell+1)}),
\label{Equ:JointZzXR}
\end{align}
where $\mathcal{P}(z^{(\ell)}|Z^{(\ell)})=\mathcal{N}(z^{(\ell)}|Z^{(\ell)},V^{(\ell)})$ and  $\mathcal{P}(R^{(x,\ell+1)}|x^{(\ell+1)})=\mathcal{N}(x^{(\ell+1)}|R^{(x,\ell+1)},\Sigma^{(x,\ell+1)})$. Besides, the distribution $\mathcal{P}(Z^{(\ell)})$ can be obtained by solving the following equation
\begin{align}
\int \mathcal{P}(Z^{(\ell)})\mathcal{P}(z^{(\ell)}|Z^{(\ell)})\text{d}Z^{(\ell)}=\mathcal{P}(z^{(\ell)}).
\label{Equ:cov_pZ}
\end{align}
Note that $z^{(\ell)}$ is the sum of a large number of independent terms, i.e., $z_{mk}^{(\ell)}=\sum_{n}h_{mn}^{(\ell)}x_{nk}^{(\ell)}$. It allows us to treat $z^{(\ell)}$ as Gaussian random variable with zero mean and variance $\chi_z^{(\ell)}$
\begin{align}
\chi_z^{(\ell)}
&=\mathbb{E} \left\{\left(\sum_{n=1}^{N_{\ell}}h_{mn}^{(\ell)}x_{nk}^{(\ell)}\right)\left(\sum_{r=1}^{N_{\ell}}h_{mr}^{(\ell)}x_{rk}^{(\ell)}\right)\right\}\\
&=\sum_{n=1}^{N_{\ell}}\mathbb{E} \left\{\left(h_{mn}^{(\ell)}\right)^2(x_{nk}^{(\ell)})^2\right\}\\
&=N_{\ell}\chi_h^{(\ell)}\chi_x^{(\ell)},
\label{Equ:chi_z}
\end{align}
where $\chi_h^{(\ell)}=\int \left(h^{(\ell)}\right)^2 \mathcal{P}(h^{(\ell)})\text{d}h^{(\ell)}$ and
\begin{align}
&\ell=1: \ \chi_x^{(\ell)}=\int x^2\mathcal{P}(x)\text{d}x,\\
&\ell>1: \ \chi_x^{(\ell)}=\int (x^{(\ell)})^2\mathcal{P}(x^{(\ell)}|z^{(\ell-1)})\mathcal{N}(z^{(\ell-1)}|0,\chi_z^{(\ell-1)})\text{d}z^{(\ell-1)}\text{d}x^{(\ell)}.
\end{align}

As a result, solving (\ref{Equ:cov_pZ}) yields
\begin{align}
\mathcal{P}(Z^{(\ell)})=\mathcal{N}(Z^{(\ell)}|0,\chi_z^{(\ell)}-V^{(\ell)}).
\end{align}
Further, the distribution of a pair random variables $(\textsf{Z}^{(\ell)},\textsf{R}^{(x,\ell+1)})$ is evaluated as
\begin{align}
\mathcal{P}(Z^{(\ell)},R^{(x,\ell+1)})
&=\mathcal{P}(Z^{(\ell)})\int \mathcal{P}(z^{(\ell)}|Z^{(\ell)})\mathcal{P}(x^{(\ell+1)}|z^{(\ell)}) \mathcal{P}(R^{(x,\ell+1)}|x^{(\ell+1)})\text{d}x^{(\ell+1)}\text{d}z^{(\ell)}.
\label{Equ:JointZRX}
\end{align}

From (\ref{Equ:tildev}), we have
\begin{align}
\overline{\tilde{v}^{(\ell)}}=\chi_z^{(\ell)}-q_z^{(\ell)},
\end{align}
where $\mathbb{E}_{\textsf{Z}^{(\ell)},\textsf{R}^{(x,\ell+1)}}\{\mathbb{E}\{|z^{(\ell)}|^2\}\}=\chi_z^{(\ell)}$, and $\mathbb{E}_{\textsf{Z}^{(\ell)},\textsf{R}^{(x,\ell+1)}}\{|\mathbb{E}\{z^{(\ell)}\}|^2\}=q_z^{(\ell)}$ with $q_z^{(\ell)}$ being
\begin{align}
q_z^{(\ell)}&=\int \frac{\left[\int z^{(\ell)}\mathcal{N}_{x|z}^{(\ell)}\left(\sqrt{\chi_z^{(\ell)}-V^{(\ell)}}\xi,V^{(\ell)},\zeta,\Sigma^{(x,\ell+1)}\right)\text{d}x^{(\ell+1)}\text{d}z^{(\ell)}\right]^2}
{\int \mathcal{N}_{x|z}^{(\ell)}\left(\sqrt{\chi_z^{(\ell)}-V^{(\ell)}}\xi,V^{(\ell)},\zeta,\Sigma^{(x,\ell+1)}\right)\text{d}x^{(\ell+1)}\text{d}z^{(\ell)}}\text{D}\xi\text{d}\zeta
,
\end{align}
where $\text{D}\xi$ is Gaussian measure expressed as $\text{D}\xi=\mathcal{N}(\xi|0,1)\text{d}\xi$.

For the case of $\ell=L$, we also have $\overline{\tilde{v}^{(\ell)}}=\chi_z^{(\ell)}-q_z^{(\ell)}$, where $q_z^{(\ell)}$ is the form of
\begin{align}
q_z^{(\ell)}=\int \frac{\left[\int z^{(\ell)} \mathcal{P}(y|z^{(\ell)})\mathcal{N}\left(z^{(\ell)}|\sqrt{\chi_z^{(\ell)}-V^{(\ell)}}\xi,V^{(\ell)}\right)\text{d}z^{(\ell)}\right]^2}
{\int  \mathcal{P}(y|z^{(\ell)})\mathcal{N}\left(z^{(\ell)}|\sqrt{\chi_z^{(\ell)}-V^{(\ell)}}\xi,V^{(\ell)}\right)\text{d}z^{(\ell)}}\text{D}\xi\text{d}y.
\end{align}

\textbf{Step 2}: The evaluation of $\overline{v^{(x,\ell)}}$ is similar to that of $\overline{\tilde{v}^{(\ell)}}$. For $1<\ell\leq L$,
\begin{align}
\overline{v^{(x,\ell)}}&=\mathbb{E}_{\textsf{Z}^{(\ell-1)},\textsf{R}^{(x,\ell)}}\{\varphi^{(x)}(\textsf{Z}^{(\ell-1)},\textsf{R}^{(x,\ell)})\}\\
&=\mathbb{E}_{\textsf{Z}^{(\ell-1)},\textsf{R}^{(x,\ell)}}\{\mathbb{E}\{|x^{(\ell)}|^2\}-|\mathbb{E}\{x^{(\ell)}\}|^2\},
\label{Equ:overline_vx}
\end{align}
where the inner expectation is taken over the approximate posterior distribution $\hat{\mathcal{P}}(x^{(\ell)}|y)$ in (\ref{Equ:xi_x}).
The distribution of random variables $(\textsf{Z}^{(\ell-1)},\textsf{z}^{(\ell-1)},\textsf{x}^{(\ell)},\textsf{R}^{(x,\ell)})$ is given in (\ref{Equ:JointZzXR}) and the distribution of random variables $(\textsf{Z}^{(\ell-1)},\textsf{R}^{(x,\ell)})$ is given in (\ref{Equ:JointZRX}). From (\ref{Equ:overline_vx}), the following can be obtained
\begin{align}
\overline{v^{(x,\ell)}}=\chi_x^{(\ell)}-q_x^{(\ell)},
\end{align}
where $\mathbb{E}_{\textsf{Z}^{(\ell-1)},\textsf{R}^{(x,\ell)}}\{\mathbb{E}\{|x^{(\ell)}|^2\}\}=\chi_x^{(\ell)}$, and $\mathbb{E}_{\textsf{Z}^{(\ell-1)},\textsf{R}^{(x,\ell)}}\{|\mathbb{E}\{x^{(\ell)}\}|^2\}=q_x^{(\ell)}$ with $q_x^{(\ell)}$ being
\begin{align}
\!\!\!\!\! \ell>1:\ q_x^{(\ell)}&=\int \frac{\left[\int x^{(\ell)} \mathcal{N}_{x|z}^{(\ell-1)}(\sqrt{\chi_z^{(\ell-1)}-V^{(\ell-1)}}\xi,V^{(\ell-1)},\zeta,\Sigma^{(x,\ell)})\text{d}x^{(\ell)}\text{d}z^{(\ell-1)}\right]^2}{\int \mathcal{N}_{x|z}^{(\ell-1)}(\sqrt{\chi_z^{(\ell-1)}-V^{(\ell-1)}}\xi,V^{(\ell-1)},\zeta,\Sigma^{(x,\ell)})\text{d}x^{(\ell)}\text{d}z^{(\ell-1)}}\text{D}\xi\text{d}\zeta
,\\
\!\!\!\!\! \ell=1: \ q_x^{(\ell)}&=\int \frac{\left[\int x\mathcal{P}(x)\mathcal{N}(x|\zeta,\Sigma^{(x,\ell)})\text{d}x\right]^2}
{\int \mathcal{P}(x)\mathcal{N}(x|\zeta,\Sigma^{(x,\ell)})\text{d}x}\text{d}\zeta.
\end{align}
Note that $\overline{v^{(x,\ell)}}$ refers to the MSE associated with approximate posterior  $\hat{\mathcal{P}}(x^{(\ell)}|y)$.

Additionally, the evaluation of $\overline{v^{(h,\ell)}}$ is easier relative to that of $\overline{\tilde{v}^{(\ell)}}$ and $\overline{v^{(x,\ell)}}$ due to the known prior $\mathcal{P}(h^{(\ell)})$. After some algebras, the following can be obtained
\begin{align}
\overline{v^{(h,\ell)}}&=\chi_h^{(\ell)}-q_h^{(\ell)},\\
q_h^{(\ell)}&=\int \frac{\left[\int h^{(\ell)}\mathcal{P}(h^{(\ell)})\mathcal{N}(h^{(\ell)}|\zeta,\Sigma^{(h,\ell)})\text{d}h^{(\ell)}\right]^2}
{\int \mathcal{P}(h^{(\ell)})\mathcal{N}(h^{(\ell)}|\zeta,\Sigma^{(h,\ell)})\text{d}h^{(\ell)}}\text{d}\zeta.
\end{align}
It is worthy of noting that  $\overline{v^{(h,\ell)}}$ represents the MSE associated with  $\hat{\mathcal{P}}(h^{(\ell)}|y)$.

\textbf{Step 3}: It is found that only the variance related parameters have impact on $\overline{\tilde{v}^{(\ell)}}$, $\overline{v^{(x,\ell)}}$, and $\overline{v^{(h,\ell)}}$. These parameters are $V^{(\ell)}$, $\Sigma^{(x,\ell)}$, and $\Sigma^{(h,\ell)}$. We thus apply the results above to represent those variance related parameters, which yields
\begin{align}
V^{(\ell)}
&=N_{\ell}(\chi_x^{(\ell)}\chi_h^{(\ell)}-q_x^{(\ell)}q_h^{(\ell)}),\\
v^{(s,\ell)}&=\frac{q_z^{(\ell)}-N_{\ell}q_x^{(\ell)}q_h^{(\ell)}}{N_{\ell}^2(\chi_x^{(\ell)}\chi_h^{(\ell)}-q_x^{(\ell)}q_h^{(\ell)})^2},\\
\Sigma^{(x,\ell)}
&=\frac{N_{\ell}(\chi_x^{(\ell)}\chi_h^{(\ell)}-q_x^{(\ell)}q_h^{(\ell)})^2}{\beta_{\ell}q_h^{(\ell)}(q_z^{(\ell)}-N_{\ell}q_x^{(\ell)}q_h^{(\ell)})},\\
\Sigma^{(h,\ell)}
&=\frac{\alpha \prod_{l=1}^{\ell-1}\beta_{l}N_{\ell}(\chi_x^{(\ell)}\chi_h^{(\ell)}-q_x^{(\ell)}q_h^{(\ell)})^2}{q_x^{(\ell)}(q_z^{(\ell)}-N_{\ell}q_x^{(\ell)}q_h^{(\ell)})}.
\end{align}

\section{Replica analysis}
\label{Appendix:C}

In this section, we firstly calculate the free energy of a representative two-layer model, and it leads to a set of saddle point equations after applying some techniques (e.g., central limit theorem); Secondly, based on replica symmetry assumption, the fixed point equations could be obtained by solving the saddle point equations. Finally, the results of the two-layer model can be extended to the multi-layer regime with similar procedures.

\subsection{Representative Two-Layer Model}
The representative two-layer model described as below is the multi-layer model (1) in $L=2$,
\begin{align}
\begin{cases}
\text{1-st layer:}\  \ \bs{S}=\boldsymbol{\phi}^{(1)}(\bs{HX},\bs{W}^{(1)})\\
\text{2-ed layer:}\  \bs{Y}=\boldsymbol{\phi}^{(2)}(\bs{CS},\bs{W}^{(2)})
\end{cases},
\end{align}
where we  use $(\bs{H},\bs{S},\bs{C})$ to represent $(\bs{H}^{(1)},\bs{X}^{(2)},\bs{H}^{(2)})$. In addition, we define $\bs{U}=\bs{HX}$ and $\bs{V}=\bs{CS}$, and apply the notations $(N_1, N_2, N_3)\leftarrow (N, M, P)$ and $(\beta_1,\beta_2)\leftarrow(\beta,\gamma)$.

The free energy \cite{Kabashima2016Phase} of this model is written as
\begin{align}
\mathcal{F}=\lim_{N\rightarrow \infty} \frac{1}{N^2}\lim_{\tau\rightarrow 0}\frac{\partial }{\partial \tau}\log \mathbb{E}_{\bs{Y}}\left\{\mathcal{P}^{\tau}(\bs{Y})\right\},
\label{Equ:FreeEnergy}
\end{align}
where $\mathcal{P}(\bs{Y})$ is the partition function given by
\begin{align}
\mathcal{P}(\bs{Y})&=\int \mathcal{P}(\bs{Y}|\bs{C},\bs{S})\mathcal{P}(\bs{C})\mathcal{P}(\bs{S})\text{d}\bs{C}\text{d}\bs{S},
\label{Equ:Py}\\
\mathcal{P}(\bs{S})&=\int \mathcal{P}(\bs{S}|\bs{H},\bs{X})\mathcal{P}(\bs{H})\mathcal{P}(\bs{X})\text{d}\bs{H}\text{d}\bs{X}.
\label{Equ:Ps}
\end{align}

\subsection{Begin at The Last Layer}
\label{Sec:Firstlayer}
From (\ref{Equ:FreeEnergy}), (\ref{Equ:Py}), and (\ref{Equ:Ps}), the term $\mathbb{E}\{\mathcal{P}^{\tau}(\bs{Y})\}$ in free energy can be rewritten as
\begin{align}
\mathbb{E}\{\mathcal{P}^{\tau}(\bs{Y})\}&=\int_{\bs{Y}}\prod_{a=0}^{\tau}\int_{\bs{C}^{(a)},\bs{S}^{(a)}} \mathcal{P}(\bs{Y}|\bs{C}^{(a)},\bs{S}^{(a)})\mathcal{P}(\bs{C}^{(a)})  \mathcal{P}(\bs{S}^{(a)}) \text{d}\bs{C}^{(a)}\text{d}\bs{S}^{(a)}\text{d}\bs{Y}\\
&=\int \mathcal{P}(\bs{Y}|\boldsymbol{\mathcal{V}})\mathbb{E}_{\boldsymbol{\mathcal{C}},\boldsymbol{\mathcal{S}}}\left\{\delta(\boldsymbol{\mathcal{V}}-\boldsymbol{\mathcal{CS}})\right\}
\text{d}\boldsymbol{\mathcal{V}}\text{d}\bs{Y},
\label{Equ:EY}
\end{align}
where the fact $\mathcal{P}(\bs{Y}|\bs{C},\bs{S})=\int \mathcal{P}(\bs{Y}|\bs{V})\delta(\bs{V}-\bs{CS})\text{d}\bs{V}$ and the definitions $\boldsymbol{\mathcal{V}}=\{\bs{V}^{(a)},\forall a\}$, $\boldsymbol{\mathcal{C}}=\{\bs{C}^{(a)},\forall a\}$, $\boldsymbol{\mathcal{S}}=\{\bs{S}^{(a)},\forall a\}$, and $\mathcal{P}(\bs{Y}|\boldsymbol{\mathcal{V}})=\prod_{a=0}^{\tau}\mathcal{P}(\bs{Y}|\bs{V}^{(a)})$ are applied. In addition, the distribution $\mathcal{P}(\boldsymbol{\mathcal{S}})$ is given by
\begin{align}
\mathcal{P}(\boldsymbol{\mathcal{S}})=\int \mathcal{P}(\boldsymbol{\mathcal{S}}|\boldsymbol{\mathcal{U}})\mathbb{E}_{\boldsymbol{\mathcal{H}},\boldsymbol{\mathcal{X}}}
\{\delta(\boldsymbol{\mathcal{U}}-\boldsymbol{\mathcal{HX}})\}\text{d}\boldsymbol{\mathcal{U}},
\label{Equ:P_s}
\end{align}
where $\boldsymbol{\mathcal{U}}=\{\bs{U}^{(a)},\forall a\}$, $\boldsymbol{\mathcal{H}}=\{\bs{H}^{(a)},\forall a\}$, and $\boldsymbol{\mathcal{X}}=\{\bs{X}^{(a)}\}$. Note that the information of first layer is involved in the prior distribution $\mathcal{P}(\boldsymbol{\mathcal{S}})$ of the second layer.

As can be seen from $\mathbb{E}\{\mathcal{P}^{\tau}(\bs{Y})\}$ in (\ref{Equ:EY}), the key challenge is the computation of the term $\mathbb{E}_{\boldsymbol{\mathcal{C}},\boldsymbol{\mathcal{S}}}
\{\delta(\boldsymbol{\mathcal{V}}-\boldsymbol{\mathcal{CS}})\}$. In large system limit, where the dimensions of the system go into infinity, the central limit theorem (CLT) implies that the term $v_{pk}^{(a)}=\sum_{m=1}^Mc^{(a)}_{pm}s_{mk}^{(a)}$ limits to a Gaussian distribution with zero mean and covariance
\begin{align}
\mathbb{E}_{\boldsymbol{\mathcal{C}},\boldsymbol{\mathcal{S}}}\{v_{pk}^{(a)}v_{pk}^{(b)}\}
&=\mathbb{E}_{\boldsymbol{\mathcal{C}},\boldsymbol{\mathcal{S}}}\left\{\left(\sum_{m=1}^{M}c^{(a)}_{pm}s_{mk}^{(a)}\right)\left(\sum_{j=1}^{M}c^{(b)}_{pj}s_{jk}^{(b)}\right)\right\}
\label{Equ:V1}\\
&=\mathbb{E}_{\boldsymbol{\mathcal{C}},\boldsymbol{S}}\left\{\frac{1}{M}\left(\sum_{m=1}^{M}c_{pm}^{(a)}c_{pm}^{(b)}\right)\left(\sum_{j=1}^{M}s_{jk}^{(a)}s_{jk}^{(b)}\right)\right\}.
\label{Equ:V2}
\end{align}

To average over $\mathcal{P}(\boldsymbol{\mathcal{C}},\boldsymbol{\mathcal{S}})$ in (\ref{Equ:EY}), we introduce two $(\tau+1)\times (\tau+1)$ auxiliary matrices $\bs{Q}_C$ and $\bs{Q}_S$ defined by
\begin{align}
1&=\int \prod_{p=1}^{P}\prod_{0\leq a\leq b}^{\tau}\delta\left(MQ_C^{ab}-\sum_{m=1}^Mc_{pm}^{(a)}c_{pm}^{(b)}\right)\text{d}Q_C^{ab},\\
1&=\int \prod_{k=1}^{K}\prod_{0\leq a\leq b}^{\tau}\delta\left(MQ_S^{ab}-\sum_{m=1}^Ms_{mk}^{(a)}s_{mk}^{(b)}\right)\text{d}Q_S^{ab},
\end{align}
whose probability measures are represented as
\begin{align}
\mathcal{P}(\bs{Q}_C)&=\mathbb{E}_{\boldsymbol{\mathcal{C}}}\left\{\prod_{p=1}^{P}\prod_{0\leq a\leq b}^{\tau}\delta\left(MQ_C^{ab}-\sum_{m=1}^Mc_{pm}^{(a)}c_{pm}^{(b)}\right)\right\},\\
\mathcal{P}(\bs{Q}_S)&=\mathbb{E}_{\boldsymbol{\mathcal{S}}}\left\{\prod_{k=1}^{K}\prod_{0\leq a\leq b}^{\tau}\delta\left(MQ_S^{ab}-\sum_{m=1}^Ms_{mk}^{(a)}s_{mk}^{(b)}\right)\right\}.
\end{align}
Applying the probability measure of $(\bs{Q}_C,\bs{Q}_S)$ to replace the distribution of $(\boldsymbol{\mathcal{C}},\boldsymbol{\mathcal{S}})$ in (\ref{Equ:EY}) yields
\begin{align}
\mathbb{E}\{\mathcal{P}^{\tau}(\bs{Y})\}
=\mathbb{E}_{\bs{Q}_C, \bs{Q}_S}\left\{\left(\int \prod_{a=0}^{\tau}p(y|v^{(a)})\mathcal{N}(\bs{v}|\bs{0},M\bs{Q}_C\odot \bs{Q}_S)\text{d}\bs{v}\text{d}y\right)^{PK}\right\},
\end{align}
with $\bs{v}=\{v^{(a)},\forall a\}$ and $\odot$ being componentwise multiplication.

We note that $Q_C^{ab}=\frac{1}{M}\sum_{m=1}^Mc_{pm}^{(a)}c_{pm}^{(b)}$ is the sum of a large number of $\textit{i.i.d.}$ random variables. For $Q_S^{ab}=\frac{1}{M}\sum_{m=1}^Ms_{mk}^{(a)}s_{mk}^{(b)}$, there actually exists correlation in $\bs{s}_k^{(a)}=\{s_{mk}^{(a)},\forall m\}$ due to the linear mixing space. Fortunately, in large system limit, the CLT allows us to treat $\bs{u}_k^{(a)}=\bs{H}^{(a)}\bs{x}_k^{(a)}$ as Gaussian with zero mean and covariance matrix $\chi_x\bs{H}^{(a)}(\bs{H}^{(a)})^{\text{T}}$, which limits to diagonal matrix, i.e.,  $\chi_x\bs{H}^{(a)}(\bs{H}^{(a)})^{\text{T}}\rightarrow N\chi_x\chi_h\mathbf{I}$. In addition, $\mathcal{P}(\bs{s}_k^{(a)}|\bs{u}_k^{(a)})$ is componentwise. Thus, $Q_S^{ab}$ can be regarded as the sum of a large number of independent variables approximately. In the sequel, both of probability of  $\bs{Q}_C$ and $\bs{Q}_S$ satisfy large derivation theory (LDT) \cite[Chapter 2.2]{touchette2011basic}, \cite{ellis2007entropy}, which implies
\begin{align}
\mathcal{P}(\bs{Q}_C)\approx e^{-PM\mathcal{R}^{(\tau)}(\bs{Q}_C)}, \ \mathcal{P}(\bs{Q}_S)\approx e^{-MK\mathcal{R}^{(\tau)}(\bs{Q}_S)},
\end{align}
where $\mathcal{R}^{(\tau)}(\bs{Q}_C)$ and $\mathcal{R}^{(\tau)}(\bs{Q}_S)$ are rate functions from the  Legendre-Fenchel transform of $\log \mathbb{E}_{\bs{c}}\left\{\exp \left(\bs{c}^{\text{T}}\hat{\bs{Q}}_C\bs{c}\right)\right\}$ and $\frac{1}{MK}\log \mathbb{E}_{\mathcal{\bs{S}}}\left\{\exp \left(\sum_{m=1}^M\sum_{k=1}^K\bs{s}_{mk}^{\text{T}}\hat{\bs{Q}}_S\bs{s}_{mk}\right)\right\}$, respectively.
\begin{align}
\mathcal{R}^{(\tau)}(\bs{Q}_C)
&=\sup_{\hat{\bs{Q}}_C}\left\{\text{tr}(\hat{\bs{Q}}_C\bs{Q}_C)-\log \mathbb{E}_{\bs{c}}\left\{\exp \left(\bs{c}^{\text{T}}\hat{\bs{Q}}_C\bs{c}\right)\right\}\right\},\\
\mathcal{R}^{(\tau)}(\bs{Q}_S)
&=\sup_{\hat{\bs{Q}}_S}\left\{\text{tr}(\hat{\bs{Q}}_S\bs{Q}_S)-\frac{1}{MK}\log \mathbb{E}_{\mathcal{\bs{S}}}\left\{\exp \left(\sum_{m=1}^M\sum_{k=1}^K\bs{s}_{mk}^{\text{T}}\hat{\bs{Q}}_S\bs{s}_{mk}\right)\right\}\right\},
\end{align}
where $\bs{c}=\{c^{(a)},\forall a\}$ and $\bs{s}=\{s^{(a)},\forall a\}$. Additionally, another interpretation of rate function using Fourier representation can be found in Appendix \ref{Appendix:E}.

By Varadhan's theorem \cite[Section 2.4]{touchette2011basic}, from (\ref{Equ:FreeEnergy}) the following can be obtained
\begin{align}
\frac{1}{N^2}\log \mathbb{E}\left\{\mathcal{P}^{\tau}(\bs{Y})\right\}
&=\sup_{\bs{Q}_S,\bs{Q}_C}\left\{\frac{PK}{N^2}G^{(\tau)}(\bs{Q}_C,\bs{Q}_S)-\frac{PM}{N^2}\mathcal{R}^{(\tau)}(\bs{Q}_C)-\frac{MK}{N^2}\mathcal{R}^{\tau}(\bs{Q}_S)\right\}\\
\nonumber
&=\underset{\bs{Q}_C,\hat{\bs{Q}}_C,\bs{Q}_S,\hat{\bs{Q}}_S}{\text{Extr}}\left\{
\frac{PK}{N^2}G^{(\tau)}(\bs{Q}_C,\bs{Q}_S)-\frac{PM}{N^2}\text{tr}(\bs{Q}_C\hat{\bs{Q}}_C)+\frac{PM}{N^2}\log \mathbb{E}_{\bs{c}}\left\{\exp \left(\bs{c}^{\text{T}}\hat{\bs{Q}}_C\bs{c}\right)\right\}\right.\\
&\qquad \left.  -\frac{MK}{N^2}\text{tr}(\bs{Q}_S\hat{\bs{Q}}_S) +\frac{1}{N^2}\log \mathbb{E}_{\boldsymbol{\mathcal{S}}}\left\{\exp \left(\sum_{m=1}^M\sum_{k=1}^K\bs{s}_{mk}^{\text{T}}\hat{\bs{Q}}_S\bs{s}_{mk}\right)\right\}\right\},
\label{Equ:F_kappa}
\end{align}
where `Extr' denotes extremum points and
\begin{align}
G^{(\tau)}(\bs{Q}_C,\bs{Q}_S)=\log \int \mathcal{P}(y|\bs{v})\mathcal{N}(\bs{v}|\bs{0},M\bs{Q}_C\odot \bs{Q}_S)\text{d}\bs{v}\text{d}y,
\end{align}
with $\mathcal{P}(y|\bs{v})=\prod_{a=0}^{\tau}\mathcal{P}(y|v^{(a)})$.

\subsection{Move to Previous Layer}
\label{Sec:Previouslayer}
In fact, the key challenge of computing (\ref{Equ:F_kappa}) is the term $\frac{1}{N^2}\log \mathbb{E}_{\boldsymbol{\mathcal{S}}}\left\{\exp \left(\sum_{m=1}^M\sum_{k=1}^K\bs{s}_{mk}^{\text{T}}\hat{\bs{Q}}_S\bs{s}_{mk}\right)\right\}$. Similar to dealing with $\bs{V}$ (\ref{Equ:V1})-(\ref{Equ:V2}), the CLT allows us to treat $u_{mk}^{(a)}=\sum_{n=1}^Nh_{mn}^{(a)}x_{nk}^{(a)}$ as Gaussian variable with zero mean and covariance
\begin{align}
\mathbb{E}\left\{u^{(a)}_{mk}u^{(b)}_{mk}\right\}
&=
\mathbb{E}_{\boldsymbol{\mathcal{H}},\boldsymbol{\mathcal{X}}}\left\{\frac{1}{N}\left(\sum_{n=1}^Nh_{mn}^{(a)}h_{mn}^{(b)}\right)\left(\sum_{i=1}^Nx_{ik}^{(a)}x_{ik}^{(b)}\right)\right\}.
\end{align}
To handle the expectation over $(\boldsymbol{\mathcal{X}},\boldsymbol{\mathcal{H}})$, we introduce the following two $(\tau+1)\times (\tau+1)$ auxiliary matrices $\bs{Q}_X$ and $\bs{Q}_H$
\begin{align}
1&=\int \prod_{m=1}^M \prod_{0\leq a\leq b}^{\tau}\delta\left(NQ_H^{ab}-\sum_{n=1}^Nh_{mn}^{(a)}h_{mn}^{(b)}\right)\text{d}Q_H^{ab},\\
1&=\int \prod_{k=1}^K \prod_{0\leq a\leq b}^{\tau}\delta\left(NQ_X^{ab}-\sum_{n=1}^Nx_{nk}^{(a)}x_{nk}^{(b)}\right)\text{d}Q_X^{ab},
\end{align}
whose probability measures and rate functions are given by
\begin{align}
\mathcal{P}(\bs{Q}_H)&=\mathbb{E}_{\boldsymbol{\mathcal{H}}}\left\{\prod_{m=1}^{M}\prod_{0\leq a\leq b}^{\tau}\delta\left(NQ_H^{ab}-\sum_{n=1}^Nh_{mn}^{(a)}h_{mn}^{(b)}\right)\right\},\\
\!\!\!\!\!\!\!\!
\mathcal{P}(\bs{Q}_X)&=\mathbb{E}_{\boldsymbol{\mathcal{X}}}\left\{\prod_{k=1}^{K}\prod_{0\leq a\leq b}^{\tau}\delta\left(NQ_X^{ab}-\sum_{n=1}^Nx_{nk}^{(a)}x_{nk}^{(b)}\right)\right\},\\
\mathcal{R}^{(\tau)}(\bs{Q}_H)
&=\sup_{\hat{\bs{Q}}_H}\left\{\text{tr}(\hat{\bs{Q}}_H\bs{Q}_H)-\log \mathbb{E}_{\bs{h}}\left\{\exp \left(\bs{h}^{\text{T}}\hat{\bs{Q}}_H\bs{h}\right)\right\}\right\},\\
\mathcal{R}^{(\tau)}(\bs{Q}_X)
&=\sup_{\hat{\bs{Q}}_X}\left\{\text{tr}(\hat{\bs{Q}}_X\bs{Q}_X)-\log \mathbb{E}_{\bs{x}}\left\{\exp \left(\bs{x}^{\text{T}}\hat{\bs{Q}}_X\bs{x}\right)\right\}\right\}.
\end{align}

The term in (\ref{Equ:F_kappa}) is thus written as
\begin{align}
\nonumber
&\frac{1}{N^2}\log \mathbb{E}_{\boldsymbol{\mathcal{S}}}\left\{\exp \left(\sum_{m=1}^M\sum_{k=1}^K\bs{s}_{mk}^{\text{T}}\hat{\bs{Q}}_S\bs{s}_{mk}\right)\right\}\\
=&\frac{1}{N^2}\log \mathbb{E}_{\bs{Q}_S,\bs{Q}_C}\left\{\left[\int \exp \left(\bs{s}^{\text{T}}\hat{\bs{Q}}_S\bs{s}\right)\mathcal{P}(\bs{s}|\bs{u}) \mathcal{N}(\bs{u}|\bs{0},N\bs{Q}_H\odot\bs{Q}_X)\text{d}\bs{u}\text{d}\bs{s}\right]^{MK}\right\}.
\label{Equ:UndetermineTerm}
\end{align}
Further, by large partial theory and Varadhan's theorem again, the equation above becomes
\begin{align}
(\ref{Equ:UndetermineTerm})&=\sup_{\bs{Q}_H,\bs{Q}_X}\left\{\frac{MK}{N^2}G^{(\tau)}(\hat{\bs{Q}}_S,\bs{Q}_H,\bs{Q}_X)
-\frac{MN}{N^2}\mathcal{R}^{(\tau)}(\bs{Q}_H)-\frac{NK}{N^2}\mathcal{R}^{(\tau)}(\bs{Q}_X)\right\}\\
\nonumber
&=\underset{\bs{Q}_H,\bs{Q}_X,\hat{\bs{Q}}_H,\hat{\bs{Q}}_X}{\text{Extr}}\left\{\frac{MK}{N^2}G^{(\tau)}(\hat{\bs{Q}}_S,\bs{Q}_H,\bs{Q}_X)
-\frac{MN}{N^2}\text{tr}(\bs{Q}_H\hat{\bs{Q}}_H)+\frac{MN}{N^2}\log \mathbb{E}_{\bs{h}}\left\{\exp \left(\bs{h}^{\text{T}}\hat{\bs{Q}}_H\bs{h}\right)\right\}
\right.\\
&\  \left. \qquad -\frac{NK}{N^2}\text{tr}(\bs{Q}_X\hat{\bs{Q}}_X)+\frac{NK}{N^2}\log \mathbb{E}_{\bs{x}}\left\{\exp \left(\bs{x}^{\text{T}}\hat{\bs{Q}}_X\bs{x}\right)\right\}\right\},
\label{Equ:determineTerm}
\end{align}
where
\begin{align}
G^{(\tau)}(\hat{\bs{Q}}_S,\bs{Q}_H,\bs{Q}_X)&=\log \int \exp \left(\bs{s}^{\text{T}}\hat{\bs{Q}}_S\bs{s}\right) \mathcal{P}(\bs{s}|\bs{u}) \mathcal{N}(\bs{u}|\bs{0},N\bs{Q}_H\odot \bs{Q}_X)\text{d}\bs{u}\text{d}\bs{s}.
\end{align}

Meanwhile, substituting (\ref{Equ:determineTerm}) into (\ref{Equ:F_kappa}) yields
\begin{align}
\nonumber
&\frac{1}{N^2}\log \mathbb{E}\left\{\mathcal{P}^{\tau}(\bs{Y})\right\}\\
\nonumber
&=\underset{\bs{Q}_C,\hat{\bs{Q}}_C,\bs{Q}_S,\hat{\bs{Q}}_S,\bs{Q}_H,\hat{\bs{Q}}_H,\bs{Q}_X,\hat{\bs{Q}}_X}{\text{Extr}}\left\{
\frac{PK}{N^2}G^{(\tau)}(\bs{Q}_C,\bs{Q}_S)-\frac{PM}{N^2}\text{tr}(\bs{Q}_C\hat{\bs{Q}}_C)+\frac{PM}{N^2}\log \mathbb{E}_{\bs{c}}\left\{\exp \left(\bs{c}^{\text{T}}\hat{\bs{Q}}_C\bs{c}\right)\right\}\right.\\
\nonumber
& \left.  \qquad -\frac{MK}{N^2}\text{tr}(\bs{Q}_S\hat{\bs{Q}}_S)+\frac{MK}{N^2}G^{(\tau)}(\hat{\bs{Q}}_S,\bs{Q}_H,\bs{Q}_X)-\frac{MN}{N^2}\text{tr}(\bs{Q}_H\hat{\bs{Q}}_H)+\frac{MN}{N^2}\log \mathbb{E}_{\bs{h}}\left\{\exp \left(\bs{h}^{\text{T}}\hat{\bs{Q}}_H\bs{h}\right)\right\}\right.\\
&\qquad \left.-\frac{NK}{N^2}\text{tr}(\bs{Q}_X\hat{\bs{Q}}_X)+\frac{NK}{N^2}\log \mathbb{E}_{\bs{x}}\left\{\exp \left(\bs{x}^{\text{T}}\hat{\bs{Q}}_X\bs{x}\right)\right\}\right\}\\
&=\underset{\bs{Q}_C,\hat{\bs{Q}}_C,\bs{Q}_S,\hat{\bs{Q}}_S,\bs{Q}_H,\hat{\bs{Q}}_H,\bs{Q}_X,\hat{\bs{Q}}_X}{\text{Extr}}\mathcal{T}(\bs{Q}_C,\hat{\bs{Q}}_C,\bs{Q}_S,\hat{\bs{Q}}_S,\bs{Q}_H,\hat{\bs{Q}}_H,\bs{Q}_X,\hat{\bs{Q}}_X).
\label{Equ:SaddlePoints}
\end{align}

We first seek the saddle points of $\mathcal{T}(\cdot)$ defined  in (\ref{Equ:SaddlePoints}) w.r.t. $\bs{Q}_C$, $\hat{\bs{Q}}_C$, $\bs{Q}_S$, $\hat{\bs{Q}}_S$, $\bs{Q}_H$, $\hat{\bs{Q}}_H$, $\bs{Q}_X$, and $\hat{\bs{Q}}_X$. Applying the following note \footnote{The partial derivation of Gaussian vector distribution $\mathcal{N}(\bs{x}|\bs{0},\chi \bs{Q}_H\odot\bs{Q}_X)$ w.r.t. $\bs{Q}_H$ is given by
\begin{align*}
\nonumber
\frac{\partial \mathcal{N}(\bs{x}|\bs{0},\chi\bs{Q}_H\odot \bs{Q}_X)}{\partial \bs{Q}_H}&=-\frac{\mathcal{N}(\bs{x}|\bs{0},\chi\bs{Q}_H\odot\bs{Q}_X)}{2}\bs{Q}_X\odot \left[(\bs{Q}_H\odot \bs{Q}_X)^{-1}-\frac{1}{\chi}(\bs{Q}_H\odot \bs{Q}_X)^{-1}\bs{xx}^{\text{T}}(\bs{Q}_H\odot \bs{Q}_X)^{-1}\right].
\end{align*}
Proof sees Appendix \ref{Appendix:G}.
}, we obtain the saddle point equations from the free energy
\begin{subequations}
\begin{align}
\hat{\bs{Q}}_C&=-\frac{\bs{Q}_S}{2\alpha\beta}\odot \left((\bs{Q}_C\odot \bs{Q}_S)^{-1}-\frac{1}{M}(\bs{Q}_C\odot \bs{Q}_S)^{-1}\mathbb{E}_{\bs{v}}\{\bs{vv}^{\text{T}}\}(\bs{Q}_C\odot \bs{Q}_S)^{-1}\right),
\label{Equ:SD1}\\
\bs{Q}_C&=\frac{\mathbb{E}_{\bs{c}}\left\{\bs{cc}^{\text{T}}\exp \left(\bs{c}^{\text{T}}\hat{\bs{Q}}_C\bs{c}\right)\right\}}
{\mathbb{E}_{\bs{c}}\left\{\exp \left(\bs{c}^{\text{T}}\hat{\bs{Q}}_C\bs{c}\right)\right\}},
\label{Equ:SD2}\\
\hat{\bs{Q}}_S&=-\frac{\gamma \bs{Q}_C}{2}\odot \left((\bs{Q}_C\odot \bs{Q}_S)^{-1}-\frac{1}{M}(\bs{Q}_C\odot \bs{Q}_S)^{-1} \mathbb{E}_{\bs{v}}\{\bs{vv}^{\text{T}}\}(\bs{Q}_C\odot \bs{Q}_S)^{-1}\right),
\label{Equ:SD3}\\
\bs{Q}_S&=\mathbb{E}_{\bs{s}}\{\bs{ss}^{\text{T}}\},
\label{Equ:SD4}\\
\hat{\bs{Q}}_H&=-\frac{\bs{Q}_X}{2\alpha}\odot \left((\bs{Q}_H\odot\bs{Q}_X)^{-1}-\frac{1}{N}(\bs{Q}_H\odot \bs{Q}_X)^{-1} \mathbb{E}_{\bs{u}}\{\bs{uu}^{\text{T}}\}(\bs{Q}_H\odot\bs{Q}_X)^{-1}\right),
\label{Equ:SD5}\\
\bs{Q}_H&=\frac{\mathbb{E}_{\bs{h}}\left\{\bs{hh}^{\text{T}}\exp \left(\bs{h}^{\text{T}}\hat{\bs{Q}}_H\bs{h}\right)\right\}}{\mathbb{E}_{\bs{h}}\left\{\exp \left(\bs{h}^{\text{T}}\hat{\bs{Q}}_H\bs{h}\right)\right\}},
\label{Equ:SD6}\\
\hat{\bs{Q}}_X&=-\frac{\beta \bs{Q}_H}{2}\odot \left((\bs{Q}_H\odot \bs{Q}_X)^{-1}-\frac{1}{N}(\bs{Q}_H\odot \bs{Q}_X)^{-1} \mathbb{E}_{\bs{u}}\{\bs{uu}^{\text{T}}\}(\bs{Q}_H\odot \bs{Q}_X)^{-1}\right),
\label{Equ:SD7}\\
\bs{Q}_X&=\frac{\mathbb{E}_{\bs{x}}\left\{\bs{xx}^{\text{T}}\exp \left(\bs{x}^{\text{T}}\hat{\bs{Q}}_X\bs{x}\right)\right\}}{\mathbb{E}_{\bs{x}}\left\{\exp \left(\bs{x}^{\text{T}}\hat{\bs{Q}}_X\bs{x}\right)\right\}},
\label{Equ:SD8}
\end{align}
\end{subequations}
where the expectations in (\ref{Equ:SD1}), (\ref{Equ:SD4}), and (\ref{Equ:SD5}) are taken over
\begin{align}
\mathcal{P}(\bs{v})&=\frac{\int \prod_{a=0}^{\tau}\mathcal{P}(y|v^{(a)})\mathcal{N}(\bs{v}|\bs{0},M\bs{Q}_C\odot \bs{Q}_S)\text{d}y}{\int \prod_{a=0}^{\tau}\mathcal{P}(y|v^{(a)})\mathcal{N}(\bs{v}|\bs{0},M\bs{Q}_C\odot \bs{Q}_S)\text{d}\bs{v}\text{d}y},\\
\mathcal{P}(\bs{s})&=\frac{\int \exp\left(\bs{s}^{\text{T}}\hat{\bs{Q}}_S\bs{s}\right)\mathcal{P}(\bs{s}|\bs{u})\mathcal{N}(\bs{u}|\bs{0},N\bs{Q}_H\odot \bs{Q}_X)\text{d}\bs{u}}{\int \exp\left(\bs{s}^{\text{T}}\hat{\bs{Q}}_S\bs{s}\right)\mathcal{P}(\bs{s}|\bs{u})\mathcal{N}(\bs{u}|\bs{0},N\bs{Q}_H\odot \bs{Q}_X)\text{d}\bs{u}\text{d}\bs{s}},\\
\mathcal{P}(\bs{u})&=\frac{\int \exp \left(\bs{s}^{\text{T}}\hat{\bs{Q}}_S\bs{s}\right)\mathcal{P}(\bs{s}|\bs{u})\mathcal{N}(\bs{u}|\bs{0},N\bs{Q}_H\odot \bs{Q}_X)\text{d}\bs{s}}{\int \exp \left(\bs{s}^{\text{T}}\hat{\bs{Q}}_S\bs{s}\right)\mathcal{P}(\bs{s}|\bs{u})\mathcal{N}(\bs{u}|\bs{0},N\bs{Q}_H\odot \bs{Q}_X)\text{d}\bs{u}\text{d}\bs{s}}.
\end{align}

\subsection{Replica Symmetric Solution}
\label{Sec:RS}
In fact, it is prohibitive to solve the joint equations (\ref{Equ:SD1})-(\ref{Equ:SD8}) except in the simplest cases such as all priors and transition distributions being Gaussian. To address this issue, we postulate that the solutions of those saddle point equations satisfies \textit{replica symmetry} \cite{Kabashima2016Phase, mezard1987spin}, i.e.,
\begin{align}
\bs{Q}_X&=(\chi_x-q_x)\mathbf{I}+q_x\mathbf{11}^{\text{T}},\  \hat{\bs{Q}}_X=(\hat{\chi}_x-\hat{q}_x)\mathbf{I}+\hat{q}_x\mathbf{11}^{\text{T}},\\
\bs{Q}_H&=(\chi_h-q_h)\mathbf{I}+q_h\mathbf{11}^{\text{T}}, \ \hat{\bs{Q}}_H=(\hat{\chi}_h-\hat{q}_h)\mathbf{I}+\hat{q}_h\mathbf{11}^{\text{T}},\\
\bs{Q}_S&=(\chi_s-q_s)\mathbf{I}+q_s\mathbf{11}^{\text{T}}, \ \ \hat{\bs{Q}}_S=(\hat{\chi}_s-\hat{q}_s)\mathbf{I}+\hat{q}_s\mathbf{11}^{\text{T}},\\
\bs{Q}_C&=(\chi_c-q_c)\mathbf{I}+q_c\mathbf{11}^{\text{T}}, \  \ \hat{\bs{Q}}_C=(\hat{\chi}_c-\hat{q}_c)\mathbf{I}+\hat{q}_c\mathbf{11}^{\text{T}},
\end{align}
where $\bs{11}^{\text{T}}$ denotes $(\tau+1)\times (\tau+1)$ matrix with it all elements being 1.
Based on the replica symmetry assumption above, the terms $\mathbb{E}_{\bs{v}}\{\bs{vv}^{\text{T}}\}$ and $\mathbb{E}_{\bs{u}}\{\bs{uu}^{\text{T}}\}$ also have replica symmetry structure, i.e.,
\begin{align}
\bs{Q}_V &=\mathbb{E}_{\bs{v}}\{\bs{vv}^{\text{T}}\}=(\chi_v-q_v)\mathbf{I}+q_v\mathbf{11}^{\text{T}},
\label{Equ:QV}\\
\bs{Q}_{U}&=\mathbb{E}_{\bs{u}}\{\bs{uu}^{\text{T}}\}=(\chi_u-q_u)\mathbf{I}+q_u\mathbf{11}^{\text{T}}.
\end{align}

We first determine the term $\bs{Q}_V$ in (\ref{Equ:SD1}) by evaluating ($\chi_v,q_v$), which are expressed as
\begin{align}
\chi_v&=\frac{\int (v^{(0)})^2\prod_{a=0}^{\tau}p(y|v^{(a)})\mathcal{N}(\bs{v}|0, M\bs{Q}_C\odot \bs{Q}_S)\text{d}\bs{v}\text{d}y}
{\int \prod_{a=0}^{\tau}p(y|v^{(a)})\mathcal{N}(\bs{v}|0, M\bs{Q}_C\odot \bs{Q}_S)\text{d}\bs{v}\text{d}y},\\
q_v&=\frac{\int v^{(0)}v^{(1)}\prod_{a=0}^{\tau}p(y|v^{(a)})\mathcal{N}(\bs{v}|0, M\bs{Q}_C\odot \bs{Q}_S)\text{d}\bs{v}\text{d}y}
{\int \prod_{a=0}^{\tau}p(y|v^{(a)})\mathcal{N}(\bs{v}|0, M\bs{Q}_C\odot \bs{Q}_S)\text{d}\bs{v}\text{d}y}.
\end{align}
Applying the matrix inverse lemma\footnote{
$(\bs{A}+\bs{BC})^{-1}=\bs{A}^{-1}-\bs{A}^{-1}\bs{B}(\mathbf{I}+\bs{CA}^{-1}\bs{B})^{-1}\bs{CA}^{-1}$.
}, the term $(M\bs{Q}_C\odot \bs{Q}_S)^{-1}$ in $\chi_v$ and $ q_v$ can be written as
\begin{align}
(M\bs{Q}_S\odot \bs{Q}_C)^{-1}
&=\frac{1}{M(\chi_s\chi_c-q_sq_c)}\mathbf{I}-\frac{q_sq_c}{M(\chi_s\chi_c-q_sq_c)(\chi_s\chi_c+\tau q_cq_s)}\mathbf{11}^{\text{T}}.
\end{align}
We define $A=\frac{1}{M(\chi_s\chi_c-q_sq_c)}$ and $B=\frac{q_sq_c}{M(\chi_s\chi_c-q_sq_c)(\chi_s\chi_c+\tau q_cq_s)}$.
Further by Hubbard-Stratonovich transform \footnote{$e^{x^2}=\sqrt{\frac{\eta}{2\pi}}\int e^{-\frac{\eta}{2}\xi^2+\sqrt{2\eta }x\xi}\text{d}\xi$, for $\eta>0$.}, we decouple the coupled exponent component
\begin{align}
\exp \left(-\frac{1}{2}\bs{v}^{\text{T}}(M\bs{Q}_C\odot \bs{Q}_S)^{-1}\bs{v}\right)
&=\exp \left[-\frac{A}{2}\sum_{a=0}^{\tau}(v^{(a)})^2+\left(\sqrt{\frac{B}{2}}\sum_{a=0}^{\tau}v^{(a)}\right)^2\right]\\
&=\int \sqrt{\frac{\eta}{2\pi}}\exp \left[-\frac{A}{2}\sum_{a=0}^{\tau}(v^{(a)})^2-\frac{\eta}{2}\xi^2+\sqrt{\eta B}\sum_{a=0}^{\tau}v^{(a)}\xi\right]\text{d}\xi.
\end{align}

By this decoupling operation, we calculate denominator and numerator of $\chi_v$ in (\ref{Equ:QV}), respectively
\begin{align}
\nonumber
&\lim_{\tau\rightarrow 0}\int \prod_{a=0}^{\tau}\mathcal{P}(y|v^{(a)})\mathcal{N}(\bs{v}|\bs{0},N\bs{Q}_C\odot \bs{Q}_S)\text{d}\bs{v}\text{d}y
\label{Equ:chiv_B}\\
=&\lim_{\tau\rightarrow 0}C\int _y \left[\int_v\mathcal{P}(y|v)\exp \left(-\frac{A}{2}v^2+\sqrt{\eta B}v\xi\right)\text{d}v\right]^{\tau+1} \sqrt{\frac{\eta}{2\pi}}\exp \left(-\frac{\eta}{2}\xi^2\right)\text{d}\xi\text{d}y\\
=&\lim_{\tau\rightarrow 0}C \sqrt{\frac{2\pi}{A-B}},
\label{Equ:chiv_Num}\\
\nonumber
&\lim_{\tau\rightarrow 0}\int (v^{(0)})^2\prod_{a=0}^{\tau}\mathcal{P}(y|v^{(a)})\mathcal{N}(\bs{v}|\bs{0},N\bs{Q}_C\odot \bs{Q}_S)\text{d}\bs{v}\text{d}y\\
=&\lim_{\tau\rightarrow 0}C\sqrt{\frac{2\pi}{A-B}}\frac{1}{A-B}.
\label{Equ:chiv_Den}
\end{align}
Meanwhile, the denominator of $q_v$ in (\ref{Equ:QV}) is evaluated as
\begin{align}
\nonumber
&\lim_{\tau\rightarrow 0}\int v^{(0)}v^{(1)}\prod_{a=0}^{\tau}\mathcal{P}(y|v^{(a)})\mathcal{N}(\bs{v}|\bs{0},N\bs{Q}_C\odot \bs{Q}_S)\text{d}\bs{v}\text{d}y\\
=&\lim_{\tau\rightarrow 0}C\sqrt{\frac{2\pi}{A-B}}\int \frac{[\int v \mathcal{P}(y|v)\mathcal{N}\left(v|\sqrt{\frac{B}{A(A-B)}}\xi,\frac{1}{A}\right)\text{d}v]^2}{\int \mathcal{P}(y|v)\mathcal{N}\left(v|\sqrt{\frac{B}{A(A-B)}}\xi,\frac{1}{A}\right)\text{d}v}\text{D}\xi\text{d}y,
\label{Equ:qv_Num}
\end{align}
where $C=(2\pi)^{-\frac{\tau+1}{2}}\det(M\bs{Q}_S\odot \bs{Q}_C)^{-\frac{1}{2}}$.

The parameter $\chi_v$ is obtained by combining (\ref{Equ:chiv_Num}) and (\ref{Equ:chiv_Den}), and $q_v$ is obtained by combining (\ref{Equ:qv_Num}) and  (\ref{Equ:chiv_Den}). Additionally, the terms involving $\tau$ are directly replaced by themselves under the restriction of $\tau=0$. As a result, we get
\begin{align}
\chi_v&=M\chi_c\chi_s,\\
q_v&=\int\frac{\left[\int v\mathcal{P}(y|v)\mathcal{N}(v|\sqrt{Mq_sq_c}\xi,M(\chi_s\chi_c-q_sq_c))\text{d}v\right]^2}{\int \mathcal{P}(y|v)\mathcal{N}(v|\sqrt{Mq_sq_c}\xi,M(\chi_s\chi_c-q_sq_c))\text{d}v}\text{D}\xi\text{d}y.
\end{align}

Due to the replica symmetry structure, solving the equation (\ref{Equ:SD1}) and (\ref{Equ:SD3}) yields
\begin{align}
\hat{\chi}_c&=0,\\
\hat{q}_c&=\frac{q_s}{2\alpha \beta}\frac{q_v-Mq_sq_c}{M(\chi_s\chi_c-q_cq_s)^2},\\
\hat{\chi}_s&=0,\\
\hat{q}_s&=\frac{\gamma q_c}{2}\frac{q_v-Mq_sq_c}{M(\chi_s\chi_c-q_cq_s)^2}.
\end{align}

For (\ref{Equ:SD4}), we calculate the inverse term $(N\bs{Q}_H\odot \bs{Q}_X)^{-1}$ using matrix inverse lemma
\begin{align}
(N\bs{Q}_H\odot \bs{Q}_X)^{-1}
&=\frac{1}{N(\chi_h\chi_x-q_hq_x)}\mathbf{I}-\frac{q_hq_x}{N(\chi_h\chi_x-q_hq_x)(\chi_h\chi_x+\tau q_hq_x)}\mathbf{11}^{\text{T}}.
\end{align}
We define $E=\frac{1}{N(\chi_h\chi_x-q_hq_x)}$ and $F=\frac{q_hq_x}{N(\chi_h\chi_x-q_hq_x)(\chi_h\chi_x+\tau q_hq_x)}$.
Also, applying Hubbard-stratonvich transform the coupled exponent components in $\mathcal{P}(\bs{s})$ and $\mathcal{P}(\bs{u})$  can be decoupled as
\begin{align}
\exp \left(-\frac{1}{2}\bs{u}^{\text{T}}(N\bs{Q}_H\odot \bs{Q}_X)^{-1}\bs{u}\right)
&=\int \sqrt{\frac{\eta}{2\pi}}\exp \left(-\frac{E}{2}\sum_{a=0}^{\tau}(u^{(a)})^2-\frac{\eta}{2}\xi^2+\sqrt{\eta F}\xi\sum_{a=0}^{\tau}u^{(a)}\right)\text{d}\xi,
\label{Equ:Decoupling1}\\
 \exp \left(\bs{s}^{\text{T}}\hat{\bs{Q}}_S\bs{s}\right)
&=\int \sqrt{\frac{\varrho}{2\pi}}\exp \left(-\hat{q}_s\sum_{a=0}^{\tau}(s^{(a)})^2-\frac{\varrho}{2}\zeta^2+\sqrt{2\varrho\hat{q}_s}\zeta\sum_{a=0}^{\tau}s^{(a)}\right)\text{d}\zeta.
\label{Equ:Decoupling2}
\end{align}

Similar to the computation of $\chi_v$ and $q_v$ in (\ref{Equ:chiv_B})-(\ref{Equ:qv_Num}), we calculate the denominators and numerators of $\chi_s$, $q_s$, $\chi_u$, and $q_u$, respectively. Those parameters can be obtained by combining their denominators and numerators, and by setting $\tau=0$, which yields
\begin{align}
\chi_s&=\int s^2\mathcal{P}(s|u)\mathcal{N}(u|0,N\chi_x\chi_h)\text{d}u\text{d}s,\\
q_s&=\int \frac{\left[\int s\mathcal{N}_{s|u}(\sqrt{Nq_hq_x}\xi,N(\chi_h\chi_x-q_hq_x),\zeta,\frac{1}{2\hat{q}_s})\text{d}u\text{d}s\right]^2}
{\int \mathcal{N}_{s|u}(\sqrt{Nq_hq_x}\xi,N(\chi_h\chi_x-q_hq_x),\zeta,\frac{1}{2\hat{q}_s})\text{d}u\text{d}s}\text{D}\xi\text{d}\zeta,\\
\chi_u&=N\chi_x\chi_h,\\
q_u&=\int \frac{\left[\int u\mathcal{N}_{s|u}(\sqrt{Nq_hq_x}\xi,N(\chi_h\chi_x-q_hq_x),\zeta,\frac{1}{2\hat{q}_s})\text{d}s\text{d}u\right]^2}
{\int \mathcal{N}_{s|u}(\sqrt{Nq_hq_x}\xi,N(\chi_h\chi_x-q_hq_x),\zeta,\frac{1}{2\hat{q}_s})\text{d}s\text{d}u}\text{D}\xi\text{d}\zeta.
\end{align}
where $\mathcal{N}_{s|u}(a,A,b,B)=\mathcal{P}(s|u)\mathcal{N}(u|a,A)\mathcal{N}(s|b,B)$. The detailed derivation of computing the parameters $(\chi_s,q_s,\chi_u,q_u)$ is given in Appendix \ref{Appendix:F}.

By replica symmetry structure, solving the equations (\ref{Equ:SD5}) and (\ref{Equ:SD7}) yields
\begin{align}
\hat{\chi}_h&=0,\\
\hat{q}_h&=\frac{q_x}{2\alpha}\frac{q_u-Nq_hq_x}{N(\chi_h\chi_x-q_hq_x)^2},\\
\hat{\chi}_x&=0,
\label{Equ:chi_x}\\
\hat{q}_x&=\frac{\beta q_h}{2}\frac{q_u-Nq_hq_x}{N(\chi_h\chi_x-q_hq_x)^2}.
\end{align}

 We move to the computation of the remaining equations, i.e., (\ref{Equ:SD2}), (\ref{Equ:SD6}), and (\ref{Equ:SD8}). Here, we only give the procedures of evaluating (\ref{Equ:SD8}) while the evaluations of (\ref{Equ:SD2}) and (\ref{Equ:SD6}) are the same as that of (\ref{Equ:SD8}). By the fact $\hat{\chi}_x=0$ and  Hubbard-Stratonovich transform, we have
\begin{align}
\lim_{\tau \rightarrow 0}\mathbb{E}_{\bs{x}}\left\{\exp \left(\bs{x}^{\text{T}}\hat{\bs{Q}}_X\bs{x}\right)\right\}
=&\int\sqrt{\frac{\eta}{2\pi}}\exp \left[-\frac{\eta}{2}\left(\xi'-\sqrt{\frac{2\hat{q}_x}{\eta}}x\right)^2\right]\mathcal{P}(x)\text{d}x\text{d}\xi'\\
\overset{(a)}{=}&\int \mathcal{N}\left(\xi|x,\frac{1}{2\hat{q}_x}\right)\mathcal{P}(x)\text{d}x\text{d}\xi
\label{Equ:SISO}\\
=&1,
\end{align}
where $(a)$ holds by changing of variable $\xi'=\sqrt{\frac{2\hat{q}_x}{\eta}}\xi$. Furthermore, we calculate
\begin{align}
\chi_x
&=\lim_{\tau \rightarrow 0}\mathbb{E}_{\bs{x}}\left\{(x_0)^2\exp \left(\bs{x}^{\text{T}}\hat{\bs{Q}}_X\bs{x}\right)\right\}=\int x^2\mathcal{P}(x)\text{d}x,\\
q_x
&=\lim_{\tau\rightarrow 0}\mathbb{E}_{\bs{x}}\left\{x_0x_1\exp \left(\bs{x}^{\text{T}}\hat{\bs{Q}}_X\bs{x}\right)\right\}=\int \frac{\left[\int x\mathcal{P}(x)\mathcal{N}(x|\zeta,\frac{1}{2\hat{q}_x})\text{d}x\right]^2}
{\int \mathcal{P}(x)\mathcal{N}(x|\zeta,\frac{1}{2\hat{q}_x})\text{d}x}\text{d}\zeta.
\end{align}
Indeed, the following equivalent single-input and single-output (SISO) system can be directly established from (\ref{Equ:SISO})
\begin{align}
y_x=x+w_x \ \text{with} \ w_x\sim \mathcal{N}(0,\frac{1}{2\hat{q}_x}).
\end{align}
Accordingly, the MSE of $x$ is expressed as a combination of parameters ($\chi_x,q_x$) i.e.,
\begin{align}
\textsf{mse}_x=\chi_x-q_x.
\end{align}

Similar to (\ref{Equ:SD8}), solving equations (\ref{Equ:SD2}) and (\ref{Equ:SD6}) yields
\begin{align}
y_{h}&=h+w_{h} \ \text{with}\ w_{h}\sim \mathcal{N}(0,\frac{1}{2\hat{q}_h}),\\
y_{c}&=c+w_{c} \ \text{with}\ w_{c}\sim \mathcal{N}(0,\frac{1}{2\hat{q}_c}),
\end{align}
The MSEs of MMSE estimators of $h$ and $c$ are given by
\begin{align}
\textsf{mse}_{h}=\chi_h-q_h,\\
\textsf{mse}_{c}=\chi_c-q_c,
\end{align}
where $\chi_h=\int h^2\mathcal{P}(h)\text{d}h$, $\chi_c=\int c^2\mathcal{P}(c)\text{d}c$,  and
\begin{align}
q_h=\int \frac{\left[\int h\mathcal{P}(h)\mathcal{N}(h|\zeta,\frac{1}{2\hat{q}_h})\text{d}h\right]^2}
{\int \mathcal{P}(h)\mathcal{N}(h|\zeta,\frac{1}{2\hat{q}_h})\text{d}h}\text{d}\zeta,\\
q_c=\int \frac{\left[\int c\mathcal{P}(c)\mathcal{N}(c|\zeta,\frac{1}{2\hat{q}_c})\text{d}c\right]^2}
{\int \mathcal{P}(c)\mathcal{N}(c|\zeta,\frac{1}{2\hat{q}_c})\text{d}c}\text{d}\zeta.
\end{align}

In summary, the parameters ($\chi_c,q_c,\hat{q}_c,\chi_s, q_s, \hat{q}_s, \chi_h,q_h$, $\hat{q}_h, \chi_x, q_x, \hat{q}_x$ ) constitute the fixed point of MMSE estimator in two-layer model case. It is easy to validate that the fixed points of the exact MMSE estimator by replica method match perfectly with the SE equations of ML-BiGAMP ($L=2$) depicted in Algorithm 2.

\subsection{Extension to Multi-Layer}
To extend the results of two-layer to multi-layer case, the procedures include: Appendix \ref{Sec:Firstlayer} (begin at last layer) $\rightarrow$ Appendix \ref{Sec:Previouslayer} (move to previous layer) $\rightarrow \cdots \rightarrow $ Appendix \ref{Sec:Previouslayer} (until the first layer) $\rightarrow $ Appendix \ref{Sec:RS} (replica symmetry solution). After some algebras, the fixed point equations of MMSE in multi-layer regime derived by replica method are summarized as (39)-(48). It can be found that the ML-BiGAMP'SE in Algorithm 2 matches perfectly the fixed point equations of MMSE in multi-layer regime under the setting
\begin{align}
\Sigma^{(x,\ell)}=\frac{1}{2\hat{q}_x^{(\ell)}}, \quad \Sigma^{(h,\ell)}=\frac{1}{2\hat{q}_h^{(\ell)}}.
\end{align}
This consistency indicates that the Bayes-optimal error can be achieved by the efficient ML-BiGAMP algorithm.

\section{}
\label{Appendix:D}
Here we explain the reason of ignoring the item $\sum_{m=1}^Mv_{mn}^{(h,\ell)}(t)(|\hat{s}_{mk}^{(\ell)}(t)|^2-v^{(s,\ell)}_{mk}(t))$. This item can be written as
\begin{align}
\sum_{m=1}^Mv_{mn}^{(h,\ell)}(t)(|\hat{s}_{mk}^{(\ell)}(t)|^2-v^{(s,\ell)}_{mk}(t))
=&\sum_{m=1}^Mv_{mn}^{(h,\ell)}(t)\left(\frac{(\tilde{z}_{mk}^{(\ell)}-Z_{mk}^{(\ell)})^2}{(V_{mk}^{(\ell)}(t))^2}-\frac{V_{mk}^{(\ell)}(t)-\tilde{v}_{mk}^{(\ell)}(t)}{(V_{mk}^{(\ell)}(t))^2}\right)\\
=&\sum_{m=1}^M\frac{v_{mn}^{(h,\ell)}(t)}{V_{mk}^{(\ell)}(t)}\left(\mathbb{E}\left\{\frac{\left(z_{mk}^{(\ell)}-Z_{mk}^{(\ell)}(t)\right)^2}{V_{mk}^{(\ell)}(t)}\right\}-1\right),
\label{App:A1}
\end{align}
We now show the reason of $\mathbb{E}\left\{\frac{\left(z_{mk}^{(\ell)}-Z_{mk}^{(\ell)}(t)\right)^2}{V_{mk}^{(\ell)}(t)}\right\}=1$ with expectation over $\zeta_{mk}^{(\ell)}(t)$ in (\ref{Pz_A}) for  $\ell<L$ and in (\ref{Pz_B}) for $\ell=L$. In large system limits, we assume
\begin{align}
v^{(h,\ell)}_{mn}(t)&\approx \frac{1}{N_{\ell+1}N_{\ell}}\sum_{m=1}^{N_{\ell+1}}\sum_{n=1}^{N_{\ell}}v^{(h,\ell)}_{mn}(t)=\overline{v^{(h,\ell)}}(t),\\
v^{(x,\ell)}_{nk}(t)&\approx \frac{1}{N_{\ell}K}\sum_{n=1}^{N_{\ell}}\sum_{k=1}^Kv^{(x,\ell)}_{nk}(t)=\overline{v^{(x,\ell)}}(t).
\end{align}
We define $V^{(\ell)}(t)$ from $V^{(\ell)}_{mk}(t)$ by using $\overline{v^{(h,\ell)}}(t)$ and $\overline{v^{(x,\ell)}}(t)$ to replace $v^{(h,\ell)}_{mn}(t)$ and $v^{(x,\ell)}_{nk}(t)$. By empirical convergence of RVs, we ignore subscripts and iteration times and approximate  (\ref{App:A1}) as
\begin{align}
(\ref{App:A1})
&\approx  \frac{M\overline{v^{(h,\ell)}}}{V^{(\ell)}}\frac{1}{KM}\sum_{m=1}^M\sum_{k=1}^K\left(\mathbb{E}\left\{\frac{\left(z_{mk}^{(\ell)}-Z_{mk}^{(\ell)}\right)^2}{V_{mk}^{(\ell)}}\right\}-1\right)\\
&\approx \frac{M\overline{v^{(h,\ell)}}}{V^{(\ell)}} \mathbb{E}\left\{\mathbb{E}\left\{\frac{\left(z^{(\ell)}-Z^{(\ell)}\right)^2}{V^{(\ell)}}\right\}-1\right\},
\label{App:A2}
\end{align}
where the inner expectation is taken over $p(z^{(\ell)}|y)$ while the outer expectation is  over $\mathcal{P}(Z^{(\ell)},R^{(x,\ell+1)})$ in $\ell<L$ or  $\mathcal{P}(Z^{(\ell)},y)$ in $\ell=L$ given by
\begin{align}
\nonumber
\mathcal{P}(Z^{(\ell)},R^{(x,\ell+1)})&=\mathcal{P}(Z^{(\ell)})\int \mathcal{P}(x^{(\ell)}|z^{(\ell)}) \mathcal{N}(z^{(\ell)}|Z^{(\ell)},V^{(\ell)})\\
&\qquad \times \mathcal{N}(x^{(\ell+1)}|R^{(x,\ell+1)},V^{(x,\ell+1)})\text{d}x^{(\ell+1)}\text{d}z^{(\ell)},\\
\mathcal{P}(Z^{(L)},y)&=\mathcal{P}(Z^{(L)})\int \mathcal{P}(y|z^{(L)})\mathcal{N}(z^{(L)}|Z^{(L)},V^{(L)})\text{d}z^{(L)},
\end{align}
with $\mathcal{P}(Z^{(\ell)})=\mathcal{N}(Z^{(\ell)}|0,\chi_z^{(\ell)}-V^{(\ell)})$, $\chi_z^{(\ell)}=N\chi_h\chi_x$, $\chi_h^{(\ell)}=\int |h^{(\ell)}|^2\mathcal{P}(h^{(\ell)})\text{d}h^{(\ell)}$, and $\chi_x^{(\ell)}=\int |x^{(\ell)}|^2\mathcal{P}(x^{(\ell)}|z^{(\ell)})\mathcal{N}(z^{(\ell-1)}|0,\chi_z^{(\ell-1)})\text{d}z^{(\ell-1)}\text{d}x^{(\ell)}$.

 From (\ref{App:A2}), for $\ell<L$ or $\ell=L$, we  have
\begin{align}
\mathbb{E}\left\{\mathbb{E}\left\{\frac{\left(z^{(\ell)}-Z^{(\ell)}\right)^2}{V^{(\ell)}}\right\}\right\}-1=0.
\end{align}
Similarly, the term $\sum_{k=1}^{K}v_{nk}^{(x,\ell)}(t)(|\hat{s}_{mk}^{(\ell)}(t)|^2-v^{(s,\ell)}_{mk}(t))$ can also be neglected.

\section{Rate function of $\bs{Q}_C$ and $\bs{Q}_S$}
\label{Appendix:E}
The auxiliary matrices $\bs{Q}_C=\{Q_C^{ab}, \forall a,b\}$ and $\bs{Q}_S=\{Q_S^{ab},\forall a,b\}$ are defined as below
\begin{align}
1&=\int \prod_{p=1}^{P}\prod_{0\leq a\leq b}^{\tau}\delta\left(MQ_C^{ab}-\sum_{m=1}^Mc_{pm}^{(a)}c_{pm}^{(b)}\right)\text{d}Q_C^{ab},\\
1&=\int \prod_{k=1}^{K}\prod_{0\leq a\leq b}^{\tau}\delta\left(MQ_C^{ab}-\sum_{m=1}^Ms_{mk}^{(a)}s_{mk}^{(b)}\right)\text{d}Q_C^{ab},
\end{align}
with probability measure
\begin{align}
\!\!\!\!\!
\mathcal{P}(\bs{Q}_C)&=\mathbb{E}_{\boldsymbol{\mathcal{C}}}\left\{\prod_{p=1}^{P}\prod_{0\leq a\leq b}^{\tau}\delta\left(MQ_C^{ab}-\sum_{m=1}^Mc_{pm}^{(a)}c_{pm}^{(b)}\right)\right\},\\
\!\!\!\!\!
\mathcal{P}(\bs{Q}_S)&=\mathbb{E}_{\boldsymbol{\mathcal{S}}}\left\{\prod_{k=1}^{K}\prod_{0\leq a\leq b}^{\tau}\delta\left(MQ_S^{ab}-\sum_{m=1}^Ms_{mk}^{(a)}s_{mk}^{(b)}\right)\right\}.
\end{align}
To present their rate functions, we firstly introduce the Fourier representation of Dirac function.

\subsection{Fourier representation of Dirac function}
\label{Fourier}

By the fact
\begin{align}
\delta(x)=\frac{1}{2\pi}e^{\mathbb{J}\tilde{x}x}\text{d}\tilde{x}=\delta(\mathbb{J} x)=\frac{1}{2\pi}\int \exp \left({-\tilde{x}x}\right)\text{d}\tilde{x},
\end{align}
we have
\begin{align}
\delta \left(MQ_C^{ab}-\sum_{m=1}^M c_{pm}^{(a)}c_{pm}^{(b)}\right)=\frac{1}{2\pi}\int \exp \left[{-\tilde{Q}_C^{ab}\left(MQ_C^{ab}-\sum_{m=1}^Mc_{pm}^{(a)}c_{pm}^{(b)}\right)}\right]\text{d}\tilde{Q}_C^{ab},
\end{align}
and further
\begin{align}
\nonumber
&\mathbb{E}_{\boldsymbol{\mathcal{C}}}\left\{\prod_{p=1}^P \prod_{a\leq b}\delta\left(MQ_C^{ab}-\sum_{m=1}^Mc_{pm}^{(a)}c_{pm}^{(b)}\right)\right\}\\
&= \frac{1}{(2\pi)^{\frac{P(\tau+2)(\tau+1)}{2}}} \mathbb{E}_{\boldsymbol{\mathcal{C}}}\left\{\int\exp \left(-PM\sum_{a\leq b}\tilde{Q}_{C}^{ab}Q_C^{ab}\right)\exp \left(\sum_{p=1}^P\sum_{a\leq b}\sum_{m=1}^M\tilde{Q}_C^{ab}c_{pm}^{(a)}c_{pm}^{(b)}\right)\text{d}\tilde{\bs{Q}}_C\right\}.
\label{APEN_1}
\end{align}
Note that the summation is over $a\leq b$ because $Q_C^{ab}=Q_C^{ba}$. Finally, we make the change of variables by
\begin{align}
&\forall a, \quad\ \ \ \hat{Q}_C^{aa}=\tilde{Q}_C^{aa},\\
&\forall a\ne b, \ \hat{Q}_C^{ab}=2\tilde{Q}_C^{ab},
\end{align}
which allows us to write the sums in (\ref{APEN_1}) more compactly
\begin{align}
\sum_{a\leq b}^{\tau}\tilde{Q}_{C}^{ab}Q_C^{ab}&=\text{tr}\left(\bs{Q}_C\hat{\bs{Q}}_C\right),\\
\sum_{a\leq b}^{\tau}\tilde{Q}_C^{ab}c_{pm}^{(a)}c_{pm}^{(b)}&=\bs{c}_{pm}^{\text{T}}\hat{\bs{Q}}_C\bs{c}_{pm},
\end{align}
where $\bs{c}_{pm}\overset{\triangle}{=}\{c_{pm}^{(a)}, \forall a \}$.

\subsection{Rate function $\mathcal{R}^{(\tau)}(\bs{Q}_C)$ and $\mathcal{R}^{(\tau)}(\bs{Q}_S)$}
\label{Rate}
Using Fourier transform representation of Dirac function above, we rewrite $\mathcal{P}(\bs{Q}_C)$ as
\begin{align}
\mathcal{P}(\bs{Q}_C)=\text{const}\cdot \int  \mathbb{E}_{\boldsymbol{\mathcal{C}}}\left\{\exp \left(\sum_{p=1}^P\sum_{m=1}^M\bs{c}_{pm}^{\text{T}}\hat{\bs{Q}}_C\bs{c}_{pm}\right)\right\}\exp(-PM\text{tr}(\hat{\bs{Q}}_C\bs{Q}_C))\text{d}\hat{\bs{Q}}_C,
\end{align}
where ``const'' denotes a constant. We then evaluate
\begin{align}
\mathcal{R}^{(\tau)}(\bs{Q}_C)
&=-\frac{1}{PM}\log \mathcal{P}(\bs{Q}_C)\\
&=-\frac{1}{PM}\log \int  \mathbb{E}_{\boldsymbol{\mathcal{C}}}\left\{\exp \left(\sum_{p=1}^P\sum_{m=1}^M\bs{c}_{pm}^{\text{T}}\hat{\bs{Q}}_C\bs{c}_{pm}\right)\right\}\exp(-PM\text{tr}(\hat{\bs{Q}}_C\bs{Q}_C))\text{d}\hat{\bs{Q}}_C+o\\
&=\sup_{\hat{\bs{Q}}_C}\left\{\text{tr}(\hat{\bs{Q}}_C\bs{Q}_C)-\frac{1}{PM}\log \mathbb{E}_{\boldsymbol{\mathcal{C}}}\left\{\exp \left(\sum_{p=1}^P\sum_{m=1}^M\bs{c}_{pm}^{\text{T}}\hat{\bs{Q}}_C\bs{c}_{pm}\right)\right\} \right\}.
\end{align}

By the fact
\begin{align}
\nonumber
&\frac{1}{PM}\log \mathbb{E}_{\boldsymbol{\mathcal{C}}}\left\{\exp \left(\sum_{p=1}^P\sum_{m=1}^M\bs{c}_{pm}^{\text{T}}\hat{\bs{Q}}_C\bs{c}_{pm}\right)\right\}\\
=&\frac{1}{PM}\log \left(\prod_{p=1}^P\prod_{m=1}^M\mathbb{E}_{\bs{c}_{pm}}\left\{\exp \left(\bs{c}_{pm}^{\text{T}}\hat{\bs{Q}}_C\bs{c}_{pm}\right)\right\}\right)\\
=&\log \mathbb{E}_{\bs{c}}\left\{\exp \left(\bs{c}^{\text{T}}\hat{\bs{Q}}_C\bs{c}\right)\right\},
\end{align}
the rate function $\mathcal{R}^{(\tau)}(\bs{Q}_{C})$ can also be written as
\begin{align}
\mathcal{R}^{(\tau)}(\bs{Q}_C)&=\sup_{\hat{\bs{Q}}_C}\left\{\text{tr}(\hat{\bs{Q}}_C\bs{Q}_C)-\log \mathbb{E}_{\boldsymbol{\mathcal{C}}}\left\{\exp \left(\bs{c}^{\text{T}}\hat{\bs{Q}}_C\bs{c}\right)\right\} \right\}.
\end{align}

Similar to calculating $\mathcal{R}^{(\tau)}(\bs{Q}_C)$, the following can be obtained
\begin{align}
\mathcal{R}^{(\tau)}(\bs{Q}_{S})&=-\frac{1}{MK}\log \mathcal{P}(\bs{Q}_S)\\
&=\sup_{\hat{\bs{Q}}_S}\left\{\text{tr}(\hat{\bs{Q}}_S\bs{Q}_S)-\frac{1}{MK}\log \mathbb{E}_{\boldsymbol{\mathcal{S}}}\left\{\exp \left(\sum_{m=1}^M\sum_{k=1}^K\bs{s}_{mk}^{\text{T}}\hat{\bs{Q}}_S\bs{s}_{mk}\right)\right\}\right\}.
\end{align}

\section{Calculation of parameters ($\chi_s,q_s,\chi_u,q_u$)}
\label{Appendix:F}
With decoupling operations (\ref{Equ:Decoupling1})-(\ref{Equ:Decoupling2}), we first calculate the denominator of $\chi_s$
\begin{align}
\nonumber
&\lim_{\tau\rightarrow 0}\int_{\bs{s}}\int_{\bs{u}}\exp \left(\bs{s}^{\text{T}}\hat{\bs{Q}}_S\bs{s}\right)\mathcal{P}(\bs{s}|\bs{u})\mathcal{N}(\bs{u}|\bs{0},N\bs{Q}_H\odot \bs{Q}_X)\text{d}\bs{u}\text{d}\bs{s}\\
\nonumber
=&\lim_{\tau\rightarrow 0}C\int_{\bs{s}}\int_{\bs{u}}\mathcal{P}(\bs{s}|\bs{u})\left[\int \sqrt{\frac{\eta}{2\pi}}\exp \left(-\frac{1}{2}E\sum_{a=0}^{\tau}(u^{(a)})^2-\frac{\eta}{2}\xi^2+\sqrt{\eta F}\xi\sum_{a=0}^{\tau}u^{(a)}\right)\text{d}\xi\right]\\
&\qquad \times \left[\int \sqrt{\frac{\varrho}{2\pi}}\exp \left(-\hat{q}_s\sum_{a=0}^{\tau}(s^{(a)})^2-\frac{\varrho}{2}\zeta^2+\sqrt{2\varrho\hat{q}_s}\zeta\sum_{a=0}^{\tau}s^{(a)}\right)\text{d}\zeta\right]\text{d}\bs{u}\text{d}\bs{s}\\
\nonumber
=&\lim_{\tau\rightarrow 0}C\int_{\zeta}\int_{\xi} \left[\int_{s}\int_u\mathcal{P}(s|u)\exp \left(-\frac{1}{2}Eu^2+\sqrt{\eta F}u\xi\right)\exp \left(-\hat{q}_ss^2+\sqrt{2\varrho\hat{q}_s}\zeta s\right)\text{d}u\text{d}s\right]^{\tau+1}\\
&\qquad \times \sqrt{\frac{\eta}{2\pi}}\exp \left(-\frac{\eta}{2}\xi^2\right)\sqrt{\frac{\varrho}{2\pi}}\exp \left(-\frac{\varrho}{2}\zeta^2\right)\text{d}\xi \text{d}\zeta\\
\nonumber
=&\lim_{\tau\rightarrow 0}C\int_{\zeta}\int_{\xi} \left[\int_s\int_u\mathcal{P}(s|u)\mathcal{N}\left(u|\frac{\sqrt{\eta F}}{E}\xi,\frac{1}{E}\right)\mathcal{N}\left(s|\sqrt{\frac{\varrho}{2\hat{q}_s}}\zeta,\frac{1}{2\hat{q}_s}\right)\text{d}u\text{d}s\right]\\
&\qquad \times \sqrt{\frac{\varrho}{2\hat{q}_s}}\sqrt{\frac{2\pi}{E-F}}\mathcal{N}\left(\xi|0,\frac{E}{\eta(E-F)}\right)\text{d}\xi\text{d}\zeta.
\label{Equ:DenS}
\end{align}
Let $\zeta\leftarrow \sqrt{\frac{\varrho}{2\hat{q}_s}}\zeta$, $\xi\leftarrow \sqrt{\frac{\eta(E-F)}{E}}\xi$, we write the equation above as
\begin{align}
(\ref{Equ:DenS})
&=\lim_{\tau\rightarrow 0}C\sqrt{\frac{2\pi}{E-F}}\int_{\zeta}\int_{\xi}\int_{s}\int_u \mathcal{P}(s|u)\mathcal{N}\left(s|\zeta,\frac{1}{2\hat{q}_s}\right)\mathcal{N}\left(u|\sqrt{\frac{F}{E(E-F)}}\xi,\frac{1}{E}\right)\text{d}u\text{d}s\text{D}\xi\text{d}\zeta\\
&=\lim_{\tau\rightarrow 0}C\sqrt{\frac{2\pi}{E-F}}\int_{\xi}\int_u \mathcal{N}\left(u|\sqrt{\frac{F}{E(E-F)}}\xi,\frac{1}{E}\right)\text{d}u\text{D}\xi\\
&\overset{(a)}{=}\lim_{\tau\rightarrow 0}C\sqrt{\frac{2\pi}{E-F}}\int_u \sqrt{\frac{E(E-F)}{F}} \mathcal{N}\left(\sqrt{\frac{E(E-F)}{F}}u|0,\frac{E-F}{F}+1\right)\text{d}u\\
&=\lim_{\tau\rightarrow 0}C\sqrt{\frac{2\pi}{E-F}}\int_u\mathcal{N}\left(u|0,\frac{1}{E-F}\right)\text{d}u\\
&=\lim_{\tau\rightarrow 0}C\sqrt{\frac{2\pi}{E-F}},
\label{Equ:S_Den}
\end{align}
where $C=(2\pi)^{-\frac{\tau+1}{2}}[\text{det}(N\bs{Q}_H\odot \bs{Q}_X)]^{-\frac{1}{2}}$ and $(a)$ holds by Gaussian reproduction property.

The numerator of $\chi_s$ is calculated by
\begin{align}
\nonumber
&\lim_{\tau\rightarrow 0}\int_{\bs{s}}\int_{\bs{u}}(s^{(0)})^2\exp \left(\bs{s}^{\text{T}}\hat{\bs{Q}}_S\bs{s}\right)\mathcal{P}(\bs{s}|\bs{u})\mathcal{N}(\bs{u}|\bs{0},N\bs{Q}_H\odot \bs{Q}_X)\text{d}\bs{u}\text{d}\bs{s}\\
\nonumber
=&\lim_{\tau\rightarrow 0}C\int_{\bs{s}}(s^{(0)})^2\int_{\bs{u}}\mathcal{P}(\bs{s}|\bs{u})\left[\int \sqrt{\frac{\eta}{2\pi}}\exp \left(-\frac{1}{2}E\sum_{a=0}^{\tau}(u^{(a)})^2-\frac{\eta}{2}\xi^2+\sqrt{\eta F}\xi\sum_{a=0}^{\tau}u^{(a)}\right)\text{d}\xi\right]\\
&\qquad \times \left[\int \sqrt{\frac{\varrho}{2\pi}}\exp \left(-\hat{q}_s\sum_{a=0}^{\tau}(s^{(a)})^2-\frac{\varrho}{2}\zeta^2+\sqrt{2\varrho\hat{q}_s}\zeta\sum_{a=0}^{\tau}s^{(a)}\right)\text{d}\zeta\right]\text{d}\bs{u}\text{d}\bs{s}\\
\nonumber
=&\lim_{\tau\rightarrow 0}C\int_{\zeta}\int_{\xi}\left[\int_s\int_us^2 \mathcal{P}(s|u)\exp \left(-\frac{1}{2}Eu^2+\sqrt{\eta F}u\xi\right)\exp \left(-\hat{q}_ss^2+\sqrt{2\varrho\hat{q}_s}\zeta s\right)\text{d}u\text{d}s \right]\\
\nonumber
&\qquad \times  \left[\int_s\int_u \mathcal{P}(s|u)\exp \left(-\frac{1}{2}Eu^2+\sqrt{\eta F}u\xi\right)\exp \left(-\hat{q}_ss^2+\sqrt{2\varrho\hat{q}_s}\zeta s\right)\text{d}u\text{d}s \right]^{\tau}\\
&\qquad  \times \sqrt{\frac{\eta}{2\pi}}\exp \left(-\frac{\eta}{2}\xi^2\right)\sqrt{\frac{\varrho}{2\pi}}\exp \left(-\frac{\varrho}{2}\zeta^2\right)\text{d}\xi \text{d}\zeta\\
=& \lim_{\tau\rightarrow 0}C\sqrt{\frac{2\pi}{E-F}}\int_{\zeta}\int_{\xi}\int_{s}\int_u s^2 \mathcal{P}(s|u)\mathcal{N}\left(s|\zeta,\frac{1}{2\hat{q}_s}\right)\mathcal{N}\left(u|\sqrt{\frac{F}{E(E-F)}}\xi,\frac{1}{E}\right)\text{d}u\text{d}s\text{D}\xi\text{d}\zeta\\
=&\lim_{\tau \rightarrow 0}C\sqrt{\frac{2\pi}{E-F}}\int_s\int_u s^2\mathcal{P}(s|u)\mathcal{N}\left(u|0,\frac{1}{E-F}\right) \text{d}u\text{d}s.
\label{Equ:S_Num}
\end{align}
Combining (\ref{Equ:S_Den}) and (\ref{Equ:S_Num}) yields
\begin{align}
\chi_s
&=\lim_{\tau\rightarrow 0}\int_s\int_us^2\mathcal{P}(s|u)\mathcal{N}\left(u|0,\frac{1}{E-F}\right)\text{d}u\text{d}s\\
&=\int_s\int_us^2\mathcal{P}(s|u)\mathcal{N}(u|0,N\chi_x\chi_h)\text{d}u\text{d}s.
\end{align}

The number of $q_s$ is given
\begin{align}
\nonumber
&\lim_{\tau\rightarrow 0}\int_{\bs{s}}\int_{\bs{u}}s^{(0)}s^{(1)}\exp \left(\bs{s}^{\text{T}}\hat{\bs{Q}}_S\bs{s}\right)\mathcal{P}(\bs{s}|\bs{u})\mathcal{N}(\bs{u}|\bs{0},N\bs{Q}_H\odot \bs{Q}_X)\text{d}\bs{u}\text{d}\bs{s}\\
\nonumber
=&\lim_{\tau\rightarrow 0}C\int_{\bs{s}}s^{(0)}s^{(1)}\int_{\bs{u}}\mathcal{P}(\bs{s}|\bs{u})\left[\int \sqrt{\frac{\eta}{2\pi}}\exp \left(-\frac{1}{2}E\sum_{a=0}^{\tau}(u^{(a)})^2-\frac{\eta}{2}\xi^2+\sqrt{\eta F}\xi\sum_{a=0}^{\tau}u^{(a)}\right)\text{d}\xi\right]\\
&\qquad \times \left[\int \sqrt{\frac{\varrho}{2\pi}}\exp \left(-\hat{q}_s\sum_{a=0}^{\tau}(s^{(a)})^2-\frac{\varrho}{2}\zeta^2+\sqrt{2\varrho\hat{q}_s}\zeta\sum_{a=0}^{\tau}s^{(a)}\right)\text{d}\zeta\right]\text{d}\bs{u}\text{d}\bs{s}\\
\nonumber
=&\lim_{\tau\rightarrow 0}C\int_{\zeta}\int_{\xi}\left[\int_s\int_u s \mathcal{P}(s|u)\exp \left(-\frac{1}{2}Eu^2+\sqrt{\eta F}u\xi\right)\exp \left(-\hat{q}_ss^2+\sqrt{2\varrho\hat{q}_s}\zeta s\right)\text{d}u\text{d}s \right]^2\\
\nonumber
&\qquad \times  \left[\int_s\int_u \mathcal{P}(s|u)\exp \left(-\frac{1}{2}Eu^2+\sqrt{\eta F}u\xi\right)\exp \left(-\hat{q}_ss^2+\sqrt{2\varrho\hat{q}_s}\zeta s\right)\text{d}u\text{d}s \right]^{\tau-1}\\
&\qquad  \times \sqrt{\frac{\eta}{2\pi}}\exp \left(-\frac{\eta}{2}\xi^2\right)\sqrt{\frac{\varrho}{2\pi}}\exp \left(-\frac{\varrho}{2}\zeta^2\right)\text{d}\xi \text{d}\zeta\\
\nonumber
=&\lim_{\tau\rightarrow 0}C\int_{\zeta}\int_{\xi}\frac{\left[\int_s\int_us \mathcal{P}(s|u)\mathcal{N}\left(u|\frac{\sqrt{\eta F}}{E}\xi,\frac{1}{E}\right)\mathcal{N}\left(s|\sqrt{\frac{\varrho}{2\hat{q}_s}}\zeta,\frac{1}{2\hat{q}_s}\right)\text{d}u\text{d}s\right]^2}
{\int_s\int_u \mathcal{P}(s|u)\mathcal{N}\left(u|\frac{\sqrt{\eta F}}{E}\xi,\frac{1}{E}\right)\mathcal{N}\left(s|\sqrt{\frac{\varrho}{2\hat{q}_s}}\zeta,\frac{1}{2\hat{q}_s}\right)\text{d}u\text{d}s}\\
&\qquad \times \sqrt{\frac{\varrho}{2\hat{q}_s}}\sqrt{\frac{2\pi}{E-F}}\mathcal{N}\left(\xi|0,\frac{E}{\eta(E-F)}\right)\text{d}\xi\text{d}\zeta.
\label{Equ:Numqs}
\end{align}
Let $\zeta\leftarrow \sqrt{\frac{\varrho}{2\hat{q}_s}}\zeta$, $\xi\leftarrow \sqrt{\frac{\eta(E-F)}{E}}\xi$, we have
\begin{align}
(\ref{Equ:Numqs})=\lim_{\tau\rightarrow 0}C\sqrt{\frac{2\pi}{E-F}}\int_{\zeta}\int_{\xi}\frac{\left[\int_s\int_us \mathcal{P}(s|u)\mathcal{N}\left(u|\sqrt{\frac{F}{E(E-F)}}\xi,\frac{1}{E}\right)\mathcal{N}\left(s|\zeta,\frac{1}{2\hat{q}_s}\right)\text{d}u\text{d}s\right]^2}
{\int_s\int_u \mathcal{P}(s|u)\mathcal{N}\left(u|\sqrt{\frac{F}{E(E-F)}}\xi,\frac{1}{E}\right)\mathcal{N}\left(s|,\frac{1}{2\hat{q}_s}\right)\text{d}u\text{d}s}\text{D}\xi\text{d}\zeta.
\label{Equ:qs_Num}
\end{align}
Combining (\ref{Equ:qs_Num}) and (\ref{Equ:S_Den}) gets
\begin{align}
q_s
&=\lim_{\tau\rightarrow 0}\int_{\zeta}\int_{\xi}\frac{\left[\int_s\int_us \mathcal{P}(s|u)\mathcal{N}\left(u|\sqrt{\frac{F}{E(E-F)}}\xi,\frac{1}{E}\right)\mathcal{N}\left(s|\zeta,\frac{1}{2\hat{q}_s}\right)\text{d}u\text{d}s\right]^2}
{\int_s\int_u \mathcal{P}(s|u)\mathcal{N}\left(u|\sqrt{\frac{F}{E(E-F)}}\xi,\frac{1}{E}\right)\mathcal{N}\left(s|,\frac{1}{2\hat{q}_s}\right)\text{d}u\text{d}s}\text{D}\xi\text{d}\zeta\\
&=\int_{\zeta}\int_{\xi}\frac{\left[\int_s\int_u s\mathcal{P}(s|u)\mathcal{N}\left(u|\sqrt{Nq_xq_h}\xi,N(\chi_h\chi_x-q_hq_x)\right)\mathcal{N}(s|\zeta,\frac{1}{2\hat{q}_s})\text{d}u\text{d}s\right]^2}
{\int_s\int_u s\mathcal{P}(s|u)\mathcal{N}\left(u|\sqrt{Nq_xq_h}\xi,N(\chi_h\chi_x-q_hq_x)\right)\mathcal{N}(s|\zeta,\frac{1}{2\hat{q}_s})\text{d}u\text{d}s}
\text{D}\xi\text{d}\zeta\\
&=\int_{\zeta}\int_{\xi}\frac{\left[\int s\mathcal{N}_{s|u}\left(\sqrt{Nq_xq_h}\xi,N(\chi_h\chi_x-q_hq_x),\zeta,\frac{1}{2\hat{q}_s}\right)\text{d}u\text{d}s\right]^2}{\int \mathcal{N}_{s|u}\left(\sqrt{Nq_xq_h}\xi,N(\chi_h\chi_x-q_hq_x),\zeta,\frac{1}{2\hat{q}_s}\right)\text{d}u\text{d}s}\text{D}\xi\text{d}\zeta,
\end{align}
where $\mathcal{N}_{s|u}(\cdot)=\mathcal{P}(s|u)\mathcal{N}(u|\cdot)\mathcal{N}(s|\cdot)$.

We then move to calculating $\chi_u$ and $q_u$. The number of $\chi_u$ is given by
\begin{align}
\nonumber
&\lim_{\tau\rightarrow 0}\int_{\bs{s}}\int_{\bs{u}}(u^{(0)})^2\exp \left(\bs{s}^{\text{T}}\hat{\bs{Q}}_S\bs{s}\right)\mathcal{P}(\bs{s}|\bs{u})\mathcal{N}(\bs{u}|\bs{0},N\bs{Q}_H\odot \bs{Q}_X)\text{d}\bs{u}\text{d}\bs{s}\\
=& \lim_{\tau\rightarrow 0}C\sqrt{\frac{2\pi}{E-F}}\int_{\zeta}\int_{\xi}\int_{s}\int_u u^2 \mathcal{P}(s|u)\mathcal{N}\left(s|\zeta,\frac{1}{2\hat{q}_s}\right)\mathcal{N}\left(u|\sqrt{\frac{F}{E(E-F)}}\xi,\frac{1}{E}\right)\text{d}u\text{d}s\text{D}\xi\text{d}\zeta\\
=&\lim_{\tau \rightarrow 0}C\sqrt{\frac{2\pi}{E-F}}\int_{\xi}\int_u u^2\mathcal{N}\left(u|\sqrt{\frac{F}{E(E-F)}}\xi,\frac{1}{E}\right)\text{d}u\text{D}\xi\\
=&\lim_{\tau\rightarrow 0}C\sqrt{\frac{2\pi}{E-F}}\int_{\xi}\int_u u^2\mathcal{N}\left(u|0,\frac{1}{E-F}\right)\text{d}u\\
=&\lim_{\tau\rightarrow 0}C\sqrt{\frac{2\pi}{E-F}}\frac{1}{E-F}.
\end{align}
Then the following could be obtained
\begin{align}
\chi_u=\lim_{\tau\rightarrow 0}\frac{1}{E-F}=N\chi_h\chi_x.
\end{align}

The calculation of $q_u$ is the same as $q_s$. After some algebras, we get
\begin{align}
q_u=\int_{\zeta}\int_{\xi}\frac{\left[\int u\mathcal{N}_{s|u}\left(\sqrt{Nq_xq_h}\xi,N(\chi_h\chi_x-q_hq_x),\zeta,\frac{1}{2\hat{q}_s}\right)\text{d}s\text{d}u\right]^2}{\int \mathcal{N}_{s|u}\left(\sqrt{Nq_xq_h}\xi,N(\chi_h\chi_x-q_hq_x),\zeta,\frac{1}{2\hat{q}_s}\right)\text{d}s\text{d}u}\text{D}\xi\text{d}\zeta.
\end{align}

\section{Proof for partial derivation of Gaussian}
\label{Appendix:G}
Given a Gaussian distribution
\begin{align}
\mathcal{N}\left(\mathbf{x}|\mathbf{0},\chi \mathbf{Q}_H\odot \mathbf{Q}_X\right)
=(2\pi)^{-\frac{N}{2}}\det(\chi \mathbf{Q}_H\odot \mathbf{Q}_X)^{-\frac{1}{2}}\exp \left[-\frac{1}{2}\mathbf{x}^{\text{T}}(\chi \mathbf{Q}_H\odot \mathbf{Q}_X)^{-1}\mathbf{x}\right],
\end{align}
where `$\odot$' denotes the element-wise multiply, its partial derivation w.r.t. $\mathbf{Q}_H$ denotes
\begin{align}
\nonumber
&\frac{\partial \mathcal{N}(\mathbf{x}|\mathbf{0},\chi \mathbf{Q}_H\odot \mathbf{Q}_X)}{\partial \mathbf{Q}_H}\\
\nonumber
&=(2\pi)^{-\frac{N}{2}}\frac{\partial \det(\chi \mathbf{Q}_H\odot \mathbf{Q}_X)^{-\frac{1}{2}}}{\partial \mathbf{Q}_H}\exp \left[-\frac{1}{2}\mathbf{x}^{\text{T}}(\chi \mathbf{Q}_H\odot \mathbf{Q}_X)^{-1}\mathbf{x}\right]\\
&\quad +(2\pi)^{-\frac{N}{2}}\det(\chi \mathbf{Q}_H\odot \mathbf{Q}_X)^{-\frac{1}{2}}\frac{\partial }{\partial \mathbf{Q}_H}\exp \left[-\frac{1}{2}\mathbf{x}^{\text{T}}(\chi \mathbf{Q}_H\odot \mathbf{Q}_X)^{-1}\mathbf{x}\right],
\end{align}
The partial derivation in the equation above are as follows
\begin{align}
\frac{\partial \det(\chi \mathbf{Q}_H\odot \mathbf{Q}_X)^{-\frac{1}{2}}}{\partial \mathbf{Q}_H}
&=-\frac{1}{2}\det(\chi\mathbf{Q}_H\odot \mathbf{Q}_X)^{-\frac{1}{2}}(\mathbf{Q}_H\odot \mathbf{Q}_X)^{-1}\odot \mathbf{Q}_X,
\end{align}
and
\begin{align}
\nonumber
&\frac{\partial }{\partial \mathbf{Q}_H}\exp \left[-\frac{1}{2}\mathbf{x}^{\text{T}}(\chi \mathbf{Q}_H\odot \mathbf{Q}_X)^{-1}\mathbf{x}\right]\\
&=-\frac{1}{2}\exp \left(-\frac{1}{2}\mathbf{x}^{\text{T}}(\chi \mathbf{Q}_H\odot \mathbf{Q}_X)^{-1}\mathbf{x}\right)\frac{\partial }{\partial \mathbf{Q}_H}(\mathbf{x}^{\text{T}}(\chi \mathbf{Q}_H\odot \mathbf{Q}_X)^{-1}\mathbf{x}),
\end{align}
where
\begin{align}
\frac{\partial }{\partial \mathbf{Q}_H}\mathbf{x}^{\text{T}}(\chi \mathbf{Q}_H\odot \mathbf{Q}_X)^{-1}\mathbf{x}
&=
 \left[
\begin{matrix}
\frac{\partial \mathbf{x}^{\text{T}}(\chi \mathbf{Q}_H\odot \mathbf{Q}_X)^{-1}\mathbf{x}}{\partial [Q_H]_{11}}& \cdots &\frac{\partial \mathbf{x}^{\text{T}}(\chi \mathbf{Q}_H\odot \mathbf{Q}_X)^{-1}\mathbf{x}}{\partial [Q_H]_{1N}}\\
\vdots &\ddots &\vdots\\
\frac{\partial \mathbf{x}^{\text{T}}(\chi \mathbf{Q}_H\odot \mathbf{Q}_X)^{-1}\mathbf{x}}{\partial [Q_H]_{N1}}&\cdots &\frac{\partial \mathbf{x}^{\text{T}}(\chi \mathbf{Q}_H\odot \mathbf{Q}_X)^{-1}\mathbf{x}}{\partial [Q_H]_{NN}}
\end{matrix}
\right].
\end{align}

Using the fact\footnote{$\frac{\partial g(\mathbf{U})}{\partial x}=\text{Tr}\left\{\frac{\partial g(\mathbf{U})}{\partial \mathbf{U}}\frac{\partial \mathbf{U}}{\partial x}\right\}$ and $\frac{\partial \mathbf{U}^{-1}}{\partial x}=-\mathbf{U}^{-1}\frac{\partial \mathbf{U}}{\partial x}\mathbf{U}^{-1}$, where $\bs{U}$ is square matrix with argument $x$.} we have
\begin{align}
\frac{1}{\chi}\frac{\partial \mathbf{x}^{\text{T}}( \mathbf{Q}_H\odot \mathbf{Q}_X)^{-1}\mathbf{x}}{\partial [Q_H]_{ij}}
&=\frac{1}{\chi}\text{Tr}\left\{\frac{\partial \mathbf{x}^{\text{T}}( \mathbf{Q}_H\odot \mathbf{Q}_X)^{-1}\mathbf{x}}{\partial ( \mathbf{Q}_H\odot \mathbf{Q}_X)^{-1}}\frac{\partial ( \mathbf{Q}_H\odot \mathbf{Q}_X)^{-1}}{\partial [Q_H]_{11}}\right\}\\
&=\frac{1}{\chi}\text{Tr}\left\{\mathbf{xx}^{\text{T}}\frac{\partial (\chi \mathbf{Q}_H\odot \mathbf{Q}_X)^{-1}}{\partial [Q_H]_{11}}\right\}\\
&=-\frac{1}{\chi}\text{Tr}\left\{\mathbf{xx}^{\text{T}}(\mathbf{Q}_H\odot \mathbf{Q}_X)^{-1}\frac{\partial \mathbf{Q}_H\odot \mathbf{Q}_X}{\partial [Q_H]_{ij}}( \mathbf{Q}_H\odot \mathbf{Q}_X)^{-1}\right\}\\
&=-\frac{1}{\chi}\text{Tr}\left\{\mathbf{xx}^{\text{T}}( \mathbf{Q}_H\odot \mathbf{Q}_X)^{-1} [Q_X]_{ij}\boldsymbol{e}_i\boldsymbol{e}_j^T( \mathbf{Q}_H\odot \mathbf{Q}_X)^{-1}\right\}\\
&=-\frac{1}{\chi}[Q_X]_{ij}\boldsymbol{e}_j^{\text{T}}( \mathbf{Q}_H\odot \mathbf{Q}_X)^{-1}\mathbf{xx}^T(\mathbf{Q}_H\odot \mathbf{Q}_X)^{-1}\boldsymbol{e}_i,
\end{align}
where $\boldsymbol{e}_j$ is column vector with all elements being zeros expect $j$-th element being 1. We then have
\begin{align}
\frac{\partial }{\partial \mathbf{Q}_H}\mathbf{x}^{\text{T}}(\chi \mathbf{Q}_H\odot \mathbf{Q}_X)^{-1}\mathbf{x}=-\frac{1}{\chi} \mathbf{Q}_X\odot [( \mathbf{Q}_H\odot \mathbf{Q}_X)^{-1}\mathbf{xx}^{\text{T}}( \mathbf{Q}_H\odot \mathbf{Q}_X)^{-1}].
\end{align}

As a result, we obtain
\begin{align}
\nonumber
\frac{\partial \mathcal{N}(\mathbf{x}|\mathbf{0},\chi \mathbf{Q}_H\odot \mathbf{Q}_X)}{\partial \mathbf{Q}_H}
&=-\frac{\mathcal{N}(\mathbf{x}|\mathbf{0},\chi \mathbf{Q}_H\odot \mathbf{Q}_X)}{2}\\
& \qquad \times \mathbf{Q}_X\odot \left[(\mathbf{Q}_H\odot \mathbf{Q}_X)^{-1}-\frac{1}{\chi}(\mathbf{Q}_H\odot \mathbf{Q}_X)^{-1}\mathbf{xx}^{\text{T}}(\mathbf{Q}_H\odot \mathbf{Q}_X)^{-1}\right].
\end{align}

\end{appendices}

\bibliographystyle{IEEEtran}
\bibliography{ZQY_bib}

\end{document}